\documentclass[10pt]{article}

\usepackage{arxiv}

\usepackage{url}
\usepackage{microtype}
\usepackage{graphicx}
\usepackage{subcaption}
\usepackage{booktabs}
\usepackage{setspace}
\usepackage{enumitem}
\usepackage{wrapfig}
\usepackage{algorithm}
\usepackage{algorithmic}

\usepackage{amsmath}
\usepackage{amssymb}
\usepackage{mathtools}
\usepackage{amsthm}
\usepackage[hyperfootnotes=false]{hyperref}

\usepackage{amsmath,amsfonts,bm}

\def\eqref#1{equation~\ref{#1}}

\def\1{\bm{1}}

\DeclareMathAlphabet{\mathsfit}{\encodingdefault}{\sfdefault}{m}{sl}
\SetMathAlphabet{\mathsfit}{bold}{\encodingdefault}{\sfdefault}{bx}{n}

\newcommand{\One}{\boldsymbol{\mathbf{1}}}

\newcommand{\bzeta}{\boldsymbol{\mathbb{\zeta}}}
\newcommand{\bupsilon}{\boldsymbol{\mathbb{\upsilon}}}

\newcommand{\bpi}{\boldsymbol{{\pi}}}
\newcommand{\bP}{\mathbb{P}}
\newcommand{\bR}{\mathbb{R}}
\newcommand{\bE}{\mathbb{E}}

\newcommand{\cF}{\mathcal{F}}

\newcommand{\cL}{\mathcal{L}}
\newcommand{\cG}{\mathcal{G}}

\theoremstyle{plain}
\newtheorem{theorem}{Theorem}[section]
\newtheorem{proposition}[theorem]{Proposition}
\newtheorem{lemma}[theorem]{Lemma}
\newtheorem{corollary}[theorem]{Corollary}
\theoremstyle{definition}

\newtheorem{assumption}[theorem]{Assumption}
\theoremstyle{remark}
\newtheorem{remark}[theorem]{Remark}

\usepackage[capitalize,noabbrev]{cleveref}

\title{SpeedCP: Fast Kernel-Based Conditional Conformal Prediction}

\author{%
  Yating Liu\textsuperscript{1} \\
  Department of Statistics\\
  University of Chicago\\
  \texttt{yatingliu@uchicago.edu} \\
  \And
  Yeo Jin Jung\textsuperscript{1} \\
  Department of Statistics\\
  University of Chicago\\
  \texttt{yeojinjung@uchicago.edu} \\
  \AND
  Zixuan Wu \\
  Department of Statistics\\
  University of Chicago\\
  \texttt{zixuanwu@uchicago.edu} \\
   \And
  So Won Jeong \\
  Booth Business School\\
  University of Chicago\\
\texttt{sowonjeong@chicagobooth.edu}\\
  \And
  Claire Donnat \\
  Department of Statistics\\
  University of Chicago\\
  \texttt{cdonnat@uchicago.edu}
}

\date{}

\begin{document}

\maketitle
\footnotetext[1]{Equal contribution.}

\begin{abstract}
Conformal prediction provides distribution-free prediction sets with finite-sample conditional guarantees. RKHS-based frameworks---while promising for complex covariate shifts---suffer from prohibitive computational costs. To guarantee conditional validity under such shifts while ensuring feasibility, we build upon the framework of \citet{gibbs2023conformal} by introducing a stable and efficient algorithm that computes the full solution path of the regularized RKHS conformal optimization problem, at essentially the same cost as a single kernel quantile fit. Our approach provides simultaneous hyperparameter tuning for smoothness control and data-adaptive calibration. To extend the method to high-dimensional settings, we further integrate our approach with low-rank latent embeddings that capture conditional validity in a data-driven latent space. Empirically, our method provides reliable conditional coverage across a variety of modern black-box predictors, improving the interval length of \citet{gibbs2023conformal} by 30\%, while achieving a 40-fold speedup.

\end{abstract}

\section{Introduction}\label{sec: intro}

Conformal prediction is a framework for constructing prediction sets that are valid under minimal distributional assumptions. Given a trained predictor $\hat\mu(X)$, and calibration data $(X_i, Y_i)_{i\in[n]}$ together with a test point $X_{n+1}$, all drawn i.i.d. (or more generally, exchangeable) from an unknown and arbitrary distribution $P$, conformal methods such as split conformal prediction (SplitCP) \citep{papadopoulos2002inductive} calculate conformity scores $\{S_i\}_{i\in[n]}$ on the calibration data to construct a prediction set $\hat{C}(X_{n+1})$. This procedure guarantees \textit{marginal coverage}, ensuring that the resulting set includes the true label $Y_{n+1}$ with probability at least $1-\alpha$, for any specified $\alpha\in (0, 1)$.

However, marginal coverage does not preclude significant variability in \textit{conditional coverage} on the test input $X_{n+1}$, defined as $\bP(Y_{n+1}\in \hat{C}(X_{n+1})\mid X_{n+1}=x)=1-\alpha$ for all $x$. This limitation can be particularly problematic in high-stakes applications such as drug discovery or socially sensitive decision-making, where systematic under-coverage on critical subgroups may lead to unreliable or even harmful outcomes. Unfortunately, prior works \citep{vovk2012conditional,foygel2021limits} have shown that in distribution-free settings, any interval satisfying conditional coverage must have an infinite expected length, $\hat C(X_{n+1})=\mathbb{R}$, making meaningful prediction impossible without further assumptions.

To address this issue, \citet{gibbs2023conformal} note that the conditional coverage can be equivalently reformulated as a marginal guarantee over any measurable function $f$, i.e., $\mathbb{E}[f(X_{n+1})\cdot (\One\{Y_{n+1}\in \hat C(X_{n+1})\}-(1-\alpha))]=0$.
This observation motivates them to relax the objective by restricting the requirement to a user-specified function class $\cF$:
\begin{equation}
    \label{eq: condcalib}
\begin{aligned}
     &\mathbb{E}\big[f(X_{n+1})\cdot \big(\One\{Y_{n+1}\in \hat C(X_{n+1})\}-(1-\alpha)\big)\big]\\=&0, \text{ for all }f\in \mathcal{F}.
 \end{aligned}\end{equation}
 Different choices of $\cF$ yield different notions of conditional validity. For example, taking $\cF^0=\{\eta:\eta\in \bR\}$ to be the set of all constant functions in \eqref{eq: condcalib} is equivalent to guaranteeing marginal coverage. Taking $\cF^g$ to be the set of piecewise constant functions  over a set of pre-specified (potentially overlapping) groups $\cG$, so that $\cF^g=\left\{\sum_{G\in \cG} \eta_G\One\{x\in G\}:\eta\in \bR^{|\cG|}\right\}$, yields group-conditional coverage \citep{vovk2003mondrian, jung2022batch}, i.e., $\bP(Y_{n+1}\in \hat C(X_{n+1})\mid X_{n+1}\in G)=1-\alpha$ for all $G\in \cG$.

 In this paper, we consider a more flexible class associated with a reproducing kernel Hilbert space (RKHS) that is capable of achieving coverage guarantees under \textit{complex, nonlinear covariate shifts}:
   \begin{equation}\label{eq:rkhs}
        \cF^{\mathrm{RKHS}}=\left\{f_{\psi}(\cdot)+\Phi(\cdot )^\top \eta:f_{\psi} \in \cF_{\psi}, \eta\in \bR^{d} \right\}\footnotemark{},\end{equation}\footnotetext{Given a positive definite kernel $\psi: \mathcal{X}\times \mathcal{X}\to \bR$, let $\cF_{\psi}$ denote the associated RKHS with an inner product $\langle\cdot, \cdot\rangle_{\psi}$ and a norm $\|\cdot \|_{\psi}$. Using the representer theorem \citep{kimeldorf1971some}, any function $f_{\psi}\in \cF_{\psi}$ has a finite form $f_{{\psi}}(X)=\sum_{i\in[n+1]} \upsilon_{i} \psi(X, X_i)$ for some coefficient vector $ \upsilon\in \mathbb{R}^{n+1}$. The norm has form $\|f_{\psi}\|^2_{\psi}=\langle f_{\psi},f_{\psi}\rangle_{\psi}=\sum_{i,j}  \upsilon_{i}  \upsilon_{j} \psi(X_j, X_i)$. We provide notations used in the paper in Appendix \ref{notation}.}with a given positive definite kernel $\psi: \mathcal{X}\times \mathcal{X}\to \bR$ and any covariate representation $\Phi:\mathcal{X}\to \mathbb{R}^{d}$. The linear component $\Phi(\cdot)^{\top}\eta$ enables marginal, group-conditional, or other linear adjustments, while the RKHS component $f_{\psi}(\cdot)$ controls local smoothness over complex data structures.
Notably, both $\cF^0$, $\cF^g$ are special cases of $\cF^{\mathrm{RKHS}}$.  For instance, setting $f_{\psi}=0$ and choosing
$\Phi(X) = \One\{X\in G\}$ for a group $G\in \mathcal{G}$ in \eqref{eq:rkhs} recovers group-conditional coverage.

Although RKHS function classes provide a promising surrogate for exact conditional coverage in \eqref{eq: condcalib}, their practical use remains limited. While \citet{gibbs2023conformal} provide conditional coverage guarantees for RKHS classes, the associated computational complexity remains a significant bottleneck. Unlike simpler shift classes where their method is efficient, the cost in RKHS settings is often prohibitive, limiting its scalability in large-scale applications.

To construct prediction sets, \citet{gibbs2023conformal} fit an RKHS quantile regression on the $n$ calibration points $(X_i, S_i)_{i \in [n]}$, augmented with the test point $(X_{n+1}, S)$, where $S$ is an imputed score
{for the unknown conformity score $S_{n+1}$}. The imputation of $S$ is carried out via a binary search, with each candidate value requiring a fresh RKHS regression on the $n+1$ points. To mitigate this computational burden, the authors fix the kernel bandwidth $\gamma$ and restrict hyperparameter selection to cross-validate over a pre-specified grid for the regularization parameter $\lambda$. While they demonstrate that $(\lambda, \gamma)$ do not compromise marginal coverage, these hyperparameters crucially shape the smoothness of the regression fit and thus the tightness of the resulting prediction sets.

The primary objective of this paper is to improve upon the framework of \citet{gibbs2023conformal}, enabling conditional validity within the RKHS function classes at a significantly reduced computational cost. We maintain coverage guarantees under complex covariate shifts by formulating the problem as a regularized RKHS quantile regression, which recovers score cutoffs for constructing prediction sets.

To address the above limitations, we introduce two novel \emph{$(\lambda, S)\text{-Path}$} algorithms. Inspired by the algorithmic part of the RKHS quantile regression setting in \citet{li2007quantile}, our method builds solution paths of regression parameters that are piecewise-linear in either the smoothness parameter $\lambda$ (the $\lambda\text{-Path}$) or in the candidate test score $S$ (the $S\text{-Path}$). Rather than evaluating a fixed grid, the algorithm decides the next $\lambda$ or $S$ by updating these parameters only when  {the elbow set---the indices of observations with exactly zero residuals---undergoes a change}. At each step, the solution is derived based on the current elbow set, whose size is dramatically smaller than $n{+}1$, yielding substantial computational savings. This formulation makes conditional conformal prediction with RKHS both tractable and tunable, providing prediction sets that are not only valid but also adaptively tight.

Our second objective is to deploy our method in high-dimensional settings when $X\in \mathbb{R}^p$ with $p\gg n$. In such cases, conditional coverage on low-rank representation is often more interpretable and relevant. Using raw covariates in kernel methods is often ineffective, as distance-based similarities become less discriminative. Accordingly, we approximate each covariate vector $X$ using a $K$-dimensional latent embedding (i.e., latent mixture, principal component, or layer embedding of a predictor network model) via a low-rank map $\hat\pi : \mathbb{R}^p\to \bR^K$ with $K \ll p$. We define the kernel of the RKHS function class $\cF^{\mathrm{RKHS}}$ on this representation, resulting in improved signal-to-noise ratios and enhanced predictive performance \citep{hastie2009elements, udell2019big}. {\it This yields a different notion of conditional coverage: rather than directly guaranteeing $\mathbb{P}(Y_{n+1} \in \hat{C}(X_{n+1}) | X_{n+1})$, we wish to condition on $\mathbb{P}(Y_{n+1} \in \hat{C}(X_{n+1}) | \hat\pi(X_{n+1}))$.}

\paragraph{Contributions.} Our contributions in this work are threefold:
\begin{itemize}[noitemsep, topsep=-0.1em]
\item \textit{Methods:} We extend conditional conformal prediction \citep{gibbs2023conformal} to high-dimensional settings by conditioning on learned low-rank embeddings $\hat\pi(X)$ within an RKHS, and thus improving signal-to-noise and yielding better-calibrated prediction sets, particularly in low-density data regions.
\item \textit{Algorithm:}
{We exploit the affine relationship between the imputed conformity score $S$ and the RKHS coefficients, and leverage this structure to design a fast, stable solution-path algorithm for RKHS-based conformal prediction, yielding a closed-form solution for hyperparameter selection and higher-quality prediction sets.}
\item \textit{Theory:}
{We establish finite-sample guarantees for approximate conditional coverage with respect to data-driven latent embeddings and characterize how embedding estimation error affects coverage validity in high-dimensional inference.}
\end{itemize}

We illustrate our contributions in Figure~\ref{fig:ternary}. \textbf{SpeedCP} achieves uniform 0.9 coverage across the latent space of $\hat \pi(X)$, a 2D simplex, delivering smaller prediction sets while running nearly 50 times faster than CondCP \citep{gibbs2023conformal}. Further results are discussed in Section~\ref{sec: experiments}. We provide an open-source implementation of SpeedCP at \url{https://github.com/yeojin-jung/speedcp}.

\begin{figure*}[!htbp]
    \centering
    \includegraphics[width=0.9\textwidth]{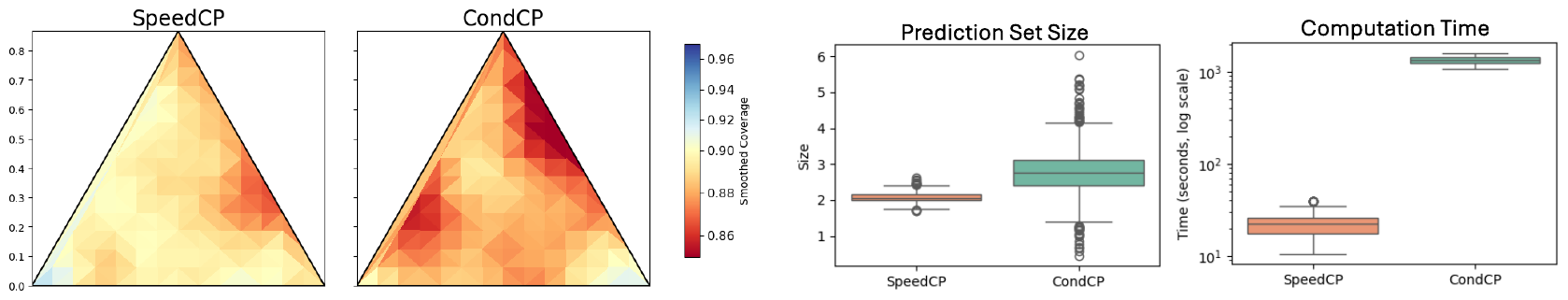}
    \caption{\textbf{SpeedCP} achieves more uniform 0.9 (pale yellow) coverage across the simplex with superior efficiency and smaller prediction set sizes. Heatmaps represent mean coverage on a fine-gridded latent space partition, aggregated over 50 independent runs.}
    \label{fig:ternary}
    \vspace{-3mm}
\end{figure*}

\paragraph{Related Literature.}
The LARS algorithm, proposed by \citet{efron2004least} for least squares regression, shows that the LASSO coefficient path is piecewise linear in the regularization parameter
$\lambda$. This inspired a broad family of path algorithms for other regularized problems, such as LASSO \citep{tibshirani1996regression, tibshirani2011solution}, generalized linear models \citep{friedman2010regularization}, and quantile regression \citep{koenker2005quantile, li2007quantile}. In our work, we build on the solution path algorithm for RKHS quantile regression developed by \citet{li2007quantile} and adapt it to our RKHS function class $\cF^{\mathrm{RKHS}}$, which has an extra linear component $\Phi(X)^{\top}\eta$. In LASSO, the active set tracks which features have nonzero coefficients as $\lambda$ decreases. The regularization parameter $\lambda$ is updated when an event (a change in the active set) occurs. In our setting, the elbow set
 plays the same role: it tracks the data points
 that are sitting exactly on the quantile boundary. What is new relative to these classic results is that we integrate this path structure into the conformal prediction framework to prove the affine structure of the solution path over the imputed score $S$, which enables efficient construction of prediction sets without the binary search used in \citet{gibbs2023conformal}.

Split conformal prediction guarantees marginal coverage through a single global cutoff \citep{papadopoulos2002inductive}, whereas some extensions seek stronger notions of calibration by conditioning on predefined groups \citep{vovk2003mondrian,jung2022batch}, local neighborhoods \citep{guan2023localized,hore2023conformal}, latent score-mixture structure \citep{zhang2024posterior}, or RKHS constraints \citep{gibbs2023conformal}; see Appendix \ref{sec: background} for a detailed comparison. SpeedCP is closest to \citet{gibbs2023conformal}, sharing the same regularized kernel quantile-regression objective, but targets a different high-dimensional regime. By conditioning on an estimated low-rank representation $\widehat{\pi}(X)$ rather than the raw covariates, SpeedCP makes conditional calibration meaningful in a learned latent space, especially when raw-space neighborhoods are sparse. The bandwidth $\gamma$ and regularization parameter $\lambda$ control the locality--stability trade-off of this latent space calibration. Algorithmically, \citet{gibbs2023conformal} exploit monotonicity of the test point dual variable to search over the imputed score, whereas SpeedCP proves a stronger piecewise-affine path structure in both $\lambda$ and $S$. This replaces repeated RKHS refitting, pre-specified $\lambda$-grid search, and binary search over $S$ with active set path tracking, making flexible conditional conformal prediction computationally feasible for high-dimensional data.

\section{Methods}\label{sec: method}

We begin by introducing preliminary notation. We partition the dataset $\{({X}_i, Y_i)\}_{i\in \mathcal{D}}$ into four disjoint subsets: $\mathcal{D}_{{\mathrm{train}}}$, $\mathcal{D}_{\mathrm{split}}$, $\mathcal{D}_{\mathrm{calib}}$, and $\mathcal{D}_{\mathrm{test}}$. A single test input is denoted as ${X}_{n+1}$, since $Y_{n+1}$ is unobserved. The training set $\mathcal{D}_{{\mathrm{train}}}$ is used to train a predictive model $\hat\mu(\cdot)$ while the calibration set $\mathcal{D}_{\mathrm{calib}}$ provides conformity scores $S_i = S({X}_i, Y_i)$ for $i \in \mathcal{D}_{\mathrm{calib}}$ (we also use $i\in [n]$ to denote calibration points since $|\mathcal{D}_{\mathrm{calib}}|=n$).

For high-dimensional covariates $\mathbf{X} \in \mathbb{R}^{n \times p}$ with $p \gg n$, we denote a low-rank embedding map by $\hat\pi : \mathcal{X} \to \mathbb{R}^K$ with $K \ll p$. Our procedure can accommodate any low-rank embedding $\hat\pi({X})$, provided that $\hat\pi(\cdot)$ is fitted symmetrically across the calibration and test set. We provide experiments on different low-rank methods in Section~\ref{sec: experiments}. When dimensionality reduction is unnecessary, the identity map $\hat\pi(X)=X$ may be used.

Our goal is to construct prediction intervals for test points $X_{n+1}$ that achieve conditional coverage defined in \eqref{eq: condcalib} within the RKHS function class $\cF^{\mathrm{RKHS}}$ (\eqref{eq:rkhs}). In the high-dimensional setting, we define the kernel on low-rank embeddings, yielding a subclass $\cF^{*} \subset\cF^{\mathrm{RKHS}}$ tailored to the latent space. The associated kernel $\psi^*(\cdot, \cdot)\equiv\psi_{\gamma, \widehat \pi}(\cdot, \cdot) $ is designed to emphasize local coverage in the latent embedding space with a given kernel bandwidth $\gamma$: \begin{align}\label{eq:kernel}
\psi^*({X}_1, {X}_2):
=\exp\left\{-\gamma\cdot d_{\pi}\left(\hat\pi({X}_1), \hat\pi({X}_2)\right)\right\},
\end{align}
where $d_{\pi}\left(\cdot, \cdot\right)$ is a distance metric between the low-dimensional embeddings (we detail this distance in Appendix \ref{sec: topic modeling}). The feature map $\Phi^*:\mathbb{R}^K\to \mathbb R^d$ is defined using the estimated embedding space generated by $\hat\pi(\cdot)$, thereby allowing linear modeling of covariate shifts within the latent representation space. The specific form of $\Phi^*$\footnote{The superscript $^*$
 is intended to emphasize that the RKHS/kernel/feature class is defined on the low-rank representation, as opposed to the raw covariate space.} depends on the application and will be specified in the theorem statements and experimental settings later. Throughout the paper, we use $f(\cdot)=f_{\psi^*}(\cdot)+\Phi^*(\cdot)^{\top}\eta\in \cF^{*}$ to denote the covariate-shift weighting of interest with a given $\gamma$ and $g(\cdot)= g_{\psi^*}(\cdot)+\Phi^*(\cdot)^{\top} \eta\in\cF^{*}$ to denote the fitted quantile estimate over the same RKHS.

\subsection{Algorithm: SpeedCP}\label{sec: algo}

In this section, we present our method for constructing conditionally calibrated prediction sets. We fit a regularized quantile regression in the RKHS class $\cF^*$. Recalling that the rank of a test point is uniformly distributed over the calibration set plus the test point, we fit using $n$ calibration covariate-score pairs $(X_i, S_i)_{i \in [n]}$ plus the test point $(X_{n+1}, S_{n+1})$. Because $S_{n+1}$ is unobserved, we impute it with an arbitrary candidate value $S$, which yields a regression function parameterized by $S$,
\begin{equation}\label{eq:quantile reg on RKHS mixture}
\begin{aligned}
\hat g_S:=&\mathrm{arg}\mathrm{min}_{g\in \cF^{*}}\frac{1}{n+1}\sum_{i\in[n]}\ell_{\alpha}(S_i-g(X_i))\\
&+\frac{1}{n+1}\ell_{\alpha}(S-g(X_{n+1}))+\frac{\lambda}{2}\|g_{\psi^*}\|_{\psi^*}^2,\footnotemark
\end{aligned}
\end{equation}\footnotetext{
The RKHS class is given by a fixed $\lambda$ such that $\cF_{\psi^*}=\{f_{\psi^*}(x)=\frac{1}{ \lambda}\sum_{i\in[n+1]}\upsilon_i \psi^*(x, X_i), \upsilon\in \mathbb{R}^{n+1}\}$.}
where $\lambda>0$ is the regularization parameter and $\ell_\alpha(z)=(1-\alpha)[z]_++\alpha[z]_{-}$ denotes the pinball loss at level $\alpha\in (0,1)$. The regularization penalty rules out the meaningless prediction set $\hat C(X_{n+1})=\mathbb{R}$ that can arise in infinite-dimensional classes. While the kernel bandwidth $\gamma$ does not appear explicitly in equation~\ref{eq:quantile reg on RKHS mixture}, it is implicitly embedded within the definition of the kernel $\psi^*$.

Accordingly, the prediction set takes the form,
\begin{align}\label{predset}
    \hat{C}^*(X_{n+1}) = \{y: S(X_{n+1},y) \leq \hat g_{S(X_{n+1},y)}(X_{n+1})\}.
\end{align}

{The RKHS class $\mathcal{F}^*$ and the corresponding quantile regressor in \eqref{eq:quantile reg on RKHS mixture} are well defined for any fixed pair $(\gamma,\lambda)$.} Our method proceeds in two stages.

{  First, we aim to select a sufficiently good pair of hyperparameters $(\gamma,\lambda)$ for the RKHS quantile regression. To maintain the validity of the downstream conformal guarantees, we use a separate data set, $\mathcal{D}_{\mathrm{split}}$, disjoint from both the calibration and test sets.

We trace the $\lambda$-path for each candidate bandwidth $\gamma$ in a prespecified grid, which provides a solution path of RKHS regression along the regularization parameter $\lambda$. We then cross-validate on the bandwidth $\gamma$ of the kernel $\psi^*$ to choose the optimal $(\hat \gamma,\hat \lambda)$ pair. This approach avoids the prohibitive cost of a full joint grid search over $(\gamma,\lambda)$. While we still iterate over the grid of $\gamma$'s, the $\lambda$-path allows for an exhaustive and efficient exploration of the regularization space without repeatedly solving the full optimization problem. We then fix the RKHS class $\mathcal{F}^*$ in \eqref{eq:quantile reg on RKHS mixture} and in all subsequent theorems with respect to this selected pair $(\hat \gamma,\hat \lambda)$.}

Second, given the optimized hyperparameters, we integrate the test observation to construct the $S$-path, which traces the maximum score cutoff $S$ that satisfies the condition in \eqref{predset}. The full procedure is detailed in Algorithm~\ref{algo:1}. We begin by outlining the setup before introducing the $(\lambda, S)$-paths.

For a given $\lambda$, the solution to \eqref{eq:quantile reg on RKHS mixture} has the following closed form:
\begin{align}
\label{eq:closed_form}
    \hat g_S(X) =\Phi^*(X)^\top\hat\eta_{S}+\frac{1}{\lambda}\sum_{i=1}^{n+1}\hat \upsilon_{ S, i}\psi^*(X,X_i),
\end{align}
where $\hat \eta_{S}$, $\hat \upsilon_{ S, i}$ are parameters when the score of the test point $S_{n+1}$ is set to $S$. For numerical stability of the algorithm, we assume the columns of $\Phi^*(X)$ are linearly independent. Plugging this into \eqref{eq:quantile reg on RKHS mixture}, the objective becomes,
\begin{align} \min_{\eta_{S}, \upsilon_{S}} &
\sum_{i=1}^{n+1}\ell_{\alpha}(S_i- \Phi^*(X_i)^\top\eta_{S}-\frac{1}{\lambda}\sum_{i^{'}=1}^{n+1} \upsilon_{S,i'}\psi^*(X_i,X_{i^{'}}))\notag\\
   &+\frac{1}{2\lambda}\sum_{i, i'=1}^{n+1}\upsilon_{ S, i}\upsilon_{S, i^{'}}\psi^*(X_i,X_{i^{'}}).\label{eq: opt obj}\end{align}
The Lagrangian formulation and the Karush–Kuhn–Tucker (KKT) conditions of \eqref{eq: opt obj} motivate us to define three index sets: the \emph{Elbow}, \emph{Left}, and \emph{Right} set,
\begin{equation}\label{eq:elbow}
    \begin{split}
        E&=\{i: S_i-g_S(X_i)=0, \upsilon_{ S,i}\in(-\alpha, 1-\alpha)\}\\
        L&=\{i: S_i-g_S(X_i)<0, \upsilon_{S,i}=-\alpha\}\\
        R&=\{i: S_i-g_S(X_i)>0, \upsilon_{S,i}=1-\alpha\}.\\
    \end{split}
\end{equation}
We observe that for the left and right sets, the kernel parameters $\upsilon_{S,i}$ are fixed to either $-\alpha$ or $1-\alpha$. Thus, we only need to solve for $\upsilon_{S,i}$'s in the elbow set, making the computation more efficient. The algorithm reduces to tracking changes in this set for different $\lambda$ or $S$ values: an \emph{event occurs when there is a change in the index sets:} 1) a point leaves the elbow or 2) when a point from the left or right set enters it.

\subsubsection{\texorpdfstring{$\lambda$-path}{lambda-path} for Smoothness Control}

To select $\lambda$, { we rely exclusively on the separate $m$ observations in $\mathcal{D}_{\mathrm{split}}$ to choose the  optimal $\lambda$ independent of calibration and test sets. The equations \ref{eq:closed_form}-\ref{eq:elbow} remain valid on this subset, so we denote the index sets as $(E(\lambda), L(\lambda), R(\lambda))$ as the sets evolve with $\lambda$. Since no imputed score $S$ is required for $S_{n+1}$, we drop $S$ from the subscripts.} We initialize $\lambda$ at the largest value for which at least two points are in the elbow, and define the step size to the next $\lambda$ as the smallest decrement that triggers an event. Importantly, {  the resulting coefficients $\{\hat{\upsilon}_{i'}(\lambda)\}_{i'\in[m]}$ and $\hat{\eta}(\lambda)$ evolve as a piecewise-linear function of $\lambda$, which we formalize in the following proposition.}

\begin{proposition}\label{prop:lambda_path} Let $\{\lambda^l\}_{l=1,2,3,\cdots}$ be the change points when an event occurs. For $\lambda^{l+1} \leq \lambda \leq \lambda^{l}$, denote $\{\hat{\upsilon}_{i'}(\lambda)\}_{i'\in [m]}$ and $\hat{\eta}(\lambda)$ as the solution of \eqref{eq: opt obj} on $\mathcal{D}_{\mathrm{split}}$, given $\lambda$. Then, $\{\hat{\upsilon}_{i'}(\lambda)\}_{i'\in [m]}$ are affine in $\lambda$ and $\hat{\eta}(\lambda)$ is affine in $1/\lambda$.
\end{proposition}

The piecewise linearity allows us to track the whole $\lambda$ solution path, not just at the change points.

To select the optimal ($\gamma, \lambda$)-pair, we fix a grid of the kernel bandwidth values $\gamma$, and run the $\lambda$-path for each fixed $\gamma$. We then perform $k$-fold cross validation on $\mathcal{D}_{\mathrm{split}}$, independent of the calibration set, to choose the combination $(\gamma, \lambda)$ that minimizes the quantile loss.

It is worth noting that our coverage guarantee, established in Section~\ref{sec: theory}, holds for any $(\gamma, \lambda)$. However, our procedure chooses the combination that reflects an appropriate level of smoothness of $\hat{g}_S$, which leads to tighter prediction sets. We provide derivation of the $\lambda$ path in Appendix~\ref{sec: sqkr}, as well as empirical results on the effect of hyperparameter tuning in Appendix~\ref{appendix:calib},  and \ref{appendix: uniform_grid}.

\subsubsection{\texorpdfstring{$S$-path}{S-path} for Constructing Prediction Sets}

We proceed to construct prediction sets using $(\hat{\gamma}, \hat{\lambda})$ selected from the $\lambda$-path.

We use the original notations of the regression parameters, $\hat{\upsilon}_{S,i}$ and $\hat{\eta}_S$, since conditions \ref{eq:quantile reg on RKHS mixture}–\ref{eq:elbow} now depend on the imputed test score $S$. Recall that the prediction set is defined as a set of $y$ such that $S(X_{n+1},y) \leq \hat g_{S(X_{n+1},y)}(X_{n+1})$. By \eqref{eq:elbow}, this is equivalent to $\hat{\upsilon}_{S(X_{n+1},y),n+1} < 1-\alpha$. Moreover, the mapping $S \mapsto \hat{\upsilon}_{S}$ is nondecreasing (which we prove in Proposition~\ref{prop: nondecreasing} in Appendix~\ref{appsec: proof}). Thus, the problem reduces to finding the largest value $S^*(X_{n+1})$ such that $\hat{\upsilon}_{S^*(X_{n+1}),n+1} < 1-\alpha$ holds.

Conceptually, the $S$-path mirrors the $\lambda$-path: it traces the evolution of the test score cutoff $S$ through a sequence of events, where events are defined identically as before. The sets in \eqref{eq:elbow} now evolve with $S$. We initialize the $S$-path with the smallest $S^1$ such that the test point enters the elbow set (i.e., $\hat{\upsilon}_{S^1,n+1}\in (-\alpha, 1-\alpha)$) and then increment $S$ to the next value at which an event occurs while the test point is still in the elbow. We iterate until the test point exits the elbow and set the final $S$ as $S^*(X_{n+1})$. Similar to the $\lambda$-path, we prove that $\hat{\upsilon}_{S,i}$'s and $\hat{\eta}_S$ evolve as an affine function  of $S$ between any two change points:

\begin{proposition}\label{prop:S_path} Let $\{S^l\}_{l=1,2,3,\cdots}$ be the change points when an event occurs. For $S^{l} \leq S \leq S^{l+1}$, denote $\{\hat{\upsilon}_{S,i}\}_{i\in[n+1]}$ and $\hat{\eta}_S$ as the solution of \eqref{eq: opt obj}. Then, $\{\hat{\upsilon}_{S,i}\}_{i\in[n+1]}$ and $\hat{\eta}_S$ are affine in $S$.
\end{proposition}

As shown in Appendix Lemmas \ref{lemma: cc} and \ref{theorem: oracle mixture}, using the threshold $S^*(X_{n+1})$ can inflate the conditional coverage. To mitigate this, we instead prefer the randomized cutoff $ S^{rand}(X_{n+1}) = \sup\{S\mid \hat \upsilon_{S,n+1}< U\}$, where $1-\alpha$ is replaced by $U\sim Unif(-\alpha,1-\alpha)$. The final prediction set is then defined as:
\begin{align}\label{pred_rand}
    \hat{C}^*_{\mathrm{rand}}(X_{n+1}) = \{y: S(X_{n+1},y) \leq S^{\mathrm{rand}}(X_{n+1})\}.
\end{align}

The affine path structure in Propositions~\ref{prop:lambda_path}-\ref{prop:S_path} builds primarily on the RKHS quantile regression path framework of \citet{li2007quantile}, which is itself analogous to the classical LASSO solution path of \citet{efron2004least}. In this framework, the elbow set $E$ plays a role similar to the active set in LASSO, identifying observations
that determine the current affine segment of the solution path. Our contribution is to extend this pathwise idea to conformal prediction, where the path must also be traced over the imputed score $S$ in order to construct prediction sets.

\paragraph{Computational Complexity.} At each iteration of the $\lambda$- and $S$-paths, we solve the inverse of
$
\begin{pmatrix}
 \mathbf\Phi^*_{E} & \frac{1}{\lambda}\mathbf\Psi^*_{EE} \\
 \mathbf{0} & \mathbf\Phi_{E}^{*\top}\\
 \end{pmatrix}.$
Here, $\mathbf \Phi^*_E \in \mathbb{R}^{|E|\times d}$ and $\mathbf \Psi^*_{EE} \in \mathbb{R}^{|E|\times |E|}$ denote feature submatrix and  kernel submatrix of the current elbow set $E$, respectively. This requires inverting a $(|E|+d)\times(|E|+d)$ matrix at each iteration. While the worst-case complexity is $O((n+d)^3)$, in practice $|E| \ll n$, making our procedure more efficient than refitting the full RKHS quantile regression at every step. In Appendix~\ref{app:elbow-set-size}, we further show through experiments that the average elbow-set size is substantially smaller than $n$ across the solution path. We detail the initialization and update functions of the $\lambda$-, and $S$-paths as well as the proofs of Proposition 1,2 in Appendix~\ref{sec: sqkr}.

\begin{algorithm}[tb]

\caption{SpeedCP} \label{algo:1}
\begin{algorithmic}
\STATE \textbf{Input:} $\mathcal{D}_{\mathrm{train}}$, $\mathcal{D}_{\mathrm{split}}$, $\mathcal{D}_{\mathrm{calib}}$, $\mathcal{D}_{\mathrm{test}}$, latent map $\hat\pi : \mathcal{X} \to \mathbb{R}^K, \ (K \ll p) $, kernel bandwidth grid $\Gamma$, miscoverage level $\alpha$
\STATE \textbf{Output:} Conditionally calibrated prediction set
\STATE 1. Train $\hat{\mu}$ on $\mathcal{D}_{\mathrm{train}}$
\STATE 2. Get calibration scores: $S_{i}=S(X_i, Y_i), i \in \mathcal{D}_{\mathrm{calib}}$
\STATE 3. Get latent embeddings: \\ \quad $\hat{\pi}_{\mathrm{calib}}$, $\hat{\pi}_{\mathrm{split}}$, $\hat{\pi}_{\mathrm{test}}$
\STATE 4. Tune hyperparameters $(\hat{\gamma},\hat{ \lambda})$ using $(\hat{\pi}_{\mathrm{split}}, S_{\mathrm{split}})$
\FOR{$\gamma \in \Gamma$}
\FOR{$j=1,\cdots k$}
\STATE \textit{A. Compute $\lambda\text{-Path}$ on $\mathcal{D}_{\text{train}_j}= (\hat{\pi}, S)_{\mathrm{split}\setminus \mathrm{fold}_j}$}
\STATE $ \ \{ \hat{\mathbf{\upsilon}}^{\gamma}(\lambda^l), \hat{\mathbf{\eta}}^{\gamma}(\lambda^l) \}_{l \geq 1} \gets \lambda\text{-Path}\big( \mathcal{D}_{\text{train}_j}; \gamma \big)$
\STATE \textit{B. Define scoring function for each $\lambda^l$}
\STATE $\ \hat{g}^l(X_i) = \Phi^*(X_i)^\top \hat{\eta}^{\gamma}(\lambda^l) + \frac{1}{\lambda^l} \mathbf{\Psi}_{j,\mathcal{D}_{\text{train}_j}}\hat{\upsilon}_{\mathcal{D}_{\text{train}_j}}^{\gamma}(\lambda^l)$
\STATE \textit{C. Evaluate pinball loss on validation fold}
\STATE $\ \mathrm{CV}_j(\gamma, \lambda^l) \gets \sum_{i \in \text{fold}_j} l_{\alpha} \big( S_i - \hat{g}^l(X_i) \big)$

\ENDFOR
\STATE $\mathrm{CV}(\gamma, \lambda^l) = \frac{1}{k}\sum_{j=1}^k\mathrm{CV}_j(\gamma,\lambda^l)\text{ for } l=1,2,\cdots$
\ENDFOR
\STATE $(\hat{\gamma}, \hat{\lambda}) \gets \arg\min_{ \{\gamma, \lambda^l\}_{\gamma \in \Gamma, l\geq 1} }\mathrm{CV}(\gamma, \lambda^l)$
\STATE 5. For each test point $X_{n+1}$, find the maximum score $S^*$ such that $S^* \leq \hat{g}_{S^*}(X_{n+1})$. Use $U \sim \mathrm{Unif}[-\alpha,1-\alpha]$ to get the corresponding score $S^{\mathrm{rand}}$ for a randomized prediction set,
\FOR{$X_{n+1} \in \mathcal{D}_{\mathrm{test}}$}
\STATE $S^{\mathrm{rand}} = S\text{-Path}(X_{n+1}, \mathcal{D}_{\mathrm{calib}};\hat{\gamma},\hat{ \lambda},U)$
\STATE  $\hat{C}^*_{\mathrm{rand}}(X_{n+1}) = \{y\in\mathcal{Y}: S(X_{n+1},y)\leq S^{\mathrm{rand}}\}$
\ENDFOR
\end{algorithmic}
\end{algorithm}

\subsection{Coverage Under Covariate Shift}\label{sec: theory}
{  In our setting, covariate shift is encoded by a tilting function $f\in\mathcal F^*$ with $\mathbb{E}_P[f(X)]>0$, which reweights the original distribution $P$ to emphasize specific regions or subpopulations of the embedding space on which we seek to condition, $
dP_f(x)
\;=\;
\frac{f(x)}{\mathbb{E}_{P}[f(X)]}\,dP(x)$. Since the solution-path formulation allows us to fit the RKHS-based quantile regression model for any pre-selected $\lambda$ and $\gamma$, we can apply Theorem~3 of \citet{gibbs2023conformal} to obtain a conditional guarantee with respect to all such tilts $f\in\mathcal F^*$ under selected $(\hat\gamma, \hat\lambda)$ (as shown in Appendix \ref{appsec: proof}). Because $\mathcal{F}^*$ is defined in terms of an estimated low-rank projection $\hat\pi(\cdot)$ rather than the unknown true embedding of the covariates, the coverage validity is robust to errors in $\hat\pi(\cdot)$. Estimation error of $\hat\pi(\cdot)$ only impacts the effectiveness of prediction set size and the deviation from conditional guarantee given the true embedding $\pi(\cdot)$ directly. We illustrate this further via the following results.
To do so, we need the following assumptions:
}

\begin{assumption}\label{ass: data}
The pairs $\{(X_i,S_i)\}_{i\in[n+1]}$ are exchangeable. Conditional on $\{X_i\}_{i\in[n+1]}$, the responses $\{Y_i\}_{i\in[n+1]}$ are independent, with $Y_i\mid X_i\sim P_{Y\mid X=X_i}$ for each $i\in[n+1]$.
\end{assumption}
\begin{assumption}\label{ass: low-rank projection}
   The projection $\hat\pi(\cdot)$ is computed symmetrically with respect to the $n+1$ inputs.
\end{assumption}
Assumption~\ref{ass: data} relaxes the i.i.d.\ condition used in \citet{gibbs2023conformal} to exchangeability, which is standard in conformal inference and accommodates latent-variable generative structures (e.g., admixture models such as LDA \citep{blei2003latent}) that induce dependence among \(\{X_i\}\) while preserving exchangeability (see Theorem \ref{cor: mixture coverage} for details). Assumption~\ref{ass: low-rank projection} ensures the validity of the tilt function $f$ and exchangeability of $n+1$ samples under $P_f$.

To achieve a distribution-free guarantee for $\mathbb{P}(Y_{n+1} \in \hat{C}^*_{\mathrm{rand}}(X_{n+1}) | \hat\pi(X_{n+1}))$ without producing overly wide intervals, we consider one standard relaxation of conditional coverage using kernel reweighting
{such that the tilt   $f(x):=\psi^*(x, x')$ with a given fixed point $x'$, that emphasizes coverage in a neighborhood around the latent embeddings of $x'$. In this analysis, we focus purely on the RKHS component and set $\Phi^*(\cdot)\equiv 0$.}
\begin{theorem}\label{cor: local coverage}
Suppose $\{(X_i, S_i)\}_{i\in[n+1]}\overset{i.i.d}{\sim} P$ and Assumption \ref{ass: low-rank projection} holds. Assume there exists a density kernel $\psi^*_W(w, \cdot)$ on the latent space such that, for all $x_1, x_2\in\mathcal{X}$, $\psi^*_W(\hat\pi(x_1),\hat\pi(x_2))=\psi^*(x_1, x_2)$. Let $W'\mid X_{n+1}=x\sim \psi^*_W(\hat\pi(x), \cdot)$, then we have
\begin{equation}
    \label{eq: corollary 2 neighborhood}
\begin{aligned}
  &\bP(Y_{n+1}\in\hat C^{*}_{\mathrm{rand}}(X_{n+1})\mid W')\\=&1-\alpha-\frac{2\bE[\sum_{i\in[n+1]}\hat \upsilon_{S^{\mathrm{rand}}, i} \psi^*_W(W', \hat\pi(X_i))]}{\mathbb{E}[\psi^*_W( W', \hat\pi(X))]}.
\end{aligned}
\end{equation}
\end{theorem}
This localized version of conformal prediction can be viewed as an approximation of conditional coverage on the event that $W'\approx \hat\pi(X_{n+1})$, with a kernel bandwidth
$\hat \gamma$ that governs the trade-off between conditional adaptivity and statistical stability.

Since our kernel takes the form
$\psi_W^\ast(w,w')=\exp\{-\gamma d_\pi(w,w')\}$, larger values of $\gamma$ induce stronger localization. This improves local adaptivity, but it can also reduce the effective local sample mass
$
p_\gamma(W') := \mathbb{E}\!\left[\psi_W^\ast(W',\widehat{\pi}(X))\right],
$
especially in sparse regions of the latent space, thereby weakening the coverage-gap bound. In contrast, smaller values of $\gamma$ average over a wider neighborhood, increasing $p_\gamma(W')$ and yielding a more stable, though less local, guarantee. Therefore, $\widehat{\gamma}$ is selected data-adaptively by our tuning procedure, together with $\widehat{\lambda}$, to balance locality and stability. Further discussion of this trade-off is provided in Lemma~\ref{lem:gap_order_sqrt_lambda}.

{ The coverage gap on the right-hand side of \eqref{eq: corollary 2 neighborhood} quantifies
how difficult it is to enforce conditional coverage in the neighborhood defined by the kernel, and decreases to at least $O(\sqrt{\hat\lambda})$ (Lemma \ref{lem:gap_order_sqrt_lambda} in Appendix~\ref{appsec: coverage gap}). When $W'$ (equivalently, $\hat\pi(X_{n+1})$) lies in a dense region of the embedding, nearby calibration points are abundant, and the coverage gap shrinks.
It requires a stronger i.i.d. assumption than exchangeability in Assumption \ref{ass: data} in order to give more relevance to data points closer to the test point in the latent space. Under this condition, the coverage gap term becomes more stable.
Rather than showing the gap is asymptotically zero \citep{guan2023localized} with strong distribution and model assumptions, our decomposition of the gap makes the source of deviation explicit and directly estimable (we provide gap estimation in Appendix~\ref{appsec: coverage gap})}

Note, however, that \eqref{eq: corollary 2 neighborhood} is stated for neighborhoods centered at the estimated embedding $\hat{\pi}(X_{n+1})$, not the true one. When $\hat{\pi}(\cdot)$ is a good approximation of the true embedding $\pi(\cdot)$, the guarantee in \eqref{eq: corollary 2 neighborhood} closely matches the conditional guarantee under the true latent embedding, which we show in Appendix \ref{appsec: embedding error local}.

Moreover, the conditional guarantee can be extended to any finite collection of groups encoded by the feature map \(\Phi^*(\cdot)\). In particular, when the covariates are generated from a mixture of latent clusters, \(\Phi^*(\cdot)\) can encode the cluster assignments using the mixture weights $\hat \pi(\cdot)$. In the oracle setting when $\pi(\cdot)$ is known, running the quantile regression in \eqref{eq:quantile reg on RKHS mixture} directly on $\pi(\cdot)$ yields exact conditional coverage for each latent cluster, as formalized below.

\begin{theorem}\label{cor: mixture coverage}
Fix $K\geq 2$ and consider the latent mixture weights $\{W_i\in \Delta^{K-1}\}_{i\in[n+1]}\overset{i.i.d}\sim P_{W}$\footnote{$\Delta^{K-1}=\{x\in \mathbb{R}^{K}: 0\leq x_k\leq 1, \sum_{k\in[K]}x_k =1\}$ is the $(K-1)$-dimensional simplex.}, and observations $\{X_i\mid W_i\}_{i\in[n+1]}\overset{i.i.d}\sim P_{X\mid W}$. Define the true embedding as $\pi(X):=\bE[W\mid X]\in\Delta^{K-1}$. Let
$$T(X):=\arg\max_{k\in[K]}\pi_{k}(X)$$ as the latent cluster of $X$. Let Assumptions \ref{ass: data}, \ref{ass: low-rank projection} hold and further assume $\bP( T(X)=k)>0$ for any $k\in[K]$.
{ Let $\hat C^{*}_{\mathrm{rand}}(\cdot)$ be the randomized prediction set calibrated with the linear term $\Phi^*(X)=(\One\{ T(X)=1\},\dots,\One\{ T(X)=K\})^\top$. Then for every $k\in[K]$,
\begin{equation}\label{eq:mixture-coverage}
\mathbb{P}\big(Y_{n+1}\in \hat C^{*}_{\mathrm{rand}}(X_{n+1}) \mid T(X_{n+1})=k\big)=1-\alpha.
\end{equation}
}
\end{theorem}

In practice, neither $W$ nor $\pi(\cdot)$ is observed and we condition on the estimated embedding $\hat\pi(\cdot)$ and its induced cluster assignment $\hat T(X)$. Our finding is that the finite-sample coverage guarantee with respect to  $\hat T(X_{n+1})$ also holds for any low-rank projection $\hat\pi(\cdot)$ (Corollary \ref{cor:learn_group_coverage} in Appendix~\ref{appsec:proof_mixture_corollary}). In the same section, we further quantify how the finite-sample guarantee based on the estimated embedding $\hat \pi(X)$ deviates from this oracle guarantee. The estimation error in $\hat \pi(X)$ does not compromise the coverage guarantee, but it can affect efficiency. In particular, a coarser embedding still maintains coverage but may yield wider prediction sets.

\section{Experiments}\label{sec: experiments}

In this section, we evaluate SpeedCP across four diverse settings: synthetic admixture data, molecular property prediction with GNNs, and brain tumor MRI analysis with a CNN. We also analyze citation-count prediction on the arXiv dataset using topic-modeling features (see Appendix~\ref{appendix: arxiv}). We summarize the results of SpeedCP and compare them with four other benchmarks: CondCP \citep{gibbs2023conformal}, PCP \citep{zhang2024posterior},  RLCP \citep{hore2023conformal} and SplitCP \citep{papadopoulos2002inductive}.

The main experiments in this section use the RBF kernel on the low-rank latent representation. We further assess whether the empirical performance of SpeedCP is sensitive to the kernel choice in Appendix \ref{app:kernel-family-sensitivity}.

\paragraph{Synthetic Experiments.}
We evaluate the performance of our method using synthetic datasets in the admixture setting where $X$ is generated from a mixture of $K=3$ latent distributions. We use the mixture proportion $\hat{\pi}(X)$ as an input to all CP methods. In this case, $\sum_{k=1}^K\hat\pi_k(X)=1$ and $\hat\pi_k(X)>0$, yielding the latent space as a simplex.

{  To test whether a method can effectively adapt to a covariate shift, we symmetrically sample the calibration mixture proportions over the simplex, but sample the test mixture proportions highly concentrated near one vertex (see the density plots in Figure \ref{fig:density}). We also consider two different predictors, a linear regression and a two-layer neural network (NN), to assess the model-agnostic behavior of the conditional conformal methods.}

\begin{figure*}[!ht]
    \centering
    \includegraphics[width=\textwidth]{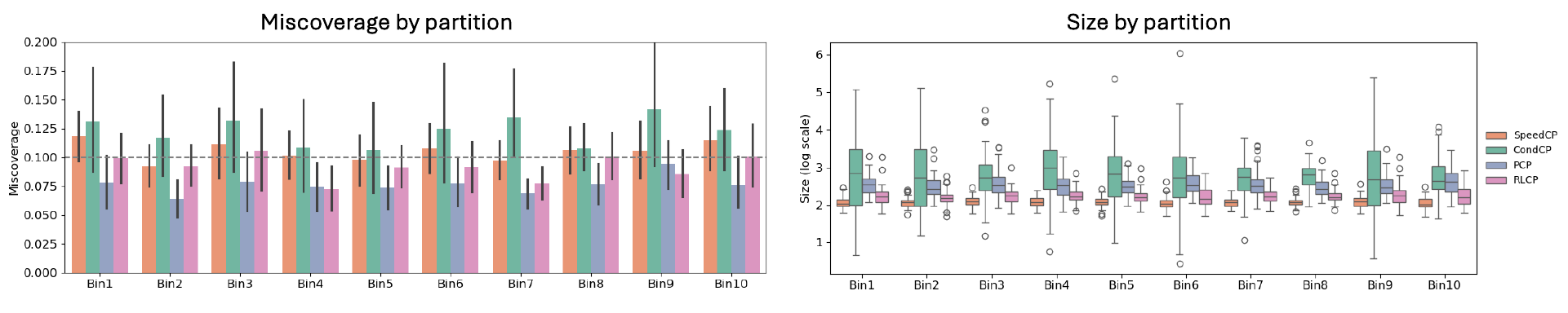}
    \vspace{-4mm}
    \caption{\textbf{SpeedCP} exhibits robust conditional coverage, maintaining a near-constant 0.1 miscoverage rate across the latent space partitions in the synthetic admixture setting. This result holds even with a low-quality predictor $\hat{\mu}$, given by linear regression. Furthermore, SpeedCP achieves the highest efficiency, consistently producing the smallest prediction sets across all bins. The binning scheme is shown in Figure~\ref{fig:density}.}

    \label{fig:bin2}
\end{figure*}

{
\begin{table*}[!ht]
\vspace{-1mm}
\caption{Marginal miscoverage, prediction set size, and computation time across $\hat \mu$ in the synthetic admixture setting.}
\vspace{-1.2mm}
  \centering
  \resizebox{0.78\textwidth}{!}{
    \begin{tabular}{l c c c c c}
      \toprule

      Method
        & \multicolumn{2}{c}{\textbf{Marginal miscoverage ($\alpha=0.1$)}}
        & \multicolumn{2}{c}{\textbf{Prediction set size}}
        & \multicolumn{1}{c}{\textbf{Time (seconds)}}\\
      \cmidrule(lr){2-3}
      \cmidrule(lr){4-5}
        & LR & NN & LR & NN & LR \\
      \midrule
      \textbf{SpeedCP}
        & \textbf{0.105 \,$\pm$\,0.07}
        & \textbf{0.098 \,$\pm$\,0.02}
        & \textbf{2.074 \,$\pm$\,0.14}
        & \textbf{0.804 \,$\pm$\,0.06}
        & 22.05 \,$\pm$\,6.22\\
      CondCP
        & 0.123 \,$\pm$\,0.13
        & 0.124 \,$\pm$\,0.05
        & 2.782 \,$\pm$\,0.72
        & 2.126 \,$\pm$\,0.31
        & 1332.67 \,$\pm$\,129.93 \\

      PCP
        & 0.076 \,$\pm$\,0.06
        & 0.088 \,$\pm$\,0.02
        & 2.535 \,$\pm$\,0.02
        & 0.910 \,$\pm$\,0.13
        & 141.64 \,$\pm$\,14.48 \\
      RLCP
        & 0.092 \,$\pm$\,0.07
        & 0.089 \,$\pm$\,0.02
        & 2.228 \,$\pm$\,0.22
        & 0.864 \,$\pm$\,0.07
        & 22.05 \,$\pm$\,0.07\\
      \bottomrule
    \end{tabular}
  }
  \label{tab:predictors}
   \vspace{-3mm}
\end{table*}
}

Figure~\ref{fig:bin2} shows that SpeedCP achieves miscoverage closest to the target level of 0.1 while producing the smallest prediction sets.
CondCP solves the same regularized quantile regression problem, but it fails to achieve reasonable coverage in several bins and produces overly wide intervals. This happens because the optimization solver it relies on does not return exact solutions and outputs conservative approximations.

In contrast, our path algorithm uses a stable piecewise-linear structure of the problem and tracks boundary events precisely, yielding tighter and more accurate prediction sets. Both PCP and RLCP tend to overcover in most bins and produce large prediction sets, as their performance is sensitive to the quality of the base predictor $\hat\mu(\cdot)$.

Appendix~\ref{appendix:sim} provides additional details on the experimental design, ablations over different sample sizes $n$, and additional conditional coverage results using a neural network predictor.

\begin{figure*}[!htbp]
    \centering
   \includegraphics[width=0.88\textwidth]{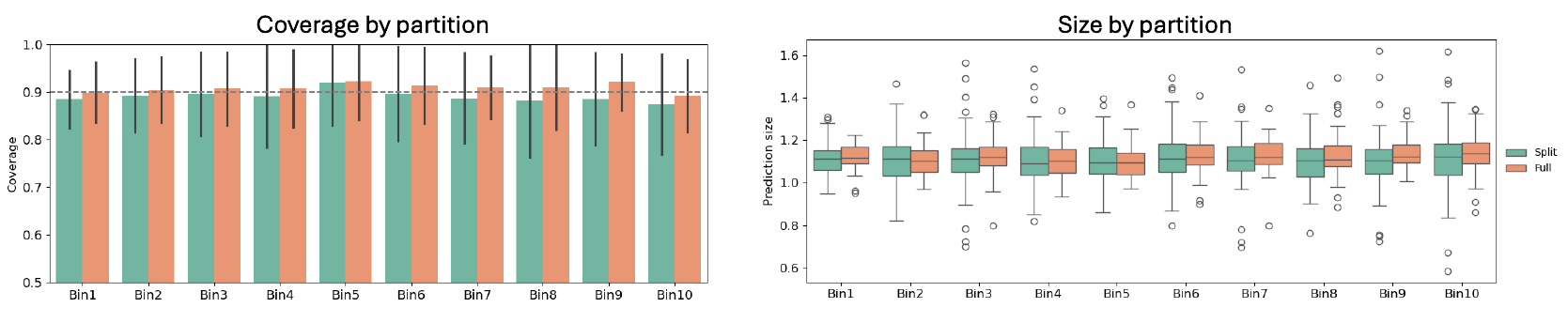}
    \caption{Conditional miscoverage and prediction set size for each fixed partition on the latent space in the synthetic admixture setting. We observe small overestimation of coverage when we use the full calibration set instead of setting aside a dataset for tuning $(\gamma, \lambda)$. Overall, the two methods are marginally different.}
    \label{fig:split}
    \vspace{-3mm}
\end{figure*}
\paragraph{Hyperparameter Tuning.}
In the theoretical results (e.g., Theorem~\ref{cor: local coverage}), we treat the hyperparameters $(\lambda,\gamma)$ as fixed to preserve exchangeability and avoid additional dependence in the conformal construction. In practice, however, we tune $(\lambda,\gamma)$ via cross-validation using the calibration data to avoid introducing an additional validation split, which can substantially increase variability in small-sample settings. In Figure \ref{fig:split}, we find empirically that running cross-validation on the calibration set, instead of using a held-out split set, does not significantly impact coverage. The resulting marginal and partition-wise coverages are all close to those obtained under a split strategy that reserves part of the calibration data exclusively for tuning (see Appendix~\ref{appendix:calib}). When sufficient data are available, we recommend using a held-out split set for hyperparameter selection, so that $(\hat\gamma,\hat\lambda)$ can be chosen independently of the conformal calibration step.

\paragraph{Molecule Graphs.} We evaluate our method on three molecular property prediction benchmarks: QM9, QM7b, and ESOL \citep{wu2018moleculenet}. For each dataset, we train a GNN to predict a molecular property: the HOMO--LUMO gap for QM9, polarizability for QM7b, and solubility for ESOL. We extract the last 64-dimensional graph embedding after pooling, and reduce it to 3 dimensions via PCA. Our objective is to achieve nominal $0.9$ coverage across this low-dimensional representation of the molecular graphs. To assess conditional coverage, we partition the PC space into 6--8 regions using Voronoi tessellation, and compute coverage within each region. We aggregate results over 50 random subsamples of 2000 graphs, and report the results in Figure~\ref{fig:molecule} and Table~\ref{tab:1}.

We observe that SpeedCP achieves nominal coverage consistently across all partitions, while achieving sharp prediction sets. In contrast, SplitCP undercovers in some regions of the latent feature space, since it relies on a single global threshold and therefore cannot adapt to local heterogeneity in the score distribution. The other baselines achieve more stable conditional coverage than SplitCP, but at the cost of larger prediction sets.
We provide additional results for each dataset in Appendix~\ref{appendix:real}.

\begin{figure*}[!htbp]
    \centering
    \includegraphics[width=0.79\textwidth]{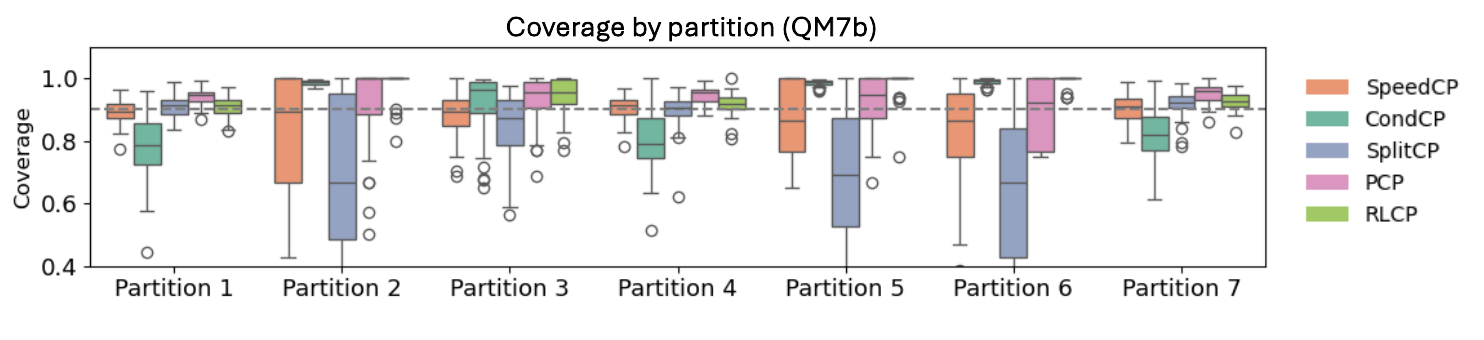}
    \vspace{-3mm}
    \caption{Coverage on fixed partitions of the PC space for QM7b. We use PCA on the last layer embeddings of GNN with $K=3$ dimensions. The dashed line denotes the target coverage rate $1-\alpha=0.9$.}
    \label{fig:molecule}
\end{figure*}

\begin{table*}[!ht]
\vspace{-1mm}
\caption{Mean prediction set size and computation time of molecule datasets.}
\vspace{-1.2mm}
  \centering
  \resizebox{0.89\textwidth}{!}{

      \begin{tabular}{l c c c c c c c}
        \toprule
        Method
          & \multicolumn{3}{c}{\textbf{Prediction set size}} & \multicolumn{3}{c}{\textbf{Time (seconds)}}\\
        \cmidrule(lr){2-4} \cmidrule(lr){5-7}

          & QM9 & QM7b & ESOL
          & QM9 & QM7b & ESOL \\
        \midrule
        \textbf{SpeedCP} & 1.135 \,$\pm$0.25 & \textbf{0.902\,$\pm$0.44} & \textbf{1.789\,$\pm$0.38}
                  & 31.061\,$\pm$2.94 & 33.056\,$\pm$7.23
                  & 15.442\,$\pm$1.55 \\
        CondCP      & 1.922\,$\pm$0.40 & 1.447\,$\pm$1.17 & 2.683\,$\pm$0.42
                  & 1531.15\,$\pm$195.60 & 1890.38\,$\pm$166.62
                  & 625.06\,$\pm$64.54 \\
        SplitCP & \textbf{1.122\,$\pm$0.122} & 0.999\,$\pm$0.37 & 1.800\,$\pm$0.17
                  & $< 0.01$ & $< 0.01$
                  &  $< 0.01$\\
        PCP & 1.530\,$\pm$0.87 & 1.303\,$\pm$1.07 & 2.261\,$\pm$1.00
                  & 38.018\,$\pm$3.48 & 47.218\,$\pm$6.50
                  & 21.659\,$\pm$2.36 \\
        RLCP & 1.554\,$\pm$0.89 & 1.286\,$\pm$1.04  & 2.248\,$\pm$1.02
                  & 1.157\,$\pm$0.02 & 1.148\,$\pm$0.01
                  & 0.668\,$\pm$0.00 \\
        \bottomrule
      \end{tabular}
      }
  \label{tab:1}\vspace{-3mm}
\end{table*}

\begin{table*}[!t]
  \centering
  \vspace{-1mm}
  \caption{Mean coverage, prediction set size, and computation time across predicted labels in the MRI dataset.}
  \vspace{-1.2mm}
  \resizebox{\textwidth}{!}{
      \begin{tabular}{l c c c c c c c c }
        \toprule
        Method
          & \multicolumn{3}{c}{\textbf{Target coverage ($1-\alpha=0.9$)}}
          & \multicolumn{3}{c}{\textbf{Prediction set size}} & \textbf{Time (seconds)}\\
        \cmidrule(lr){2-4} \cmidrule(lr){5-7}

          & Marginal & Healthy & Tumor & Marginal & Healthy & Tumor \\
        \midrule
        \textbf{SpeedCP($\One$)}\footnotemark &{0.910\,$\pm$0.01} & \textbf{0.902\,$\pm$0.02} & 0.914\,$\pm$0.02
                  & {0.262\,$\pm$0.09} & \textbf{0.250\,$\pm$0.09}
                  & 0.275\,$\pm$0.08 &  244.1\,$\pm$9.2\\
       \textbf{SpeedCP}($\Phi^*$)&{0.908\,$\pm$0.02} &\textbf{0.902\,$\pm$0.02} & \textbf{0.901\,$\pm$0.02 }
                  & {0.282\,$\pm$0.08}  & {0.266\,$\pm$0.08}
                  &{0.295\,$\pm$0.08}&  270.5\,$\pm$13.9\\
        SplitCP
        & \textbf{ 0.898\,$\pm$0.01} & 0.888\,$\pm$0.02 & 0.903\,$\pm$0.02
                  & 0.348\,$\pm$0.00 & 0.348\,$\pm$0.00
                  & 0.348\,$\pm$0.00  & $< 0.01$\\
        PCP
        & 0.918\,$\pm$0.01 & 0.945\,$\pm$0.02 & 0.902\,$\pm$0.02
                  &\textbf{ 0.231\,$\pm$0.27} & 0.281\,$\pm$0.26
                  & \textbf{0.201\,$\pm$0.28} & 162.1\,$\pm$ 13.9 \\
        RLCP
        & \textbf{ 0.898\,$\pm$0.01} & 0.888\,$\pm$0.02 & 0.903\,$\pm$0.02
                  & 0.348\,$\pm$0.00 & 0.348\,$\pm$0.00
                  & 0.348\,$\pm$0.00  & 3.48$\pm$ 0.08\\
        \bottomrule
      \end{tabular}
  }

  \vspace{-3mm}
  \label{tab:coverage for mri}
\end{table*}

\footnotetext{For the Brain Tumor MRI data, we use \textbf{SpeedCP$(\Phi^*)$} to denote calibration with a linear term that includes predicted labels, whereas \textbf{SpeedCP($\One$)} uses an intercept-only linear term with $\Phi^*(X)=1$.}

\paragraph{Brain Tumor MRI.} We evaluate on a brain--tumor MRI dataset from Kaggle\footnote{{https://www.kaggle.com/datasets/murtozalikhon/brain-tumor-multimodal-image-ct-and-mri}.} with labels $\{\texttt{healthy},\texttt{tumor}\}$. We train a CNN classifier $\hat\mu(\cdot) $ on 2{,}000 images and extract NN features from the last layer for calibration (training details in Appendix~\ref{appsec: brain mri}). Table~\ref{tab:coverage for mri} shows that even with intercept-only calibration ($\Phi^*(X)=1$), our RKHS component alone gives a good approximation for predicted-label coverage. When covariate shift aligns with label groups, adding linear terms for the predicted label, { $\Phi^*(X) = (\One\{\hat\mu(X)=\texttt{healthy}\},\ \One\{\hat\mu(X)=\texttt{tumor}\})^\top$}, provides better conditional coverage. In contrast, SplitCP achieves comparable coverage but requires more conservative sets than ours, while RLCP fails to exploit locality in the 256-dimensional feature space and effectively reduces to uniform weighting, thus converging to SplitCP.
PCP tends to overcover, especially for the healthy group, and its cutoffs are unstable with high variance and frequent near-zero values (see Appendix Table \ref{tab:cutoff-mri-raw}), thereby producing overly conservative conditional coverage.

\raggedbottom
\section{Limitations and Future Directions}\label{sec:limitation}

While we believe our algorithm can be broadly applicable in high-dimensional problems, especially when prior knowledge is limited, we highlight several limitations and directions for future work.

(1) We currently fix the miscoverage level $\alpha$ for all test points. However, $\alpha$ could be made adaptive based on latent structure or user-specified utility. For example, one might use a stricter $\alpha$ for subpopulations deemed more critical \citep{takeuchi2006entire,gauthier2025adaptive}, thereby allocating tighter guarantees where they matter most. (2) Incorporating weights into our quantile regression based on uncertainty or embeddings' importance could further refine coverage and interpretability \citep{jang2023tight}.  Although we focus on scalar regression tasks, the RKHS-based framework can be extended to structured prediction problems such as text generation \citep{sun2023conformal,farquhar2024detecting,su2024api,shahrokhi2025conformal}, image completion \citep{angelopoulos2020uncertainty,wieslander2020deep}, and other multivariate problems \citep{messoudi2021copula, johnstone2022exact, xu2024conformal} where uncertainty quantification over complex outputs is crucial.

\paragraph{Acknowledgements.}
C.D. was supported by the National Science Foundation under Award Number 2238616, as well as the resources provided by the University of Chicago's Research Computing Center. The authors are grateful to John Cherian for his valuable discussions and feedback on this manuscript.

\section*{Impact Statement}
This paper introduces SpeedCP, a method designed to improve the computational efficiency and conditional coverage of conformal prediction. Our research utilizes publicly available datasets and does not involve human subjects or sensitive personal data. We do not anticipate any significant risks of misuse or negative societal consequences arising from this methodology.

\bibliographystyle{plainnat}
\bibliography{biblio_mrlookup}

@article {gibbs2023conformal,
    AUTHOR = {Gibbs, Isaac and Cherian, John J. and Cand\`es, Emmanuel J.},
     TITLE = {Conformal prediction with conditional guarantees},
   JOURNAL = {J. R. Stat. Soc. Ser. B. Stat. Methodol.},
  FJOURNAL = {Journal of the Royal Statistical Society. Series B.
              Statistical Methodology},
    VOLUME = {87},
      YEAR = {2025},
    NUMBER = {4},
     PAGES = {1100--1126},
      ISSN = {1369-7412,1467-9868},
   MRCLASS = {62G15 (62M20)},
  MRNUMBER = {4961264},
MRREVIEWER = {Akio\ Arimoto},
       DOI = {10.1093/jrsssb/qkaf008},
       URL = {https://doi.org/10.1093/jrsssb/qkaf008},
}

@incollection {papadopoulos2002inductive,
    AUTHOR = {Papadopoulos, Harris and Proedrou, Kostas and Vovk, Volodya
              and Gammerman, Alex},
     TITLE = {Inductive confidence machines for regression},
 BOOKTITLE = {Machine learning: {ECML} 2002},
    SERIES = {Lecture Notes in Comput. Sci.},
    VOLUME = {2430},
     PAGES = {345--356},
 PUBLISHER = {Springer, Berlin},
      YEAR = {2002},
      ISBN = {3-540-44036-4},
   MRCLASS = {Expansion},
  MRNUMBER = {2050303},
       DOI = {10.1007/3-540-36755-1\{_}29}

@article {vovk2012conditional,
    AUTHOR = {Vovk, Vladimir},
     TITLE = {Conditional validity of inductive conformal predictors},
   JOURNAL = {Mach. Learn.},
  FJOURNAL = {Machine Learning},
    VOLUME = {92},
      YEAR = {2013},
    NUMBER = {2-3},
     PAGES = {349--376},
      ISSN = {0885-6125,1573-0565},
   MRCLASS = {62G15 (68T05)},
  MRNUMBER = {3080332},
MRREVIEWER = {Robert\ Hable},
       DOI = {10.1007/s10994-013-5355-6},
       URL = {https://doi.org/10.1007/s10994-013-5355-6},
}

@article {foygel2021limits,
    AUTHOR = {Barber, Rina Foygel and Cand\`es, Emmanuel J. and Ramdas,
              Aaditya and Tibshirani, Ryan J.},
     TITLE = {The limits of distribution-free conditional predictive
              inference},
   JOURNAL = {Inf. Inference},
  FJOURNAL = {Information and Inference. A Journal of the IMA},
    VOLUME = {10},
      YEAR = {2021},
    NUMBER = {2},
     PAGES = {455--482},
      ISSN = {2049-8764,2049-8772},
   MRCLASS = {62G08},
  MRNUMBER = {4270755},
MRREVIEWER = {Jesse\ C.\ Frey},
       DOI = {10.1093/imaiai/iaaa017},
       URL = {https://doi.org/10.1093/imaiai/iaaa017},
}

@article{vovk2003mondrian,
  title={Mondrian confidence machine},
  author={Vovk, Vladimir and Lindsay, David and Nouretdinov, Ilia and Gammerman, Alex},
  journal={Technical Report},
  year={2003}
}

@inproceedings{jung2022batch,
  title={Batch Multivalid Conformal Prediction},
  author={Jung, Christopher and Noarov, Georgy and Ramalingam, Ramya and Roth, Aaron},
  booktitle={International Conference on Learning Representations (ICLR)},
  year={2023}
}

@article {kimeldorf1971some,
    AUTHOR = {Kimeldorf, George and Wahba, Grace},
     TITLE = {Some results on {T}chebycheffian spline functions},
   JOURNAL = {J. Math. Anal. Appl.},
  FJOURNAL = {Journal of Mathematical Analysis and Applications},
    VOLUME = {33},
      YEAR = {1971},
     PAGES = {82--95},
      ISSN = {0022-247X},
   MRCLASS = {41.60 (46.00)},
  MRNUMBER = {290013},
MRREVIEWER = {L.\ Raymon},
       DOI = {10.1016/0022-247X(71)90184-3},
       URL = {https://doi.org/10.1016/0022-247X(71)90184-3},
}

@article {li2007quantile,
    AUTHOR = {Li, Youjuan and Liu, Yufeng and Zhu, Ji},
     TITLE = {Quantile regression in reproducing kernel {H}ilbert spaces},
   JOURNAL = {J. Amer. Statist. Assoc.},
  FJOURNAL = {Journal of the American Statistical Association},
    VOLUME = {102},
      YEAR = {2007},
    NUMBER = {477},
     PAGES = {255--268},
      ISSN = {0162-1459,1537-274X},
   MRCLASS = {62J02 (62G30)},
  MRNUMBER = {2293307},
       DOI = {10.1198/016214506000000979},
       URL = {https://doi.org/10.1198/016214506000000979},
}

@misc{hastie2009elements,
  title={The elements of statistical learning: data mining, inference, and prediction},
  author={Hastie, Trevor},
  year={2009},
  publisher={springer}
}

@article{udell2019big,
  title={Big data is low rank},
  author={Udell, Madeleine and Townsend, Alex},
  journal={SIAM News},
  volume={52},
  number={9},
  year={2019}
}

@article {efron2004least,
    AUTHOR = {Efron, Bradley and Hastie, Trevor and Johnstone, Iain and
              Tibshirani, Robert},
     TITLE = {Least angle regression},
      NOTE = {With discussion, and a rejoinder by the authors},
   JOURNAL = {Ann. Statist.},
  FJOURNAL = {The Annals of Statistics},
    VOLUME = {32},
      YEAR = {2004},
    NUMBER = {2},
     PAGES = {407--499},
      ISSN = {0090-5364,2168-8966},
   MRCLASS = {62J07},
  MRNUMBER = {2060166},
MRREVIEWER = {Holger\ Dette},
       DOI = {10.1214/009053604000000067},
       URL = {https://doi.org/10.1214/009053604000000067},
}

@article {tibshirani1996regression,
    AUTHOR = {Tibshirani, Robert},
     TITLE = {Regression shrinkage and selection via the lasso},
   JOURNAL = {J. Roy. Statist. Soc. Ser. B},
  FJOURNAL = {Journal of the Royal Statistical Society. Series B.
              Methodological},
    VOLUME = {58},
      YEAR = {1996},
    NUMBER = {1},
     PAGES = {267--288},
      ISSN = {0035-9246},
   MRCLASS = {62J05 (62J07)},
  MRNUMBER = {1379242},
       URL =
              {http://links.jstor.org/sici?sici=0035-9246(1996)58:1<267:RSASVT>2.0.CO;2-G&origin=MSN},
}

@book {tibshirani2011solution,
    AUTHOR = {Tibshirani, Ryan J.},
     TITLE = {The {S}olution {P}ath of the {G}eneralized {L}asso},
      NOTE = {Thesis (Ph.D.)--Stanford University},
 PUBLISHER = {ProQuest LLC, Ann Arbor, MI},
      YEAR = {2011},
     PAGES = {96},
      ISBN = {979-8672-19479-0},
   MRCLASS = {Thesis},
  MRNUMBER = {4172293},
       URL =
              {http://gateway.proquest.com/openurl?url_ver=Z39.88-2004&rft_val_fmt=info:ofi/fmt:kev:mtx:dissertation&res_dat=xri:pqm&rft_dat=xri:pqdiss:28168053},
}

@article{friedman2010regularization,
  title={Regularization paths for generalized linear models via coordinate descent},
  author={Friedman, Jerome H and Hastie, Trevor and Tibshirani, Rob},
  journal={Journal of statistical software},
  volume={33},
  number={1},
  pages={1--22},
  year={2010}
}

@article{koenker2005quantile,
  title={Quantile regression},
  author={Koenker, Roger and Hallock, Kevin F},
  journal={Journal of economic perspectives},
  volume={15},
  number={4},
  pages={143--156},
  year={2001},
  publisher={American Economic Association}
}

@article {guan2023localized,
    AUTHOR = {Guan, Leying},
     TITLE = {Localized conformal prediction: a generalized inference
              framework for conformal prediction},
   JOURNAL = {Biometrika},
  FJOURNAL = {Biometrika},
    VOLUME = {110},
      YEAR = {2023},
    NUMBER = {1},
     PAGES = {33--50},
      ISSN = {0006-3444,1464-3510},
   MRCLASS = {92B15 (62M20)},
  MRNUMBER = {4565442},
       DOI = {10.1093/biomet/asac040},
       URL = {https://doi.org/10.1093/biomet/asac040},
}

@article {hore2023conformal,
    AUTHOR = {Hore, Rohan and Barber, Rina Foygel},
     TITLE = {Conformal prediction with local weights: randomization enables
              robust guarantees},
   JOURNAL = {J. R. Stat. Soc. Ser. B. Stat. Methodol.},
  FJOURNAL = {Journal of the Royal Statistical Society. Series B.
              Statistical Methodology},
    VOLUME = {87},
      YEAR = {2025},
    NUMBER = {2},
     PAGES = {549--578},
      ISSN = {1369-7412,1467-9868},
   MRCLASS = {62G15 (62M20)},
  MRNUMBER = {4896656},
MRREVIEWER = {Akio\ Arimoto},
       DOI = {10.1093/jrsssb/qkae103},
       URL = {https://doi.org/10.1093/jrsssb/qkae103},
}

@article{zhang2024posterior,
  title={Posterior conformal prediction},
  author={Zhang, Yao and Cand{\`e}s, Emmanuel J},
  journal={arXiv preprint arXiv:2409.19712},
  year={2024}
}

@article{blei2003latent,
  title={Latent dirichlet allocation},
  author={Blei, David M and Ng, Andrew Y and Jordan, Michael I},
  journal={Journal of machine Learning research},
  volume={3},
  number={Jan},
  pages={993--1022},
  year={2003}
}

@article{wu2018moleculenet,
  title={MoleculeNet: a benchmark for molecular machine learning},
  author={Wu, Zhenqin and Ramsundar, Bharath and Feinberg, Evan N and Gomes, Joseph and Geniesse, Caleb and Pappu, Aneesh S and Leswing, Karl and Pande, Vijay},
  journal={Chemical science},
  volume={9},
  number={2},
  pages={513--530},
  year={2018},
  publisher={Royal Society of Chemistry}
}

@inproceedings{takeuchi2006entire,
  title={The Entire Solution Path of Kernel-based Nonparametric Conditional Quantile Estimator},
  author={Takeuchi, Ichiro and Nomura, Kaname and Kanamori, Takafumi},
  booktitle={The 2006 IEEE International Joint Conference on Neural Network Proceedings},
  pages={153--158},
  year={2006},
  organization={IEEE}
}

@inproceedings{gauthier2025adaptive,
  title={Adaptive Coverage Policies in Conformal Prediction},
  author={Gauthier, Etienne and Bach, Francis and Jordan, Michael I},
  booktitle={Proceedings of the 29th International Conference on Artificial Intelligence and Statistics},
  year={2026}
}

@article{jang2023tight,
  title={Tight Distribution-Free Confidence Intervals for Local Quantile Regression},
  author={Jang, Jayoon and Cand{\`e}s, Emmanuel},
  journal={arXiv preprint arXiv:2307.08594},
  year={2023}
}

@article{sun2023conformal,
  title={Conformal prediction for uncertainty-aware planning with diffusion dynamics model},
  author={Sun, Jiankai and Jiang, Yiqi and Qiu, Jianing and Nobel, Parth and Kochenderfer, Mykel J and Schwager, Mac},
  journal={Advances in Neural Information Processing Systems},
  volume={36},
  pages={80324--80337},
  year={2023}
}

@article{farquhar2024detecting,
  title={Detecting hallucinations in large language models using semantic entropy},
  author={Farquhar, Sebastian and Kossen, Jannik and Kuhn, Lorenz and Gal, Yarin},
  journal={Nature},
  volume={630},
  number={8017},
  pages={625--630},
  year={2024},
  publisher={Nature Publishing Group UK London}
}

@inproceedings{su2024api,
    title = "{API} Is Enough: Conformal Prediction for Large Language Models Without Logit-Access",
    author = "Su, Jiayuan  and
      Luo, Jing  and
      Wang, Hongwei  and
      Cheng, Lu",
    booktitle = "Findings of the Association for Computational Linguistics: EMNLP 2024",
    year = "2024",
    address = "Miami, Florida, USA",
    publisher = "Association for Computational Linguistics",
    pages = "979--995"
}

@inproceedings{shahrokhi2025conformal,
  title={Conformal Prediction Sets for Deep Generative Models via Reduction to Conformal Regression},
  author={Shahrokhi, Hooman and Roy, Devjeet Raj and Yan, Yan and Arnaoudova, Venera and Doppa, Jana},
  booktitle={Conference on Uncertainty in Artificial Intelligence},
  pages={3718--3748},
  year={2025},
  organization={PMLR}
}

@inproceedings{angelopoulos2020uncertainty,
  title={Uncertainty Sets for Image Classifiers using Conformal Prediction},
  author={Angelopoulos, Anastasios Nikolas and Bates, Stephen and Jordan, Michael and Malik, Jitendra},
  booktitle={International Conference on Learning Representations},
  year={2021}
}

@article{wieslander2020deep,
  title={Deep learning with conformal prediction for hierarchical analysis of large-scale whole-slide tissue images},
  author={Wieslander, H{\aa}kan and Harrison, Philip J and Skogberg, Gabriel and Jackson, Sonya and Frid{\'e}n, Markus and Karlsson, Johan and Spjuth, Ola and W{\"a}hlby, Carolina},
  journal={IEEE journal of biomedical and health informatics},
  volume={25},
  number={2},
  pages={371--380},
  year={2020},
  publisher={IEEE}
}

@article{messoudi2021copula,
  title={Copula-based conformal prediction for multi-target regression},
  author={Messoudi, Soundouss and Destercke, S{\'e}bastien and Rousseau, Sylvain},
  journal={Pattern Recognition},
  volume={120},
  pages={108101},
  year={2021},
  publisher={Elsevier}
}

@inproceedings{johnstone2022exact,
  title={Exact and Approximate Conformal Inference for Multi-Output Regression},
  author={Johnstone, Chancellor and Ndiaye, Eugene},
  booktitle={Fourteenth Symposium on Conformal and Probabilistic Prediction with Applications (COPA 2025)},
  pages={153--172},
  year={2025},
  organization={PMLR}
}

@inproceedings{xu2024conformal,
author = {Xu, Chen and Jiang, Hanyang and Xie, Yao},
title = {Conformal prediction for multi-dimensional time series by ellipsoidal sets},
year = {2024},
publisher = {JMLR.org},
booktitle = {Proceedings of the 41st International Conference on Machine Learning},
articleno = {2268},
numpages = {24},
location = {Vienna, Austria},
series = {ICML'24}
}

@article{angelopoulos2024theoretical,
  title={Theoretical foundations of conformal prediction},
  author={Angelopoulos, Anastasios N and Barber, Rina Foygel and Bates, Stephen},
  journal={arXiv preprint arXiv:2411.11824},
  year={2024}
}

@inproceedings{hofmann1999probabilistic,
  title={Probabilistic latent semantic indexing},
  author={Hofmann, Thomas},
  booktitle={Proceedings of the 22nd annual international ACM SIGIR Conference on Research and Development in Information Retrieval},
  pages={50--57},
  year={1999}
}

@article{donoho2003does,
  title={When does non-negative matrix factorization give a correct decomposition into parts?},
  author={Donoho, David and Stodden, Victoria},
  journal={Advances in Neural Information Processing Systems},
  volume={16},
  year={2003}
}

@incollection {arora2012learning,
    AUTHOR = {Arora, Sanjeev and Ge, Rong and Moitra, Ankur},
     TITLE = {Learning topic models---going beyond {SVD}},
 BOOKTITLE = {2012 {IEEE} 53rd {A}nnual {S}ymposium on {F}oundations of
              {C}omputer {S}cience---{FOCS} 2012},
     PAGES = {1--10},
 PUBLISHER = {IEEE Computer Soc., Los Alamitos, CA},
      YEAR = {2012},
      ISBN = {978-0-7685-4874-6},
   MRCLASS = {68T05 (62H25 68P99)},
  MRNUMBER = {3185945},
}

@book {aitchison1982statistical,
    AUTHOR = {Aitchison, J.},
     TITLE = {The statistical analysis of compositional data},
    SERIES = {Monographs on Statistics and Applied Probability},
 PUBLISHER = {Chapman \& Hall, London},
      YEAR = {1986},
     PAGES = {xvi+416},
      ISBN = {0-412-28060-4},
   MRCLASS = {62H25 (62-07)},
  MRNUMBER = {865647},
MRREVIEWER = {Graham\ J. G. Upton},
       DOI = {10.1007/978-94-009-4109-0},
       URL = {https://doi.org/10.1007/978-94-009-4109-0},
}

@article{araujo2001successive,
  title={The successive projections algorithm for variable selection in spectroscopic multicomponent analysis},
  author={Ara{\'u}jo, M{\'a}rio C{\'e}sar Ugulino and Saldanha, Teresa Cristina Bezerra and Galvao, Roberto Kawakami Harrop and Yoneyama, Takashi and Chame, Henrique Caldas and Visani, Valeria},
  journal={Chemometrics and Intelligent Laboratory Systems},
  volume={57},
  number={2},
  pages={65--73},
  year={2001},
  publisher={Elsevier}
}

@article{gillis2013fast,
  title={Fast and robust recursive algorithmsfor separable nonnegative matrix factorization},
  author={Gillis, Nicolas and Vavasis, Stephen A},
  journal={IEEE Transactions on Pattern Analysis and Machine Intelligence},
  volume={36},
  number={4},
  pages={698--714},
  year={2013},
  publisher={IEEE}
}

@article {javadi2020nonnegative,
    AUTHOR = {Javadi, Hamid and Montanari, Andrea},
     TITLE = {Nonnegative matrix factorization via archetypal analysis},
   JOURNAL = {J. Amer. Statist. Assoc.},
  FJOURNAL = {Journal of the American Statistical Association},
    VOLUME = {115},
      YEAR = {2020},
    NUMBER = {530},
     PAGES = {896--907},
      ISSN = {0162-1459,1537-274X},
   MRCLASS = {62H30 (62F35)},
  MRNUMBER = {4107687},
       DOI = {10.1080/01621459.2019.1594832},
       URL = {https://doi.org/10.1080/01621459.2019.1594832},
}

@article {klopp2021assigning,
    AUTHOR = {Klopp, Olga and Panov, Maxim and Sigalla, Suzanne and
              Tsybakov, Alexandre B.},
     TITLE = {Assigning topics to documents by successive projections},
   JOURNAL = {Ann. Statist.},
  FJOURNAL = {The Annals of Statistics},
    VOLUME = {51},
      YEAR = {2023},
    NUMBER = {5},
     PAGES = {1989--2014},
      ISSN = {0090-5364,2168-8966},
   MRCLASS = {62G05 (60C05)},
  MRNUMBER = {4678793},
       DOI = {10.1214/23-aos2316},
       URL = {https://doi.org/10.1214/23-aos2316},
}

@article {schwarz1978estimating,
    AUTHOR = {Schwarz, Gideon},
     TITLE = {Estimating the dimension of a model},
   JOURNAL = {Ann. Statist.},
  FJOURNAL = {The Annals of Statistics},
    VOLUME = {6},
      YEAR = {1978},
    NUMBER = {2},
     PAGES = {461--464},
      ISSN = {0090-5364,2168-8966},
   MRCLASS = {62F99},
  MRNUMBER = {468014},
MRREVIEWER = {Andrew\ Harvey},
       URL =
              {http://links.jstor.org/sici?sici=0090-5364(197803)6:2<461:ETDOAM>2.0.CO;2-5&origin=MSN},
}

@article {tan2022communication,
    AUTHOR = {Tan, Kean Ming and Battey, Heather and Zhou, Wen-Xin},
     TITLE = {Communication-constrained distributed quantile regression with
              optimal statistical guarantees},
   JOURNAL = {J. Mach. Learn. Res.},
  FJOURNAL = {Journal of Machine Learning Research (JMLR)},
    VOLUME = {23},
      YEAR = {2022},
     PAGES = {Paper No. [272], 61},
      ISSN = {1532-4435,1533-7928},
   MRCLASS = {62J05 (62F03 62F12 62G08 65K10)},
  MRNUMBER = {4577711},
MRREVIEWER = {Fengyang\ He},
}

@article {boucheron2005theory,
    AUTHOR = {Boucheron, St\'{e}phane and Bousquet, Olivier and Lugosi,
              G\'{a}bor},
     TITLE = {Theory of classification: a survey of some recent advances},
   JOURNAL = {ESAIM Probab. Stat.},
  FJOURNAL = {ESAIM. Probability and Statistics},
    VOLUME = {9},
      YEAR = {2005},
     PAGES = {323--375},
      ISSN = {1292-8100,1262-3318},
   MRCLASS = {68T10 (60E15 62G08 62H30)},
  MRNUMBER = {2182250},
MRREVIEWER = {Joachim\ Krauth},
       DOI = {10.1051/ps:2005018},
       URL = {https://doi.org/10.1051/ps:2005018},
}

@book {rasmussen2006gaussian,
    AUTHOR = {Rasmussen, Carl Edward and Williams, Christopher K. I.},
     TITLE = {Gaussian processes for machine learning},
    SERIES = {Adaptive Computation and Machine Learning},
 PUBLISHER = {MIT Press, Cambridge, MA},
      YEAR = {2006},
     PAGES = {xviii+248},
      ISBN = {978-0-262-18253-9},
   MRCLASS = {68T05 (60G15 62G08 62H30 93E35)},
  MRNUMBER = {2514435},
MRREVIEWER = {Wenbo\ V.\ Li},
}

@book{stein1999interpolation,
  title     = {Interpolation of Spatial Data: Some Theory for Kriging},
  author    = {Stein, Michael L.},
  year      = {1999},
  publisher = {Springer}
}

@inproceedings{gorham2017measuring,
  title     = {Measuring Sample Quality with Kernels},
  author    = {Gorham, Jackson and Mackey, Lester},
  booktitle = {Proceedings of the 34th International Conference on Machine Learning},
  pages     = {1292--1301},
  year      = {2017},
  publisher = {PMLR}
}

@article{clement2019arxiv,
  title={On the use of {arXiv} as a dataset},
  author={Clement, Colin B and Bierbaum, Matthew and O'Keeffe, Kevin P and Alemi, Alexander A},
  journal={arXiv preprint arXiv:1905.00075},
  year={2019}
}

\clearpage
\appendix

\section*{Appendix}
\section{Notation and Related Work}
\subsection{Notation}\label{notation}

 For any set $\mathcal{G}$, let $|\mathcal{G}|$ denote its cardinality. Given a vector $\eta\in\bR^p$, we use $\eta(i)$ or $\eta_i$ to represent the $i$-th entry. For any $n \in \mathbb{N}$, let $[n]$ denote the index set $\{1, \dots, n\}$. Throughout this paper, we denote the sets of variables with simple bold letters (e.g. $\textbf{X}\in \bR^{n \times p}=(X_1, X_2, \dots, X_n)^\top$). Let capital letter $P$ denote the joint distribution and $P_X$ denote the marginal distribution of $X$.

Given a value $z$, let $[z]_{+}=\max(z, 0)$ and $[z]_{-}=\max(-z, 0)$. Let $\mathcal{P}_{\mathcal{B}_{n}}, \mathcal{P}_{\mathcal{B}_{\infty}}:\bR^{K}\to\bR^{K}$ denote the projection operators onto sets $\mathcal{B}_{n}, \mathcal{B}_{\infty}$, respectively. We use $Q_{1-\alpha}$ to denote the empirical $1-\alpha$ quantile of the conformal scores.

Let $a_n$ and $b_n$ be sequences of real-valued random variables or deterministic quantities indexed by $n \in \mathbb{N}$. We use the following asymptotic notation:$a_n = O(b_n)$ means there exists a constant $c> 0$ such that $|a_n| \leq c |b_n|$ for all sufficiently large $n$. $a_n = O_{\bP}(b_n)$ means that for any $\epsilon > 0$, there exists $c_\epsilon > 0$ and $N_\epsilon \in \mathbb{N}$ such that $\mathbb{P}(|a_n| > c_\epsilon |b_n|) < \epsilon$, for all $n \geq N_\epsilon$. We use small $c$ to represent a constant, which may vary line by line.

\subsection{Related Work on Conformal Prediction}\label{sec: background}

In standard split conformal prediction, the data is partitioned into three sets: the training set which is used to train a predictive model $\hat\mu(\cdot)$, the calibration set $\{X_i, Y_i\}_{i\in[n]}$ which is used to calibrate conformity scores, and finally, the test point $X_{n+1}$ with unknown response $Y_{n+1}$. Throughout this paper, we work with split conformal prediction, which generates the prediction interval for $Y_{n+1}$ as:
\begin{align}\label{eq: prediction set overall}
    \hat C(X_{n+1})=\{y: S(X_{n+1}, y)\leq q^*\},
\end{align}
where
$q^*$ is chosen as the $(1-\alpha)$-quantile of the set $\{S_i\}_{i\in [n+1]}$. The resulting prediction set contains all values $y$ for which the conformity score $S(X_{n+1}, y)$ is sufficiently small.

We demonstrate below how the various coverage can be achieved depending on the information available about the predictive model $\hat \mu(\cdot )$.

\paragraph{Marginal Coverage.}
Suppose we know that the predictive model performs equally well across the entire feature space, and the
$(n+1)$-th conformity score is drawn i.i.d. from the same distribution as the first $n$ scores. By the replacement lemma in \citet{angelopoulos2024theoretical}, the prediction set in \eqref{eq: prediction set overall} can be obtained by the threshold $q^0=Q_{1-\alpha}(\sum\nolimits_{i\in[n]}\frac{1}{n+1}\delta_{S_i}+\frac{1}{n+1}\delta_{+\infty})$. It is well known that the set $\hat C^0(X_{n+1})$ given by $q^0$ has marginal validity such that $\mathbb{P}(Y_{n+1} \in \hat C^0(X_{n+1}))\geq 1-\alpha$ \citep{papadopoulos2002inductive}. As an alternative strategy, \citet{gibbs2023conformal} proposed obtaining coverage threshold $q^0$ in \eqref{eq: prediction set overall} using an intercept-only quantile regression within the constant function class $\cF^0$. Let $S$ denote an imputed value for the unknown score $S_{n+1}$ and define the pinball loss for level $\alpha$ as $\ell_\alpha(z)=(1-\alpha)[z]_++\alpha[z]_{-}$. Then they fit
\begin{align}\label{eq: quantile reg on F0}
\hat q_S^{0}:=\arg\min_{q\in \cF^0}\frac{1}{n+1}\sum\nolimits_{i\in[n]}\ell_{\alpha}(S_i-q)+\frac{1}{n+1}\ell_{\alpha}(S-q),
\end{align}
and output the nonrandomized prediction set $ \hat C^{0}(X_{n+1})=\{y: S(X_{n+1}, y)\leq \hat q^{0}_{S(X_{n+1}, y)}\}$. They show that this procedure also satisfies the marginal validity guarantee.

Applying conformal prediction in settings with latent structure is nontrivial. There exist several challenges for conformal prediction with low-rank structure: (1) misspecification of $\hat\mu(\cdot)$ may prevent the latent structure of $X$ from being faithfully reflected in the distribution of $S\mid X$; (2) if the embedding $\hat\pi(\cdot)$ is inaccurate or incomplete so that there are few neighbors near the test point in the embedding space, prediction intervals can become overly conservative or excessively wide; and (3) an inappropriate choice of rank $K$ may undermine the conditional validity.

One prominent approach is Posterior Conformal Prediction (PCP) \citep{zhang2024posterior}, which has been detailed as follows.

\paragraph{Posterior Conformal Prediction.}
\citet{zhang2024posterior} proposed a posterior conformal prediction (PCP) framework under the assumption that $X$ exhibits a latent low-rank structure, and the predictive model $\hat\mu(\cdot)$ is well-specified. Specifically, they assume the conditional distribution of the conformity score $S\mid X$ follows a mixture model: \begin{align*}S_i\mid X_i\sim \sum\nolimits_{k\in[K]} \pi_k(X_i)   \zeta_k,\end{align*}where $ \zeta_1, \dots,  \zeta_{K}$ are distinct probability densities, and $ \pi_k(X_i)$ represent cluster membership probabilities. Adapting ideas from weighted conformal prediction, the prediction set is constructed as: \begin{align*}\hat C^{\text{PCP}}(X_{n+1})=\left\{y: S(X_{n+1}, y)\leq Q_{1-\alpha}\left(\sum\nolimits_{i\in[n]}w_i\delta_{S_i}+w_{n+1}\delta_{+\infty}\right)\right\}.\end{align*} where weights $\{w_i\}_{\in[n+1]}$ are determined by the similarity between latent structures. Let $m \hat \pi\sim \text{Multinomial}(m,  \pi(X_{n+1}))$. In the randomized setting, the weights $w_{i, rand}$ are proportional to $\exp\left\{ - \sum_{k=1}^{K} m\hat \pi_k\cdot \log \frac{\pi_k(X_{n+1})}{\pi_k(X_i)} \right\}$. In the nonrandomized setting, weights are proportional to $\exp\left\{ -m D_{\mathrm{KL}}\left( \pi(X_{n+1}) \,\|\, \pi(X_i) \right) \right\}$. Under the randomized setting, \citet{zhang2024posterior} show that PCP provides conservative conditional coverage guarantees.
\begin{align}\label{eq: pcp coverage}
   1-\alpha\leq \mathbb{P}\left( Y_{n+1} \in \hat{C}^{\text{PCP}}_{rand}(X_{n+1}) \mid  \hat{\pi} \right)\leq 1- \alpha + \mathbb{E} \left[ \max\nolimits_{i \in [n+1]} w_{i,rand} \mid \hat{\pi} \right].
\end{align}

This approach relies on the assumption that the predictive model $\hat\mu(\cdot)$ is well-specified, so that the latent structure of $\mathbf{X}$ can be faithfully reflected in the mixture structure of the conditional distribution of the scores given $X$. When $\hat\mu(\cdot)$ is inaccurate, the scores $S$ can exhibit higher variability, and the distribution of $S\mid X$ may not display a meaningful latent structure.

Instead of assuming latent structure in the noise model $S\mid X$, we directly leverage latent embeddings in the covariates $\mathbf{X}$. By calibrating conformity scores as a function of $\hat\pi(X)$ within an RKHS, rather than assuming their relationships a priori, our method remains robust under model misspecification and provides reliable uncertainty quantification.

\paragraph{Localized Conformal Prediction.}
Another related method is randomly-localized conformal prediction (RLCP) \citep{hore2023conformal}, which aims to capture heterogeneity in the conformity score by adjusting the distribution based on proximity to the test point $X_{n+1}$.
Specifically, LCP assigns higher weights, instead of $1/(n+1)$ for $q^0$ in marginal coverage, to data points closer to the test point $X_{n+1}$. These weights on $\delta_{S_i}$, for instance, are proportional to the kernel distance $\exp(-\gamma\|X_i-X_{n+1}\|^2)$ for a bandwidth parameter $\gamma>0$. While \citet{hore2023conformal} showed LCP achieves marginal validity under a randomization step, increasing the bandwidth parameter $\gamma$ can significantly widen the prediction interval, especially in high-dimensional settings.

To do the low-rank projection, RLCP applies a Gaussian reweighting to conformity scores based on distances in a latent embedding space between the test point and calibration points. This approach relies on carefully chosen embeddings that maximize the mutual information between conformity scores and covariates. When either $\hat\mu(\cdot)$ or $\hat\pi(\cdot)$ is inaccurate or incomplete so that there are few neighbors near the test point in the embedding space, RLCP often produces overly conservative or excessively wide prediction intervals by increasing $\gamma$.

In contrast, our method uses $\lambda$-path adapted to the local calibration density, allowing greater flexibility in sparse regions. This selects $(\gamma,\lambda)$ to leverage the global low-rank structure and produce more stable, calibrated prediction intervals (See Figure \ref{fig:ternary}).

\paragraph{Conditional Conformal.}
Suppose no prior information is available about the covariate shift, unlike the settings discussed in LCP and PCP. In this general setting, let $\psi: \mathcal{X}\times \mathcal{X}\to \bR$ be a positive definite kernel, and let $\cF_{\psi}$ denote the associated RKHS with an inner product $\langle\cdot, \cdot\rangle_{\psi}$ and a norm $\|\cdot \|_{\psi}$.

\citet{gibbs2023conformal} proposed the regularized kernel quantile regression for class $\cF^{RKHS}$ in \eqref{eq:rkhs} with a fixed hyperparameter $\lambda>0$:
\begin{align}\label{eq: CC kernel quantile}
\hat g^{CC}_S:=\arg\min_{g\in \cF^{RKHS}}\frac{1}{n+1}\sum\nolimits_{i\in[n]}\ell_{\alpha}(S_i-g(X_i))+\frac{1}{n+1}\ell_{\alpha}(S-g(X_{n+1}))+\lambda\|g_{\psi}\|_{\psi}^2.
\end{align}
They constructed the nonrandomized prediction set as $\hat C^{cc}(X_{n+1}):=\{y: S(X_{n+1}, y)\leq \hat g^{CC}_{S(X_{n+1}, y)}(X_{n+1})\}$
    \begin{lemma}[Theorem 3 in \citep{gibbs2023conformal}]\label{lemma: cc} Let $\psi: \mathcal{X} \times \mathcal{X}\to \bR$ be a positive definite kernel, and $\Phi:\mathcal{X}\to \mathbb{R}^{d}$ a finite
    dimensional feature map. Consider the RKHS-based function class $\cF^{RKHS}$ associated with $\psi$ and $\Phi$. Assume that $\{(X_{i}, S_i)\}_{i\in[n+1]}$ are exchangeable. Then for all $f\in \cF^{RKHS}$, we have
    \begin{align*}
        \bE\left[f(X_{n+1})\cdot\left(\One\{Y_{n+1}\in\hat C^{CC}(X_{n+1}) \}-(1-\alpha)\right)\right]=-2\lambda\bE\left[\langle\hat g^{CC}_{S_{n+1}, \psi}, f_{\psi}\rangle\right]+|\epsilon_{int}|,
    \end{align*}
    where the interpolation error $\epsilon_{int}$ satisfies $|\epsilon_{int}|\leq \bE\left[f(X_i)\One\{S_i=\hat g^{CC}_{S_{n+1}}(X_i)\}\right]$.
    \end{lemma}
The interpolation term $\epsilon_{int}$ can be removed when randomized prediction sets are used (see Lemma \ref{theorem: oracle mixture}).

Similar to the challenges faced in localized conformal prediction, solving the optimization problem \eqref{eq:quantile reg on RKHS mixture} using a kernel $\psi$ defined over the original high-dimensional feature space can lead to oversmoothing and wider prediction interval. In particular, when $p\gg n$
the RKHS norm $\|g_{\psi}\|_{\psi}$ becomes large unless regularization  $\lambda$ is increased significantly, which in turn flattens the estimated quantile function $\hat g_S(\cdot)$  As a result, the prediction set may have poor local adaptivity, leading to wider intervals and coverage gaps.

\section{Computational Details for SpeedCP}

\subsection{Low-Rank Projection Using Admixture Model}\label{sec: topic modeling}
In this work, we consider high-dimensional covariates ${X}\in \mathbb{R}^p$ with $p \gg n$ and denote their low-rank representation map as $\hat\pi : \mathcal{X} \to \bR^{K}$ with $K \ll p$. A simple choice of $\hat{\pi}(\cdot)$ is principal component analysis (PCA), where $\hat\pi({X})={X}^\top \mathbf{V}$, with $K$ principal directions $\mathbf{V} \in \mathbb{R}^{p\times K}$. Alternatively, probabilistic models such as latent Dirichlet allocation (LDA) \citep{blei2003latent} provide interpretable embeddings, representing each $X$ as a mixture of latent components $\{\zeta_k\}_{k\in[K]}$. In deep learning models, one can also consider applying low-rank projections on layer embeddings. For the simulation experiments and experiments with ArXiv abstracts, we consider the admixture model under the probabilistic Latent Semantic Indexing (pLSI) \citep{hofmann1999probabilistic},
\begin{align}\label{eq: multinomial}
    m X_i\mid W_i=w_i\sim \text{Multinomial}(m, \sum\nolimits_{k\in[K]}  w_i(k)\zeta_k)
\end{align}
where $W_i\in \Delta^{K-1}$ denotes the latent mixture proportions and $\zeta_k$ represents the latent distribution. $m$ denotes the document length. This shows $\mathbb{E}[X_i\mid W_i]=\bzeta^{\top}W_i$. However, this decomposition in general may not be unique, but under the
separability condition \citep{donoho2003does} or anchor word condition \citep{arora2012learning}, $\bzeta$ is identifiable.

When applying RKHS methods to compositional data such as mixture proportions $\hat \pi(X)$, it is essential to first transform the simplex into Euclidean space.
If we perform kernel regression or smoothing over $\hat\pi$ directly, the output might be outside the simplex. Suppose $\hat\pi(X_i)$ lies in the open simplex such that all entries are positive, then the log-ratio transformation (such as additive, centered, and isometric log-ratio transformations)\citep{aitchison1982statistical} can be used.
\paragraph{Centered Log-Ratio Transformation (clr).}
If $ \hat\pi_k(X_i) > 0$ for all $i, k$,
\begin{align*}
\hat\theta_{ik}:=
 \log \hat\pi_k(X_i) - \frac{1}{K} \sum_{j \in [K]} \log \hat\pi_j(X_i)\end{align*}

Given this transformation, we define the kernel similarity between points as:
\begin{align*}
&d_{\pi}(X_i, X_j):= \|\hat\theta_i - \hat\theta_j\|^2; \qquad
\psi^*(X_i, X_j) = \exp\left\{ -\gamma \|\hat\theta_i - \hat\theta_j\|^2   \right\}
\end{align*}

\paragraph{pLSI Using SVD.} Let $\mathbf{X}:=\mathbf{X}_{train}\cup \mathbf{X}_{calib}\cup \mathbf{X}_{test}\in \bR^{n_{all}\times p}$. Here, we present one of the algorithms used to estimate the latent embeddings $\bpi:=\bpi(\mathbf{X})=\mathbb{E}[\mathbf{W}\mid \mathbf{X}]$ from $\mathbf{X}$. When $m\to \infty$, the posterior mean $\mathbb{E}[W_i\mid {X}_i]$ concentrates around the true mixture proportion $w_i$. If $\bpi$ and $\bzeta$ are full-rank matrices and the $K$-th largest singular value satisfies $\lambda_K(\bpi\bzeta^{\top})>0$, we start with the singular value decomposition of the matrix $ \bpi \bzeta ^{\top}$:
\begin{align*}
    \bpi \bzeta ^{\top}=\boldsymbol{\Xi}\boldsymbol{\Lambda} \mathbf{V}^\top\implies   \boldsymbol{\Xi}=\bpi \bzeta ^{\top}\mathbf{V}\boldsymbol{\Lambda}^{-1}:=\bpi \mathbf{H}
\end{align*}
with some nonsingular matrix $H$. Notice that each row of $\bpi \in \bR^{n_{all}\times K}$ is a probability vector (i.e., nonnegative and sums to 1). Given this simplex structure, we can recover the matrix $\boldsymbol{H}$ from $\boldsymbol{\Xi}$ using nonnegative matrix factorization techniques. In particular, methods such as the \textit{Successive Projection Algorithm} (SPA) \citep{araujo2001successive, gillis2013fast} and \textit{Archetypal Analysis} \citep{javadi2020nonnegative} are effective in recovering the extreme points (vertices) of the convex hull.
\begin{algorithm}[ht]
\setstretch{1.35}
\caption{pLSI using SVD \citep{klopp2021assigning}} \label{algo:topic modeling}
\begin{algorithmic}
\STATE \textbf{Input:} $\mathbf{X}\in \bR^{n_{all}\times p}$, latent dimension $K$
\STATE \textbf{Output:} $\widehat{\bpi}_{train},\widehat{\bpi}_{calib},\widehat{\bpi}_{test} = \hat{\pi}(\mathbf{X},K)$
\STATE 1. Get the rank-$K$ SVD of $\mathbf{X}=\widehat{\boldsymbol{\Xi}}\widehat{\boldsymbol{\Lambda}}\widehat{\mathbf V}^\top$

\STATE 2. (Vertex hunting algorithm) Apply the vertex hunting algorithm on the rows of $\widehat{\boldsymbol{\Xi}}$ to get the vertices $\widehat{\mathbf H}$

\STATE 3. Set $\hat\pi(\mathbf{X})=\widehat{\boldsymbol{\Xi}} \widehat{\mathbf H}^{-1}$ and thus $\hat \pi(X_i)=(\widehat{\mathbf H}^{-1})^\top \widehat{\boldsymbol{\Xi}}_{i\cdot}$
\end{algorithmic}
\end{algorithm}

\subsection{\texorpdfstring{Derivation of $\lambda$-path and $S$-path}{Derivation of lambda-path and S-path}}\label{sec: sqkr}

In this section, we provide technical details on our path-tracing approaches of $\lambda$ and $S$. Recall the RKHS function class $\mathcal F^*$,
\begin{equation}\label{eq:fstar}
\begin{aligned} & \cF^{*}=\left\{f_{\psi^*}(\cdot)+\Phi^{*}(\cdot)^\top \eta:f_{\psi^*} \in \cF_{\psi^*}, \eta\in \bR^{d} \right\},
\end{aligned}\end{equation}
we begin with some preliminaries.

Denote $S_i=S(X_i, Y_i)$ as the score of the $i^{th}$ point in the calibration set for $i \in [n]$ and $S_{n+1}$ as the score of a test point. To decide the score cutoff we use for a prediction set, we proceed to fit a RKHS quantile regression on $n$ calibration points together with the test point. Since the true score of the test point, $S_{n+1}$ is unknown, we set the score of the test point, $S_{n+1}$, as an arbitrary value $S$. Let $\alpha\in (0,1)$ be a user-specified miscoverage level. The objective then becomes,
\begin{equation}\label{eq: opt obj1}
\hat g_S=\arg\min_{g\in \cF^{*}}\frac{1}{n+1}\sum_{i\in[n]}\ell_{\alpha}(S_i-g(X_i))+\frac{1}{n+1}\ell_{\alpha}(S-g(X_{n+1}))+\frac{\lambda}{2}\|g_{\psi^*}\|_{\psi^*}^2,
\end{equation}
with the known solution in finite form:
\begin{equation}
\label{eq:fit}
    \hat g_S(X) =\Phi^{*}(X)^\top\hat\eta_{S}+\frac{1}{\lambda}\sum_{i=1}^{n+1}\hat \upsilon_{ S, i}\psi^*(X,X_i),
\end{equation}
We define $\Phi^{*}(X) \in \mathbb{R}^{d}$ as any feature representation of $X$ and $\eta_{S,j}$ as the coefficient of $\Phi^{*}(X)_j$, $j\in [d]$. Plugging this into the objective, it becomes,
\begin{equation*}
    \min_{\eta_{S}, \upsilon_{S}} \sum_{i=1}^{n+1}\ell_{\alpha}\big(S_i-\Phi^{*}(X_i)^\top\eta_{S}-\frac{1}{\lambda}\sum_{i'=1}^{n+1} \upsilon_{S,i'}\psi^*(X_i,X_{i'})\big)+\frac{1}{2\lambda}\sum_{i, i'=1}^{n+1}\upsilon_{ S, i}\upsilon_{S, i'}\psi^*(X_i,X_{i'}).
\end{equation*}
 with the Lagrangian primal function as
 \begin{equation}\label{eq: lagrangian}
     \begin{split}
         L_p &= (1-\alpha)\sum_{i=1}^{n+1}p_i+\alpha\sum_{i=1}^{n+1}q_i+\frac{1}{2\lambda}\upsilon^{\top}_S \mathbf \Psi^* \upsilon_S\\
         &\qquad +\sum_{i=1}^{n+1} \sigma_i(S_i-g_S(X_i)-p_i)-\sum_{i=1}^{n+1}\tau_i(S_i-g_S(X_i)+q_i)\\
         &\qquad -\sum_{i=1}^{n+1}\kappa_ip_i-\sum_{i=1}^{n+1}\rho_iq_i,
     \end{split}
 \end{equation}
and $\sigma, \tau, \kappa,\rho$ are nonnegative Lagrangian multipliers. Here, $\Psi^*\in \mathbb{R}^{(n+1)\times(n+1)}$ denotes the kernel matrix where its $(i,i')$ element denotes $\psi^*(X_i, X_{i'})$. Setting the derivatives of $L_p$ at 0,
\begin{equation}\label{eq:deriv app1}
    \begin{split}
        \frac{\partial L_p}{\partial\upsilon_{S,i}} &: \upsilon_{S,i}=\sigma_i-\tau_i\\
        \frac{\partial L_p}{\partial\eta_{S,j}} &: \sum_{i=1}^{n+1}\sigma_i\Phi^*(X_i)_j=\sum_{i=1}^{n+1}\tau_i\Phi^*(X_i)_j, \ \ j \in [d]\\
        \frac{\partial L_p}{\partial p_i} &: \sigma_i=1-\alpha-\kappa_i\\
        \frac{\partial L_p}{\partial q_i} &: \tau_i=\alpha-\rho_i.\\
    \end{split}
\end{equation}

The Karush–Kuhn–Tucker (KKT) conditions give
\begin{equation}\label{eq:kkt1}
    \begin{split}
        \sigma_i(S_i-g_S(X_i)-p_i)&=0\\
        \tau_i(S_i-g_S(X_i)+q_i)&=0 \\
        \kappa_ip_i&=0\\
        \rho_iq_i&=0
    \end{split}
\end{equation}
Since Lagrangian multipliers are nonnegative, $0\leq \sigma_i\leq 1-\alpha$ and $0\leq\tau_i\leq \alpha$, combining \eqref{eq:deriv app1} and \eqref{eq:kkt1}, we can easily see that,
\begin{equation}\label{eq:before_elbow}
    \begin{split}
        S_i-g_S(X_i)>0 \ \ &\Rightarrow \ \ \  p_i>0,\ \kappa_i=0, \ \sigma_i=1-\alpha, \ \tau_i=0 \ \ \Rightarrow \ \ \upsilon_{S,i} =1-\alpha\\
        S_i-g_S(X_i)<0 \ \ &\Rightarrow \ \ \   q_i>0, \ \rho_i=0, \ \tau_i=\alpha, \ \sigma_i = 0 \ \ \Rightarrow \ \ \upsilon_{S,i} =- \alpha\\
        S_i-g_S(X_i)=0 \ \ &\Rightarrow \ \  p_i=q_i=0, \ \sigma_i\in(0,1-\alpha], \ \tau_i\in(0,\alpha] \ \ \Rightarrow \ \ \upsilon_{S,i} \in(-\alpha, 1-\alpha)\\
    \end{split}
\end{equation}

With $\widehat r_{S,i}:=S_i-\hat g_{S}(X_i)$,
the KKT conditions induce three index sets:
\begin{align} \label{eq:active_sets}
  E&:=\bigl\{i:\;\widehat r_{S,i}=0,\;\hat\upsilon_{S,i}\in(-\alpha, 1-\alpha)\bigr\},\\
  L&:=\bigl\{i:\;\widehat r_{S,i}<0,\;  \hat\upsilon_{S,i}=-\alpha\bigr\},\\
  R&:=\bigl\{i:\;\widehat r_{S,i}>0,\;  \hat\upsilon_{S,i}=1-\alpha\bigr\}.
\end{align}

\subsection{\texorpdfstring{Derivation of $\lambda$-path}{Derivation of lambda-path}}

We use $\lambda$-path to tune the regularization (or smoothness) parameter $\lambda$, which we combine with cross validation on the kernel bandwidth $\gamma$ to determine the optimal hyperparameter pair. The same \eqref{eq: opt obj1}-\eqref{eq:active_sets} hold, but the RKHS quantile regression is now estimated with { m separate points}. The motivation for this is to fix the hyperparameters before constructing prediction sets, which is necessary for our theoretical guarantees. The index sets $(E, L, R)$ evolve with different $\lambda$ values. We denote them as $(E(\lambda), L(\lambda), R(\lambda))$.

We start with a sufficiently large initial value $\lambda^1$ and decrease it toward 0. As $\lambda$ decreases, data points move from the left of the elbow, stay in the elbow, then move to the right of the elbow (or vice versa). Any change in the elbow set is denoted as an ``event". The next $\lambda$ is updated as the largest value where such event occurs. At each update, we calculate $\hat{\upsilon}_{i}$ for the points in $E(\lambda)$ since $\{\hat{\upsilon}_{i'}\}_{i'\in[m]}$ in $L(\lambda), R(\lambda)$ are fixed.

{ In this section, we assume the columns of the $m$-point projection matrix $\mathbf\Phi^{*}\in \mathbb{R}^{m \times d}$ are linearly independent. Denote $\mathbf \Phi^{*}_A$ as a submatrix of $\mathbf \Phi^{*}$ whose row indices are in  set $A$. Also denote $\mathbf{\Psi}^{*}_{AB}$ as a submatrix of $\mathbf{\Psi}^{*}\in \mathbb{R}^{m\times m}$ whose row indices are in set $A$ and column indices are in set $B$.}
\subsubsection{Proof of Proposition \ref{prop:lambda_path}}

We now prove Proposition \ref{prop:lambda_path}, which states an affine relationship of $\hat{\upsilon}_{i'}(\lambda)$'s and $\hat{\eta}(\lambda)$ on $\lambda$ between two change points of $\lambda$. If $\hat{\upsilon}_{i'}(\lambda)$'s and $\hat{\eta}(\lambda)$ are affine in $\lambda$ between any change points, then they are piecewise-linear in $\lambda$, which makes the solution path tractable for any $\lambda\leq \lambda^1$. We provide a more detailed version of the proof in Section~\ref{sec:prop_S}, which has identical steps as Proposition~\ref{prop:lambda_path}.

\begin{proof}
Let $\{\lambda^l\}_{l=1,2,3,\cdots}$ be the change points when an event occurs. Consider an interval $\lambda^{l+1} \leq \lambda \leq \lambda^l$ during which the sets stay the same, i.e., $(E(\lambda), L(\lambda), R(\lambda))=(E(\lambda^l), L(\lambda^l), R(\lambda^l))$. Denote $\hat{\upsilon}_{i'}(\lambda)$ and $\hat{\eta}(\lambda)$ as the solution of \eqref{eq: opt obj1} given $\lambda$. In this proof, denote $E=E(\lambda)=E(\lambda^l)$, $L=L(\lambda)=L(\lambda^l)$, and $R=R(\lambda)=R(\lambda^l)$. Define two quantities.
\[
d_E:=\frac{1}{\lambda}((-\alpha)\,\mathbf{\Psi}^{*}_{EL}\mathbf 1_L+(1-\alpha)\,\mathbf{\Psi}^{*}_{ER}\mathbf 1_R),
\qquad
\mathbf\Pi_E:=I_{|E|}-\mathbf\Phi^{*}_E\bigl(\mathbf\Phi_E^{*\top}\mathbf\Phi^{*}_E\bigr)^{-1}\mathbf\Phi_E^{*\top}.
\]
Let
\(S_E:=(S_{i'})_{i'\in E},
  d_E:=(d_{i'})_{i'\in E},
  \mathbf\Phi^{*}_E\in\mathbb R^{|E|\times d},
  \mathbf{\Psi}^{*}_{EE}\in\mathbb R^{|E|\times|E|}\). By the definition of the elbow set combined with \eqref{eq:fit},
\begin{equation}
  S_E=\mathbf\Phi^{*}_E\,\hat{\eta}(\lambda)+\frac{1}{\lambda}\,\mathbf{\Psi}^{*}_{EE}\,\hat{\upsilon}_{E}(\lambda)
       +d_E. \label{eq:lambda_equalities}
\end{equation}
Projecting with $\mathbf\Pi_E$ eliminates $\hat{\eta}(\lambda)$,
\begin{equation}
  \mathbf\Pi_E\mathbf{\Psi}^{*}_{EE}\,\hat{\upsilon}_{E}(\lambda)
  \;=\;\lambda\,\mathbf\Pi_E (S_E-d_E). \label{eq:proj_lambda}
\end{equation}
Moreover, the second KKT constraint in \eqref{eq:deriv app1} gives $\mathbf\Phi^{*\top}\hat{\upsilon}  = 0$. This is equivalent to,
$$ \mathbf\Phi_E^{*\top} \hat{\upsilon}_{E}(\lambda ) = \alpha \mathbf\Phi_L^{*\top} \mathbf{1}_L - (1-\alpha) \mathbf\Phi_R^{*\top} \mathbf{1}_R. $$
Define $\mathbf A:=\mathbf\Pi_E\mathbf{\Psi}^{*}_{EE}\mathbf\Pi_E$. Using its Moore–Penrose inverse (denoted by superscript \(\dagger\)),
\begin{equation}\label{eq:proj_lambda2}
\begin{split}
       \mathbf\Pi_E \hat{\upsilon}_{E}(\lambda)
     &=
    \lambda
    \mathbf A^{\dagger}
    \,\mathbf\Pi_E
    \bigl(S_E-d_E \bigr)
  \\
  &- \alpha\, \mathbf A^{\dagger}\mathbf\Pi_E\mathbf \Psi^{*}_{EE}  \mathbf \Phi_E^{*} (\mathbf\Phi_E ^{*\top} \mathbf \Phi^{*}_E)^{\dagger} \mathbf\Phi_L^{*\top}\mathbf 1_L\\
  &
  +(1-\alpha)\,\mathbf A^{\dagger}\mathbf \Pi_E\mathbf \Psi^{*}_{EE}  \mathbf \Phi_E^{*} (\mathbf\Phi_E ^{*\top} \mathbf \Phi^{*}_E)^{\dagger}\mathbf\Phi_R^{*\top}\mathbf 1_R.
   \end{split}
\end{equation}
 Thus, the minimum–norm solution on $\mathrm{Im}(\,\mathbf \Pi_E)$ is,
\begin{equation}\label{eq:v_affine_lambda}
\begin{split}
    \hat{\upsilon}_{E}(\lambda)
  &=\lambda \mathbf A^\dagger\mathbf \Pi_E\bigl(S_E -d_E\bigr)\\
  &+[I_{|E|}-\mathbf A^\dagger\mathbf \Pi_E\mathbf \Psi^{*}_{EE}]\mathbf \Phi^{*}_E\bigl(\mathbf \Phi_E^{*\top}\mathbf \Phi^{*}_E\bigr)^{\dagger}
    \bigl[\alpha\,\mathbf \Phi_L^{*\top}\mathbf 1_L-(1-\alpha)\,\mathbf \Phi_R^{*\top}\mathbf 1_R\bigr].
\end{split}
\end{equation}
Thus, $\hat{\upsilon}_{E}(\lambda)$ is \emph{affine in $\lambda$} on the interval.
From \eqref{eq:lambda_equalities},
\begin{equation}
  \hat{\eta}(\lambda)
  =(\mathbf \Phi_E^{*\top}\mathbf \Phi^{*}_E)^{-1}\mathbf \Phi_E^{*\top}
   \Bigl[S_E-d_E-\frac{1}{\lambda}\,\mathbf \Psi^*_{EE}\,\hat{\upsilon}_{E}(\lambda)\Bigr],
  \label{eq:eta_affine_invlam}
\end{equation}
hence $\hat{\eta}(\lambda)$ is \emph{affine in $1/\lambda$}. We have shown that,

\begin{equation}\label{affine_lambda}
    \hat{\upsilon}_{E}(\lambda) = a+\lambda a, \qquad \hat{\eta}(\lambda) = a^{(1)}+\frac{a^{(1)}}{\lambda}.
\end{equation}

with \( a, b \in\mathbb R^{|E|},  a^{(1)}, b^{(1)}\in\mathbb R^{d}\) constant on the segment. For $i\in L(\lambda), R(\lambda)$, $\hat{\upsilon}_{i'}$ is constant, making it affine in $\lambda$ as well. Finally, for any $i' \in [m]$,
\begin{equation}
\begin{split}
    \hat{g}(X_{i'})
 &= \mathbf\Phi^{*}_{i'\cdot}(a^{(1)}+\frac{a^{(1)}}{\lambda}) + \frac{1}{\lambda}\mathbf\Psi^*_{i',E}(a+\lambda a)+d_{i'}\\
 &= \frac{1}{\lambda}(\mathbf\Phi^{*}_{i'\cdot}a^{(1)}+\mathbf\Psi^*_{i',E}a) +\mathbf\Phi^{*}_{i'\cdot}a^{(1)} + \mathbf\Psi^*_{i',E}a+d_{i'},
\end{split}
\end{equation}

which makes the residual $r_{i'}(\lambda)=S_{i'}-\hat{g}(X_{i'})$ affine again in $1/\lambda$ for $i'\in[m]$ on the interval.
\end{proof}

\subsubsection{\texorpdfstring{Update of $\lambda^l$}{Update of lambdal}}

Let $\lambda^l$ denote the value after the $l^{th}$ event. The elbow set $E(\lambda^l)$ is updated when one of the following events occurs,
\begin{itemize}
    \item A point $i$ in either $L(\lambda^l)$ or $R(\lambda^l)$ enters the elbow set (residual $S_i - \hat{g} (X_i)$ becomes 0).

\item A point $i$ in $E(\lambda^l)$ leaves to the left or right set ($\hat{\upsilon}_{i}(\lambda), i \in E(\lambda^l)$ becomes $-\alpha$ or $1-\alpha$).
\end{itemize}

We take $\lambda^{l+1}$ as the largest $\lambda \leq \lambda^{l}$ that triggers one of the events and update $(E, L, R)$ accordingly. Here, let $E=E(\lambda)=E(\lambda^l)$. Denote the linear parameter $\hat{\eta}^{\lambda}_{j}=\lambda\hat{\eta}_{j}(\lambda)$ for $j \in [d]$. From \eqref{eq:fit}, for $\lambda^{l+1}\leq \lambda \leq \lambda^{l}$, the fit at $\lambda$ is,
 \begin{equation*}
 \begin{split}
     \hat{g}(X_{i'}) &= \mathbf\Phi^{*}_{i'\cdot}\hat{\eta}(\lambda)+\frac{1}{\lambda}\mathbf\Psi^*_{i',\cdot}\hat{\upsilon}(\lambda) \\
     &= \mathbf\Phi^{*}_{i'\cdot}\hat{\eta}(\lambda)+\frac{1}{\lambda}\big(\mathbf\Psi^*_{i',E}\hat{\upsilon}_E(\lambda)-\alpha\mathbf\Psi^*_{i',L}\mathbf{1}_L+(1-\alpha)\mathbf\Psi^*_{i',R}\mathbf{1}_R\big)\\
     &= \frac{1}{\lambda}\big(\mathbf\Phi^{*}_{i'\cdot}\hat{\eta}^{\lambda}+\mathbf\Psi^*_{i',E}\hat{\upsilon}_E(\lambda)+d_{i'}\big) ,
     \end{split}
 \end{equation*}

where
 \begin{equation*}
      d_{i'}:= -\alpha\mathbf\Psi^*_{i',L}\mathbf{1}_L+(1-\alpha)\mathbf\Psi^*_{i',R}\mathbf{1}_R.
 \end{equation*}

Let $\hat{g}^l(X_{i'})$ be the estimated function with $\lambda^l$. Now, we can express $\hat{g}(X_{i'})$ with $\lambda^l$ and $\hat{g}^l(X_{i'})$,
 \begin{equation}\label{eq:fit_l}
     \begin{split}
         \hat{g}(X_{i'}) &=\hat{g}(X_{i'})-\frac{\lambda^l}{\lambda}\hat{g}^{l}(X_{i'})+\frac{\lambda^l}{\lambda}\hat{g}^{l}(X_{i'})\\
         &=\frac{1}{\lambda}\big[\mathbf\Phi_{i'\cdot}^{*}(\hat{\eta}^{\lambda}-\hat{\eta}^{\lambda^l}) +\mathbf\Psi^*_{i',E}(\hat{\upsilon}(\lambda)-\hat{\upsilon}^l(\lambda))+d_{i'}-d_{i'}+\lambda^l\hat{g}^l(X_{i'})\big]\\
         &=\frac{1}{\lambda}\big[\mathbf\Phi_{i'\cdot}^{*}(\hat{\eta}^{\lambda}-\hat{\eta}^{\lambda^l}) +\mathbf\Psi^*_{i',E}(\hat{\upsilon}(\lambda)-\hat{\upsilon}^l(\lambda))+\lambda^l\hat{g}^l(X_{i'})\big].
     \end{split}
 \end{equation}

 Recall from the second KKT condition \eqref{eq:deriv app1}, we have $\upsilon_{i'}=\sigma_{i'}-\tau_{i'}$ and $\sum_{i'=1}^{m}(\sigma_{i'}-\tau_i)\mathbf\Phi^{*}_{i',j}=\sum_{i'=1}^{m}\upsilon_{i'}\mathbf\Phi^{*}_{i',j}=0$ for $j=1,\cdots,d$.

 Component-wise,

 \begin{equation*}
     \mathbf\Phi_{E}^{*\top}\hat{\upsilon}_E(\lambda)
   -\alpha \mathbf\Phi_{L}^{*\top}\mathbf{1}_L
   + (1-\alpha)\mathbf\Phi_{R}^{*\top}\mathbf{1}_R=0,
 \end{equation*}
 and

 \begin{equation*}
      \mathbf\Phi_{E}^{*\top}\hat{\upsilon}_E(\lambda^l)
   -\alpha \mathbf\Phi_{L}^{*\top}\mathbf{1}_L
   + (1-\alpha)\mathbf\Phi_{R}^{*\top}\mathbf{1}_R=0,
 \end{equation*}

 leading to

 \begin{equation}\label{eq:kkt_constraint}
    \mathbf\Phi_{E}^{*\top}(\hat{\upsilon}_E(\lambda)-\hat{\upsilon}_E(\lambda^l))=0.
 \end{equation}

Denote $\bar{\upsilon}_{i'}=\hat{\upsilon}_{i'}(\lambda)-\hat{\upsilon}_{i'}(\lambda^l)$ for $i' \in E$ and $\bar{\eta}_j=\hat{\eta}^{\lambda}_{j}-\hat{\eta}^{\lambda^l}_{j}$ for $j \in [d]$. For any $i \in E^l$, $\hat{g}(X_{i})=S_{i}$. Let $S_{E}$ be the stacked scores for $E$. Then, \eqref{eq:fit_l} becomes,
 \begin{equation*}
  \mathbf\Phi^{*}_{E}\bar{\mathbf{\eta}}+ \mathbf \Psi^*_{EE}\bar{\upsilon}=(\lambda-\lambda^l)S_{E}
 \end{equation*}
 Combining with \eqref{eq:kkt_constraint} and representing in a matrix form,
 \begin{equation*}
     \begin{split}
         \begin{pmatrix}
 \mathbf\Phi^{*}_{E} & \mathbf \Psi^*_{EE} \\
 \mathbf{0} & \mathbf\Phi_{E}^{*\top}\\
 \end{pmatrix} \begin{pmatrix}
 \bar{\eta} \\ \bar{\upsilon}
 \end{pmatrix} &=(\lambda-\lambda^l)\begin{pmatrix}
 S_E \\\mathbf{0}
 \end{pmatrix}\\
 \mathbf A^l\mathbf{\beta}&=(\lambda-\lambda^l)\mathbf{S}_0\\
 a &=(\mathbf A^l)^{-1}\mathbf{S}_0,
     \end{split}
 \end{equation*}

where $a=\beta/(\lambda-\lambda^l)$. Let $a_u=\bar{\eta}/(\lambda-\lambda^l)$ and $a_v=\bar{\upsilon}/(\lambda-\lambda^l)$. Plugging $a_u, a_v$ back to \eqref{eq:fit_l}, we reexpress the estimated function as a function of $a$,
\begin{equation*}
    \hat{g}(X_{i'}) = \frac{\lambda^l}{\lambda}[\hat{g}^l(X_{i'})-h^l(X_{i'})] +h^l(X_{i'})
\end{equation*}
where
\begin{equation*}
    h^l(X_{i'}) = \mathbf\Phi^{*}_{i'\cdot}a_u+\mathbf\Psi^*_{i',E}a_{v}
\end{equation*}
for $i' \in E$. Finally, to decide $\lambda^{l+1}$, we choose which event (whether a point enters or exits the elbow set). The first event will happen for $\lambda$ such that a point in $L(\lambda^l)$ or $R(\lambda^l)$ set satisfies $\hat{g}(X_{i'})=S_{i'}$, leading to,
\begin{equation*}
    \lambda^{l+1, hit} = \max_{i'\in L(\lambda^l),R(\lambda^l)}\lambda^{l}\frac{\hat{g}^l(X_{i'})-h^l(X_{i'})}{S_{i'}-h^l(X_{i'})}\mathbf{1}\Big\{\frac{\hat{g}^l(X_{i'})-h^l(X_{i'})}{S_{i'}-h^l(X_{i'})}\leq 1\Big\}.
\end{equation*}

Here, the indicator is to ensure that the updated $\lambda$ is smaller than $\lambda^l$ so that the path is monotonically decreasing. To find $\lambda$ such that a point leaves $E(\lambda^l)$,
\begin{equation*}
    \lambda^{l+1,leave}=\lambda^l+\max_{i' \in E(\lambda^l)}\{x \in\{\frac{-\alpha-\hat{\upsilon}_{i'}(\lambda^l)}{a_{v,i'}},\frac{1-\alpha-\hat{\upsilon}_{i'}(\lambda^l)}{a_{v,i'}}\} \ | \ x \leq 0\big\}
\end{equation*}

We then take $\lambda^{l+1} = \max\big\{\lambda^{l+1,hit}, \lambda^{l+1,leave}\big\}$. We also update $(E,L,R)=(E(\lambda^{l+1}), L(\lambda^{l+1}), R(\lambda^{l+1}))$ accordingly based on which event occurred. Finally, parameters $\hat{\upsilon}_{i'}(\lambda^{l+1})$'s, $\hat{\eta}(\lambda^{l+1})$ can be updated by solving for the new elbow,

\begin{equation}
     \begin{pmatrix}
         \frac{1}{\lambda} \mathbf\Psi^{\star}_{EE}  & \mathbf\Phi^{*}_{E} \\
        \mathbf\Phi_{E}^{*\top} &0
     \end{pmatrix}  \begin{pmatrix}
         \upsilon \\ \eta
     \end{pmatrix}=  \begin{pmatrix}
         S_{E}  - \frac{1}{\lambda}( -\alpha \mathbf\Psi^{\star}_{EL}\mathbf{1}_{L}+(1-\alpha) \mathbf\Psi^{\star}_{ER}\mathbf{1}_{R})\\
         \alpha \mathbf \Phi_L^{*\top}\mathbf{1}_L-(1-\alpha)\mathbf \Phi_R^{*\top}\mathbf{1}_R
     \end{pmatrix}
 \end{equation}

\subsubsection{\texorpdfstring{Initialization of $\lambda$}{Initialization of lambda}}

We describe our strategy for selecting a sufficiently large  initial value $\lambda^1$.
At $\lambda^0=\infty$, from \eqref{eq:fit}, we can see that $\hat{g}(X_{i'}) = \mathbf\Phi_{i',\cdot}^*\hat{\eta}$. In this case, we have only one point in the elbow, which we denote as $i^0$, that satisfies $S_{i^0}=\hat{g}(X_{i^0})=\mathbf\Phi^{*}_{i^0\cdot}\hat{\eta}$. We choose $i^0$ as the $(1-\alpha)$th quantile of scores, i.e. $S_{i^0} = S_{\left\lceil (m)(1-\alpha) \right\rceil}$. Then, points that satisfy $S_{i'} < S_{i^0}$ are in $L(\lambda^0)$, and points such that $S_{i'}> S_{i^0}$ are in $R(\lambda^0)$.

To make the parameters identifiable, we set $\hat{\eta}_{j^*}(\lambda^0)=S_{i^0}/\mathbf\Phi^{*}_{i^0,j^*}$ for one $j^*\in [d]$ and set other parameters $\hat{\eta}_j(\lambda^0), \ j\neq j^*$ to 0. When $\mathbf\Phi^{*}_{i^0\cdot}$ is one-hot encoded, $j^*$ is any index such that $\mathbf\Phi^{*}_{i^0,j^*}=1$. If $\mathbf\Phi^{*}_{i^0\cdot}$ is continuous, we choose $j^*$ to be any arbitrary index. From \eqref{eq:deriv app1}, we have the condition  $\sum_{i'=1}^{m}\hat{\upsilon}_{i'}\mathbf\Phi^{*}_{i',j}=0$ for $j \in [d]$. Since $i^0$ is the only point in $E(\lambda^0)$. This leads to,

\begin{equation}\label{eq:infty_upsilon}
    \hat{\upsilon}_{i^0}(\lambda^0) = \frac{\alpha\sum_{i \in L(\lambda^0), R(\lambda^0)} \mathbf\Phi^{*}_{i,j^*}-\sum_{i \in R(\lambda^0)} \mathbf\Phi^{*}_{i,j^*}}{\mathbf\Phi^{*}_{i^0,j^*}}
\end{equation}

Next, we find the next $\lambda^1$, which will be the initial value of our solution path. This will be the largest $\lambda<\infty$ such that another point from either $L(\lambda^0), R(\lambda^0)$ enters the elbow. Let $i^1$ be the new point entering the elbow. Then, $i^1$ satisfies,
\begin{equation*}
\begin{split}
     S_{i^1} &= \mathbf\Phi^{*}_{i^1,j^*}\hat{\eta}_{j^*}(\lambda^0)+\frac{1}{\lambda^1}\Big(\mathbf\Psi^*_{i^1,i^0}\hat{\upsilon}_{i^0}(\lambda^0)-\alpha\mathbf\Psi^*_{i^1, L(\lambda^0)}\mathbf{1}_{L(\lambda^0)}+(1-\alpha) \mathbf\Psi^*_{i^1, R(\lambda^0)}\mathbf{1}_{R(\lambda^0)}\Big)\\
     &=  \mathbf\Phi^{*}_{i^1,j^*}\hat{\eta}_{j^*}(\lambda^0)+\frac{1}{\lambda^1}f(X_{i^1})
\end{split}
\end{equation*}

Since $i^0$ is still in the elbow set, it should also satisfy,
\begin{equation*}
\begin{split}
     S_{i^0} &= \mathbf\Phi^{*}_{i^0,j^*}\hat{\eta}_{j^*}(\lambda^0)+\frac{1}{\lambda^1}\Big(\mathbf\Psi^*_{i^0,i^0}\hat{\upsilon}_{i^0}(\lambda^0)-\alpha\mathbf\Psi^*_{i^0, L(\lambda^0)}\mathbf{1}_{L(\lambda^0)}+(1-\alpha) \mathbf\Psi^*_{i^0, R(\lambda^0)}\mathbf{1}_{R(\lambda^0)}\Big)\\
     &=  \mathbf\Phi^{*}_{i^0,j^*}\hat{\eta}_{j^*}(\lambda^0)+\frac{1}{\lambda^1}f(X_{i^0})
\end{split}
\end{equation*}
Putting it all together, we can choose $\lambda^1$ as,

\begin{equation}\label{eq:init_lambda}
    \lambda^1 = \max_{i' \neq i^0, i'\in [m]} \frac{f(X_{i'})-(\mathbf\Phi^{*}_{i',j^*}/\mathbf\Phi^{*}_{i^0,j^*})f(X_{i^0})}{S_{i'}-(\mathbf\Phi^{*}_{i',j^*}/\mathbf \Phi^{*}_{i^0,j^*})S_{i^0}}
\end{equation}
and the corresponding $i'$ that maximizes \eqref{eq:init_lambda} becomes $i^1$. We proceed with the same $\hat{\upsilon}(\lambda^0)$, $\hat{\eta}(\lambda^0)$ as our initial parameters and our initial elbow set as $E(\lambda^1) =\{i^0, i^1\}$.

\subsection{\texorpdfstring{Derivation of $S$-path}{Derivation of S-path}}

We fix the hyperparameters $\hat{\gamma}, \hat{\lambda}$ selected by the $\lambda$-path. Conceptually, the $S$-path mirrors the $\lambda$-path, and the conditions \ref{eq: opt obj1}–\ref{eq:active_sets} apply. Now recall the prediction set we defined for a test point $X_{n+1}$,

\begin{align*}
    \hat{C}^*(X_{n+1}) = \{y: S(X_{n+1},y) \leq \hat g_{S(X_{n+1},y)}(X_{n+1})\}.
\end{align*}

By \eqref{eq:active_sets}, this is equivalent to,

\begin{align*}
    \hat{C}^*(X_{n+1}) = \{y: \hat{\upsilon}_{S(X_{n+1},y),n+1} < 1-\alpha\}.
\end{align*}

The problem reduces to finding the largest test score $S^*(X_{n+1})$ such that $\hat{\upsilon}_{S^*(X_{n+1}),n+1} < 1-\alpha$. By Proposition~\ref{prop: nondecreasing}, the mapping $S \mapsto \hat{\upsilon}_{S}$ is monotone, which allows us to recover the prediction set as,

\begin{align*}
    \hat{C}^*(X_{n+1}) = \{y:  S(X_{n+1},y)\leq S^*(X_{n+1})\}.
\end{align*}

It remains to find the maximum $S^*(X_{n+1})$, the test score cutoff, such that $\hat{\upsilon}_{S^*(X_{n+1}),n+1} < 1-\alpha$ holds, i.e., $ S^{*}(X_{n+1}) = \sup\{S\mid \hat \upsilon_{S,n+1}< 1-\alpha\}$ which is the role of $S$-path. Denote the index sets,

\begin{align}
  E(S)&:=\bigl\{i:\;\widehat r_{S,i}=0,\;\hat\upsilon_{S,i}\in(-\alpha,1-\alpha)\bigr\},\\
  L(S)&:=\bigl\{i:\;\widehat r_{S,i}<0,\;  \hat\upsilon_{S,i}=-\alpha\bigr\},\\
  R(S)&:=\bigl\{i:\;\widehat r_{S,i}>0,\;  \hat\upsilon_{S,i}=1-\alpha\bigr\}.
\end{align}
These sets now \emph{evolve with~$S$}. We initialize $S$-path with the smallest $S^1$ such that the test point is in the elbow set (i.e., $S^1 = \hat{g}_{S^1}(X_{n+1})$) and find the smallest increment to the next $S$ such that an event occurs while the test point is still in the elbow. We use the same notion of an ``event" as before---any change in the elbow set. We iterate until the test point exits the elbow and use the final $S$ as $S^*(X_{n+1})$.

{ In this section, we assume the columns of $\mathbf\Phi^*\in \mathbb{R}^{(n+1)\times d}$ are linearly independent. The dimension of $\mathbf \Psi^*$ is now $\mathbb{R}^{(n+1)\times (n+1)}$.}
\subsubsection{Proof of Proposition \ref{prop:S_path}}\label{sec:prop_S}
\begin{proof}
Let $\{S^l\}_{l=1,2,3,\cdots}$ be the change points when an event occurs. Consider an interval $S^{l} \leq S \leq S^{l+1}$ during which the sets stay the same, i.e., $(E(S), L(S), R(S))=(E(S^l), L(S^l), R(S^l))$. Denote $\hat{\upsilon}_{S,i}$ and $\hat\eta_S$ as the solution of \eqref{eq: opt obj1} given $S$. In this proof, denote $E=E(S)=E(S^l)$, $L=L(S)=L(S^l)$, and $R=R(S)=R(S^l)$. Here, $\lambda$ is fixed as the selected hyperparameter from the previous step.

For every index \(i\) we have,

\begin{equation}\label{eq:gXi}
	 \begin{split}
	     \hat{g}_S(X_i) &= \mathbf\Phi^*_{i\cdot}\hat{\eta}_S+\frac{1}{\lambda}\mathbf\Psi^*_{i,\cdot}\hat{\upsilon}_S \\
	     &= \mathbf\Phi^*_{i\cdot}\hat{\eta}_S+\frac{1}{\lambda}\big(\mathbf\Psi^*_{i,E}\hat{\upsilon}_{S,E}-\alpha\mathbf\Psi^*_{i,L}\mathbf{1}_L+(1-\alpha)\mathbf\Psi^*_{i,R}\mathbf{1}_R\big)\\
      &= \mathbf\Phi^*_{i\cdot}\hat{\eta}_S+\frac{1}{\lambda}\mathbf\Psi^*_{i,E}\hat{\upsilon}_{S,E}+d_i ,
     \end{split}
 \end{equation}

where
 \begin{equation*}
      d_i:= \frac{1}{\lambda}(-\alpha\mathbf\Psi^*_{i,L}\mathbf{1}_L+(1-\alpha)\mathbf\Psi^*_{i,R}\mathbf{1}_R).
 \end{equation*}

From the second KKT condition in \eqref{eq:deriv app1}, we have $\upsilon_{i}=\sigma_i-\tau_i$ and $\sum_{i=1}^{n+1}(\sigma_i-\tau_i)\mathbf\Phi^*_{i,j}=\sum_{i=1}^{n+1}\upsilon_i\mathbf\Phi^*_{i,j}=0$ for $j=1,\cdots,d$. In compact form,
\[
  \mathbf\Phi_E^{*\top}\hat{\upsilon}_{S,E}
   =
  \alpha\,\mathbf\Phi_L^{*\top}\mathbf 1_L
  -(1-\alpha)\,\mathbf\Phi_R^{*\top}\mathbf 1_R .
\]
This means that,
\[
 \mathbf\Phi_E^{*} (\mathbf\Phi_E^{*\top} \mathbf\Phi_E^{*})^{\dagger} \mathbf\Phi_E^{*\top}\hat{\upsilon}_{S,E}
   =
  \alpha\, \mathbf\Phi^*_E (\mathbf\Phi_E^{*\top} \mathbf\Phi_E^*)^{\dagger} \mathbf\Phi_L^{*\top}\mathbf 1_L
  -(1-\alpha)\, \mathbf\Phi^*_E (\mathbf\Phi_E^{*\top} \mathbf\Phi^*_E)^{\dagger}\mathbf\Phi_R^{*\top}\mathbf 1_R .
\]

Let
\(S_E:=(S_i)_{i\in E},
  d_E:=(d_i)_{i\in E},
  \mathbf\Phi_E^*\in\mathbb R^{|E|\times d},
 \mathbf \Psi^{*}_{EE}\in\mathbb R^{|E|\times|E|}\).
Equation~\ref{eq:gXi} for \(i\in E\) becomes,
\begin{equation}\label{eq:SE_full}
  S_E
   =
  \mathbf\Phi_E^*\,\hat{\eta}_S
  +\frac{1}{\lambda}\,\mathbf\Psi^{*}_{EE}\,\hat{\upsilon}_{S,E}
  + d_E .
\end{equation}

Define the orthogonal projector $\mathbf\Pi_E
  := I_{|E|}
     -\mathbf\Phi^*_E\bigl(\mathbf\Phi_E^{*\top}\mathbf\Phi^*_E\bigr)^{\dagger}\mathbf\Phi_E^{*\top}
$.
Because \(\mathbf\Pi_E\mathbf\Phi^*_E=0\), multiplying \eqref{eq:SE_full} by \(\mathbf\Pi_E\) gives,
\begin{equation}\label{eq:ProjEq}
  \mathbf\Pi_E S_E
   =
  \frac{1}{\lambda}\,\mathbf\Pi_E\mathbf\Psi^{\star}_{EE}\,\hat{\upsilon}_{S,E}
  + \mathbf\Pi_E d_E .
\end{equation}

Write \(S_E = S_E^{\text{fixed}} + S\,\mathbf{e}_{n+1}\),
where \(S_E^{\text{fixed}}\) has a zero in the \((n\!+\!1)\)-st row and
\(\mathbf{e}_{n+1}\) selects that row.
Equation \ref{eq:ProjEq} becomes,
\begin{equation}\label{eq:affine}
  \mathbf\Pi_E\mathbf\Psi^{\star}_{EE}\,\hat{\upsilon}_{S,E}
   =
  \lambda\,\mathbf\Pi_E\!\bigl(S_E^{\text{fixed}}-d_E\bigr)
  +\lambda\,S\,\mathbf\Pi_E\mathbf{e}_{n+1}.
\end{equation}

Since $I_E = \mathbf\Phi_E^*\bigl(\mathbf\Phi_E^{*\top}\mathbf\Phi^*_E\bigr)^{\dagger}\mathbf\Phi_E^{*\top} + \mathbf\Pi_E$, the previous equation yields,
\begin{equation}\label{eq:affine2}
  \mathbf\Pi_E\mathbf\Psi^{\star}_{EE}\,  \mathbf\Pi_E \hat{\upsilon}_{S,E}
   =
  -\mathbf\Pi_E\mathbf\Psi^{\star}_{EE} \mathbf\Phi^*_E(\mathbf\Phi_E ^{*\top} \mathbf\Phi^*_E)^{\dagger}\mathbf\Phi_E^{*\top}  \hat{\upsilon}_{S,E} + \lambda\,\mathbf\Pi_E\!\bigl(S_E^{\text{fixed}}-d_E\bigr)
  +\lambda\,S\,\mathbf\Pi_E\mathbf{e}_{n+1}.
\end{equation}

Now, we know that:

\[
\mathbf\Pi_E\mathbf\Psi^{\star}_{EE}  \mathbf\Phi^*_E (\mathbf\Phi_E ^{*\top} \mathbf\Phi_E^*)^{\dagger} \mathbf\Phi_E^{*\top}\hat{\upsilon}_{S,E}
   =
  \alpha\, \mathbf\Pi_E\mathbf\Psi^{\star}_{EE}  \mathbf\Phi^*_E (\mathbf\Phi_E^{*\top} \mathbf\Phi^*_E)^{\dagger} \mathbf\Phi_L^{*\top}\mathbf 1_L
  -(1-\alpha)\,\mathbf\Pi_E\mathbf\Psi^{\star}_{EE}  \mathbf\Phi^*_E (\mathbf\Phi_E ^{*\top} \mathbf\Phi_E^*)^{\dagger}\mathbf\Phi_R^{*\top}\mathbf 1_R .
\]

Because \(\mathbf\Pi_E\) is an orthogonal projector (\(\mathbf\Pi_E^2=\mathbf\Pi_E\)), the matrix
\(\mathbf\Pi_E\mathbf\Psi^{\star}_{EE}\mathbf\Pi_E\) is positive definite on the image of \(\mathbf\Pi_E\).
Using its Moore–Penrose inverse (denoted by superscript \(\dagger\)) gives the
unique minimum‑norm solution,
\begin{equation}\label{eq:ups_S}
   \begin{split}
       \mathbf\Pi_E \hat{\upsilon}_{S,E}
     &=
    \lambda\,
    \bigl(\mathbf\Pi_E\mathbf\Psi^{\star}_{EE}\mathbf\Pi_E\bigr)^{\dagger}
    \,\mathbf\Pi_E
    \bigl(S_E^{\text{fixed}}-d_E + S\,\mathbf{e}_{n+1}\bigr)
  \\
  &- \alpha\, \bigl(\mathbf\Pi_E\mathbf\Psi^{\star}_{EE}\mathbf\Pi_E\bigr)^{\dagger}\mathbf\Pi_E\mathbf\Psi^{\star}_{EE}  \mathbf\Phi_E^* (\mathbf\Phi_E ^{*\top} \mathbf\Phi^*_E)^{\dagger} \mathbf\Phi_L^{*\top}\mathbf 1_L\\
  &
  +(1-\alpha)\,\bigl(\mathbf\Pi_E\mathbf\Psi^{\star}_{EE}\mathbf\Pi_E\bigr)^{\dagger}\mathbf\Pi_E\mathbf\Psi^{\star}_{EE}  \mathbf\Phi_E^* (\mathbf\Phi_E ^{*\top} \mathbf\Phi^*_E)^{\dagger}\mathbf\Phi_R^{*\top}\mathbf 1_R
   \end{split}
\end{equation}

Therefore, since $\hat{\upsilon}_{S,E} =  \mathbf\Pi_E \hat{\upsilon}_{S,E} + \mathbf\Phi_E^* (\mathbf\Phi_E ^{*\top} \mathbf\Phi^*_E)^{\dagger}\mathbf\Phi_E^{*\top} \hat{\upsilon}_{S,E}$:

\begin{equation}\label{eq:ups_S2}
   \begin{split}
\hat{\upsilon}_{S,E} &=
   \lambda\,
    \bigl(\mathbf\Pi_E\mathbf\Psi^{\star}_{EE}\mathbf\Pi_E\bigr)^{\dagger}
    \,\mathbf\Pi_E
    \bigl(S_E^{\text{fixed}}-d_E + S\,\mathbf{e}_{n+1}\bigr)
  \\
&  \qquad -\alpha\, \bigl(\mathbf\Pi_E\mathbf\Psi^{\star}_{EE}\mathbf\Pi_E\bigr)^{\dagger}\mathbf\Pi_E\mathbf\Psi^{\star}_{EE}  \mathbf\Phi_E^* (\mathbf\Phi_E ^{*\top} \mathbf\Phi_E^*)^{\dagger} \mathbf\Phi_L^{*\top}\mathbf 1_L\\
&  \qquad
  +(1-\alpha)\,\bigl(\mathbf\Pi_E\mathbf\Psi^{\star}_{EE}\mathbf\Pi_E\bigr)^{\dagger}\mathbf\Pi_E\mathbf\Psi^{\star}_{EE} \mathbf\Phi_E^* (\mathbf\Phi_E ^{*\top} \mathbf\Phi_E^*)^{\dagger}\mathbf\Phi_R^{*\top}\mathbf 1_R  \\
  &  \qquad +\alpha\, \mathbf\Phi_E^* (\mathbf\Phi_E ^{*\top} \mathbf\Phi_E^*)^{\dagger} \mathbf\Phi_L^{*\top}\mathbf 1_L\\
    &    \qquad
  -(1-\alpha)\, \mathbf\Phi_E^* (\mathbf\Phi_E ^{*\top} \mathbf\Phi_E^*)^{\dagger}\mathbf\Phi_R^{*\top}\mathbf 1_R \\
  &=
   \lambda\,
    \bigl(\mathbf\Pi_E\mathbf\Psi^{\star}_{EE}\mathbf\Pi_E\bigr)^{\dagger}
    \,\mathbf\Pi_E
    \bigl(S_E^{\text{fixed}}-d_E + S\,\mathbf{e}_{n+1}\bigr)
  \\
&  \qquad +\alpha\, [\,I_{|E|}- (\mathbf\Pi_E\mathbf\Psi^\star_{EE}\mathbf\Pi_E)^{\dagger}\mathbf\Pi_E\mathbf\Psi^\star_{EE}\,]
\mathbf\Phi_E^* (\mathbf\Phi_E ^{*\top} \mathbf\Phi_E^*)^{\dagger} \mathbf\Phi_L^{*\top}\mathbf 1_L\\
&  \qquad
  -(1-\alpha)\,[\,I_{|E|}- (\mathbf\Pi_E\mathbf\Psi^\star_{EE}\mathbf\Pi_E)^{\dagger}\mathbf\Pi_E\mathbf\Psi^\star_{EE}\,]
\mathbf\Phi_E^* (\mathbf\Phi_E ^{*\top} \mathbf\Phi_E^*)^{\dagger}\mathbf\Phi_R^{*\top}\mathbf 1_R  \\
   \end{split}
\end{equation}

In particular, the kernel parameter of the test point,
\(\hat{\upsilon}_{S,n+1}\),
is affine in \(S\) on every segment where the index
sets \((E,L,R)\) stay unchanged.

Likewise, the linear coefficient satisfies,
\begin{equation}
  \label{eq:eta_affine0}
  \begin{split}
    S_E & = \mathbf\Phi_E^* \hat\eta_S + \frac{1}{\lambda} \mathbf\Psi^\star_{EE} \hat\upsilon_{S,E} + d_E\\
   \mathbf\Phi_E ^{*\top} S_E & = \mathbf\Phi_E ^{*\top} \mathbf\Phi_E ^*\hat\eta_S + \frac{1}{\lambda} \mathbf\Phi_E ^{*\top} \mathbf\Psi^\star_{EE} \hat\upsilon_{S,E} + \mathbf\Phi_E^{*\top}d_E\\
   \implies  \hat\eta_S & =  (\mathbf\Phi_E ^{*\top} \mathbf\Phi^*_E)^{\dagger}\mathbf\Phi_E^{*\top} S_E - \frac{1}{\lambda} (\mathbf\Phi_E ^{*\top} \mathbf\Phi^*_E)^{\dagger}\mathbf\Phi_E^{*\top} \mathbf\Psi^\star_{EE} \hat\upsilon_{S,E} -(\mathbf\Phi_E ^{*\top} \mathbf\Phi_E^* )^{\dagger} \mathbf\Phi_E^{*\top}d_E\\
      \implies  \hat{\eta}_S & =  (\mathbf\Phi_E ^{*\top} \mathbf\Phi^*_E)^{\dagger}\mathbf\Phi_E^{*\top} S_E^{\text{fixed}} + S (\mathbf\Phi_E ^{*\top} \mathbf\Phi^*_E)^{\dagger}\mathbf\Phi_E^{*\top} \mathbf{e}_{n+1} \\
      &\quad\quad- \frac{1}{\lambda} (\mathbf\Phi_E ^{*\top} \mathbf\Phi^*_E)^{\dagger}\mathbf\Phi_E^{*\top} \mathbf\Psi^\star_{EE} \hat\upsilon_{S,E} -(\mathbf\Phi_E ^{*\top} \mathbf\Phi_E^* )^{\dagger} \mathbf\Phi_E^{*\top}d_E\\
      &=(\mathbf\Phi_E^{*\top}\mathbf\Phi^*_E)^{\dagger}\mathbf\Phi_E^{*\top}
   \Bigl[S_E^{\text{fixed}}+S\mathbf{e}_{n+1}-d_E-\frac{1}{\lambda}\,\mathbf\Psi^*_{EE}\,\hat\upsilon_{S,E}\Bigr],\\
  \end{split}
\end{equation}

and thus, we have shown that,

\begin{equation}\label{affine_S}
    \hat{\upsilon}_{S,E} = c+Sd, \qquad \hat{\eta}_S = c^{(1)}+Sd^{(1)}.
\end{equation}

with \(c, d \in\mathbb R^{|E|},  c^{(1)}, d^{(1)}\in\mathbb R^{d}\) constant on the segment $S^l \leq S \leq S^{l+1}$.

Insert \eqref{affine_S} back to \eqref{eq:gXi}. For any $i \in [n+1]$,
\begin{align}
  \hat g_S(X_i)
  &=
\mathbf\Phi^*_{i\cdot}\,\hat\eta_S
     + \frac1\lambda\,\mathbf\Psi_{ iE}^{\star}\,\hat\upsilon_{S,E}
     + d_i\\
  &= \mathbf\Phi^*_{i\cdot}\bigl( c^{(1)}+ \,S d^{(1)}\bigr)
     +\frac1\lambda
    \mathbf \Psi_{ iE}^{\star}\bigl( c+ \,S d\bigr)
     + d_i \\[3pt]
  &= \underbrace{\Bigl(\mathbf\Phi^*_{i\cdot}c^{(1)}
                       +\frac1\lambda\mathbf\Psi_{ iE}^{\star} c
                       +d_i\Bigr)}_{=:g_i^{(0)}}
      +
     S \underbrace{\Bigl(\,\mathbf\Phi^*_{i\cdot} d^{(1)}
                          + \frac{1}{\lambda}\mathbf\Psi_{ iE}^{\star} d\Bigr)}_{=:g_i^{(1)}}.
  \label{eq:g_affine}
\end{align}
Thus \(g_S(X_i)=g_i^{(0)}+S\,g_i^{(1)}\) is affine in \(S\). There are two cases for the residual $r_i(S) = S_i - \hat g_S(X_i)$:

\begin{enumerate}
    \item {\it Calibration index \(i\le n\).}
  The score \(S_i\) is fixed, hence
  \[
    r_i(S)  =  \bigl[S_i-g_i^{(0)}\bigr]
                  -  S\,g_i^{(1)}.
  \]
  Both $S_i-g_i^{(0)}$ and $g_i^{(1)}$ are constants on the segment.

\item {\it Test index \(i=n+1\).}
  Here \(S_{n+1}=S\), so
  \[
    r_{n+1}(S)  =
    S - g_{n+1}^{(0)} - S\,g_{n+1}^{(1)}
     =
    \bigl[\,1-g_{n+1}^{(1)}\bigr]\,S - g_{n+1}^{(0)},
  \]
  which is again affine in \(S\).
\end{enumerate}

Because every \(r_i(S)\) is an affine function,
each index outside the elbow can cross the zero‑residual line
at most once on the segment.
Likewise each \(\upsilon_{i}(S)\) in \eqref{affine_S}
can hit the bounds \(1-\alpha\) or \(-\alpha\) at most once.
Hence the overall solution path is {\it piecewise affine}
with break‑points occurring exactly when  either (a) a coefficient in \(E\) hits its bound, or  (b) a residual for an index in \(L\cup R\) reaches zero,
completing the argument used by the S‑path algorithm.

\end{proof}

\subsubsection{\texorpdfstring{Update of $S^l$}{Update of Sl}}

Let $S^l$ be the value after the $l^{th}$ event. The elbow set $E(S^l)$ is updated when one of the following events occurs:

\begin{itemize}
    \item A point $i$ in either $L(S^l)$ or $R(S^l)$ enters the elbow set (residual $S_i - \hat{g}_{S^l} (X_i)$ becomes 0).

\item A point $i$ in $E(S^l)$ leaves to the left or right set ($\hat{\upsilon}_{S^l,i}$ becomes $-\alpha$ or $1-\alpha$).
\end{itemize}

For the first event, note that $\frac{\partial r_i(S)}{\partial S} = -(\mathbf\Phi^*_{i\cdot} d^{(1)}
                          + \frac{1}{\lambda}\mathbf\Psi_{ iE}^{\star} d)$ in \eqref{eq:g_affine}. Then we have,

\begin{equation*}
    S^{l+1,hit}=S^l + \min_{i\in L(S^l)\cup R(S^l)} r_i(S^l)/\big(\mathbf\Phi^*_{i\cdot} d^{(1)}
                          + \frac{1}{\lambda}\mathbf\Psi_{ iE}^{\star} d\big)\mathbf{1}\Big(r_i(S^l)/\big(\mathbf\Phi^*_{i\cdot} d^{(1)}
                          + \frac{1}{\lambda}\mathbf\Psi_{ iE}^{\star} d\big)\geq 0\Big)
\end{equation*}

To find $S$ such that a point leaves $E(S^l)$, recall $\frac{\partial \hat{\upsilon}_{S,i}}{\partial S}={d}_i$ for $i \in E(S^l)$ (\eqref{eq:g_affine}).

\begin{equation*}
    S^{l+1,leave}=S^l+\min_{i \in E(S^l)} \big\{x \in\{\frac{-\alpha-\hat{\upsilon}_{S^l,i}}{{d}_i},\frac{1-\alpha-\hat{\upsilon}_{S^l,i}}{{d}_i}\} \ | \ x \leq 0\big\}
\end{equation*}

We then take $S^{l+1} = \min\big\{S^{l+1,hit}, S^{l+1,leave}\big\}$. We also update $(E,L,R)=(E(S^{l+1}), L(S^{l+1}), R(S^{l+1}))$ accordingly based on which event occurred. Parameters $\hat{\upsilon}_{S^{l+1},i}$'s, $\hat{\eta}_{S^{l+1}}$ can be updated by solving for the new elbow,

\begin{equation}
     \begin{pmatrix}
         \frac{1}{\lambda} \mathbf\Psi^{\star}_{EE}  & \mathbf\Phi^*_{E} \\
        \mathbf\Phi_{E}^{*\top} &0
     \end{pmatrix}  \begin{pmatrix}
         \upsilon \\ \eta
     \end{pmatrix}=  \begin{pmatrix}
         S_{E}  - \frac{1}{\lambda}( -\alpha \mathbf\Psi^{\star}_{EL}\mathbf{1}_{L}+(1-\alpha) \mathbf\Psi^{\star}_{ER}\mathbf{1}_{R})\\
         \alpha \mathbf\Phi_L^{*\top}\mathbf{1}_L-(1-\alpha)\mathbf\Phi_R^{*\top}\mathbf{1}_R
     \end{pmatrix}
 \end{equation}

\subsubsection{\texorpdfstring{Initialization of $S$}{Initialization of S}}

Let us first assume that the imputed test score $S$ is small enough so that $S < \hat{g}_S(X_{n+1})$. Then, the test point ${n+1} \in L(S)$.
We use the notation $\hat\upsilon_{\text{small}}$ and $\hat\eta_{\text{small}}$ (instead of $\hat\upsilon_S$ and $\hat\eta_S$), to denote the regression parameters. In this case, the residual $r_{n+1}(S) = S - \hat{g}_{\text{small}}(X_{n+1}) = S -\mathbf\Phi^*_{n+1, \cdot} \hat\eta_{\text{small}} - \frac{1}{\lambda}\mathbf\Psi^\star_{n+1, \cdot} \hat \upsilon_{\text{small}}$ is linear in $S$. We can therefore track the moment when it enters the elbow set $E(S)$. This happens as soon as:
$$ S = \mathbf\Phi^*_{n+1, \cdot} \hat\eta_{\text{small}} + \frac{1}{\lambda}\mathbf\Psi^\star_{n+1, \cdot} \hat \upsilon_{\text{small}}.$$

We thus solve for $\hat \upsilon_{\text{small}}$ and $\hat \eta_{\text{small}}$ with $\upsilon_{\text{small,} n+1} =-\alpha$ (as the test point is in the left set). Let $\upsilon^{\text{fixed}}\in \mathbb{R}^{n+1}$ be the vector defined as equal to $\upsilon_S$ on all entries except the $(n+1)^{th}$, where it is set to 0:

$$ \upsilon_{i}^{\text{fixed}} = \begin{cases}
\upsilon_{S,i} \text{ if } i \neq n+1\\
0 \text{ if }  i =n+1.
\end{cases}$$
This allows us to write,
$$ \upsilon_{\text{small}} = \upsilon^{\text{fixed} } -\alpha \mathbf{e}_{n+1},$$
where $\mathbf{e}_{n+1}$ is the indicator vector of the $(n+1)^{th}$ coefficient,
$$ \mathbf{e}_{n+1,i } =0 \text{ if } i\neq n+1, \quad \mathbf{e}_{n+1,n+1}  =1. $$
The problem then becomes,

 \begin{equation}\label{eq: lagrangian in Spath}
     \begin{split}
         L_p &= (1-\alpha)\sum_{i=1}^{n+1}p_i+\alpha\sum_{i=1}^{n+1}q_i+\frac{1}{2\lambda}\upsilon^{\top}_{\text{small}}  \mathbf\Psi^{\star}\upsilon_{\text{small}} \\
         &\qquad +\sum_{i=1}^{n+1} \sigma_i(S_i-g_{\text{small}} (X_i)-p_i)-\sum_{i=1}^{n+1}\tau_i(S_i-g_{\text{small}} (X_i)+q_i)\\
         &\qquad -\sum_{i=1}^{n+1}\kappa_ip_i-\sum_{i=1}^{n+1}\rho_iq_i\\
          &= \sum_{i=1}^{n+1}\sigma_i p_i+\sum_{i=1}^{n+1}\tau_i q_i+\frac{1}{2\lambda}\upsilon^{\top}_{\text{small}}  \mathbf\Psi^{\star}\upsilon_{\text{small}} \\
         &\qquad +\sum_{i=1}^{n+1}( \sigma_i -\tau_i)(S_i-\mathbf\Phi^*_{i,\cdot} \eta_{\text{small}}   - \frac{1}{\lambda}\mathbf\Psi^\star_{i\cdot} \upsilon_{\text{small}} )  -\sum_{i=1}^{n+1}( \sigma_i p_i +\tau_i q_i)\\
          &= \frac{1}{2\lambda} (\upsilon^{\text{fixed}}_{\text{small}}   -\alpha\mathbf{e}_{n+1})^\top\mathbf\Psi^{\star}(\upsilon_{\text{small}} ^{\text{fixed}} -\alpha \mathbf{e}_{n+1})\\
          &\qquad+\sum_{i=1}^{n}( \sigma_i -\tau_i)(S_i-\mathbf\Phi^*_{i,\cdot} \eta_{\text{small}}   -  \frac{1}{\lambda}\mathbf\Psi^{\star}_{i\cdot} (\upsilon_{\text{small}} ^{\text{fixed}} -\alpha\mathbf{e}_{n+1}))\\
          &\qquad-\alpha(S-\mathbf\Phi^*_{n+1, \cdot} \eta_{\text{small}}   -  \frac{1}{\lambda}\mathbf\Psi^{\star}_{n+1, 1:n} (\upsilon_{\text{small}} ^{\text{fixed}}  -\alpha\mathbf{e}_{n+1}))\\
         &=\frac{1}{2\lambda}\delta^{\top} \mathbf\Psi_{1:n, 1:n}^{\star} \delta -\frac{\alpha}{\lambda} \mathbf\Psi^{\star}_{n+1, 1:n} \delta + \frac{\alpha^2}{2\lambda} \mathbf\Psi^{\star}_{n+1, n+1}   \\
         & \qquad  + \delta^T (S_{1:n}  - \mathbf\Phi^*_{1:n,\cdot} \eta_{\text{small}}  -  \frac{1}{\lambda}\mathbf\Psi^{\star}_{1:n,1:n} \delta + \frac{\alpha}{\lambda} \mathbf\Psi^{\star}_{1:n, n+1}  ) \\
            & \qquad -\alpha(S-\mathbf\Phi^*_{n+1, \cdot} \eta_{\text{small}}   -  \frac{1}{\lambda}\Psi^{\star}_{n+1, 1:n} \delta +  \frac{\alpha}{\lambda}\mathbf\Psi^{\star}_{n+1,n+1})),
     \end{split}
 \end{equation}
with $ \delta \in \mathbb{R}^n$ the vector whose entries are defined as:   $\delta_i = \sigma_i - \tau_i =  \upsilon_{\text{small},i}  $.

We also know that $(\sigma -\tau)^\top \mathbf\Phi^* = 0 \implies \forall j \in [d], \quad \sum_{i=1}^{n} (\sigma_{i} -\tau_i) \Phi^*_{i, j} = -(\sigma_{n+1} -\tau_{n+1}) \Phi^*_{n+1, j} =\alpha \Phi^*_{n+1, j}. $

 Therefore, taking derivatives with respect to $\delta$ yields the following system:

 \begin{equation}
     \begin{pmatrix}
         \frac{1}{\lambda} \mathbf\Psi^{\star}_{EE}  & \mathbf\Phi^*_{E} \\
        \mathbf\Phi_{E}^{*\top} &0
     \end{pmatrix}  \begin{pmatrix}
         \delta_E \\ \eta_{\text{small}}
     \end{pmatrix}=  \begin{pmatrix}
         S_{E}  -\frac{1}{\lambda}(-\alpha \mathbf\Psi^{\star}_{E,n+1}-\alpha\mathbf\Psi^{\star}_{EL}\mathbf{1}_L+(1-\alpha)\mathbf\Psi^{\star}_{ER}\mathbf{1}_R) \\
         \alpha \mathbf\Phi_{n+1}^{*\top}+\alpha \mathbf\Phi_L^{*\top}\mathbf{1}_L-(1-\alpha)\mathbf\Phi_R^{*\top}\mathbf{1}_R
     \end{pmatrix}
 \end{equation}

We can therefore solve for both $\delta_E$ and $\eta_{\text{small}}$ by inverting the previous system of equations.

\paragraph{When $S> g(X_{n+1})$.} Similarly, when we start with a large $S$ so that the test point is in the right set, we can derive the coefficients for both $\upsilon$ and $\eta$ by the same derivations as for the small case:

 \begin{equation}
     \begin{pmatrix}
         \frac{1}{\lambda} \mathbf\Psi^{\star}_{EE}  & \mathbf\Phi^*_{E} \\
        \mathbf\Phi_{E}^{*\top} &0
     \end{pmatrix}  \begin{pmatrix}
         \delta_E \\ \eta_{\text{small}}
     \end{pmatrix}=  \begin{pmatrix}
         S_{E}  -\frac{1}{\lambda}((1-\alpha) \mathbf\Psi^{\star}_{E,n+1}-\alpha\mathbf\Psi^{\star}_{EL}\mathbf{1}_L+(1-\alpha)\mathbf\Psi^{\star}_{ER}\mathbf{1}_R) \\
         (\alpha-1) \mathbf\Phi_{n+1}^{*\top}+\alpha \mathbf\Phi_L^{*\top}\mathbf{1}_L-(1-\alpha)\mathbf\Phi_R^{*\top}\mathbf{1}_R
     \end{pmatrix}
 \end{equation}

\subsubsection{Computational Complexity}
At each breakpoint $l$ of $\lambda$- and $S$-path, the computational complexity includes (1) computing the projection matrix $\mathbf\Pi_{E^l}$; (2) computing the Moore-Penrose inverse of $\mathbf\Pi_{E^l}\mathbf\Psi^\star_{E^lE^l}\mathbf\Pi_{E^l}$; (3) solving $(|E^l|+d)\times (|E^l|+d)$ linear system; (4) evaluating $h^l(X_i)$ and $\hat g^l(X_i)$ for $i \in [n+1]$. Letting $n_{E^l}:=|E^l|$, the per-step complexity is $O(n_{E^l}^3+(n_{E^l}+d)^3+n(d+n_{E^l}))$.

A provable polynomial bound on the number of breakpoints is not available for general kernel quantile regression---this remains an open question even for the classical LASSO. Empirically, \citet{li2007quantile} report that the number of breakpoints is on average a small multiple of $n$. In our framework, Propositions~\ref{prop:lambda_path}, \ref{prop:S_path} establish that the kernel parameters $\hat{\upsilon}_{S,i}$ are affine in both $\lambda$ and $S$. Since an affine function crosses zero at most once, each point can either enter or exit the elbow set at most once,  suggesting that the number of breakpoints is at most $2(n+1)$ in our setting.

Let $\bar{m}$ denote the average of $n_E^l$ across breakpoints. The total complexity of our method is approximately $O(n(\bar m+d)^3)$. This compares favorably to the cross-validation approach, which costs $O(N_{\lambda} n^3)$ where $N_{\lambda}$ is the size of the $\lambda$ grid. Since a fine grid from $\lambda_{\mathrm{max}}$ to 0 requires $N_{\lambda} = O(n)$ points to resolve all meaningful changes in the solution, the complexity is effectively $O(n^4)$, matching our worst case. In practice, $\bar m^3 \ll n^3$ makes our proposed algorithm much faster, which we also show in Appendix~\ref{app:elbow-set-size}.

\paragraph{Practical Consequences.}

\begin{itemize}
  \item {\bfseries Threshold evaluation}
        Because $S\mapsto\upsilon_{S,n+1}$ is affine on each segment,
        the conformal threshold
        $\hat g(X_{n+1})=\sup\{S:\upsilon_{S,n+1}<1-\alpha\}$
        is found by a single root computation, \emph{not} by binary search.
  \item {\bfseries Randomization}
      By Lemma~\ref{theorem: oracle mixture}, using the final update of $S$-path, $S^*(X_{n+1})$, can inflate the conditional coverage. To mitigate this, we can use the randomized cutoff $ S^{rand}(X_{n+1}) = \sup\{S\mid \hat \upsilon_{S,n+1}< U\}$, for $U\sim Unif(-\alpha,1-\alpha)$. The procedure of $S$-path stays the same but we stop the algorithm as soon as $\hat{\upsilon}_{S,n+1}\geq U$.

\end{itemize}

\paragraph{Summary of Modifications Versus the $\lambda$-path.}

\begin{center}
\begin{tabular}{@{}lcc@{}}
\toprule
\textbf{Component} & \textbf{$\lambda$–path} & \textbf{S–path (fixed $\lambda$)}\\
\midrule
Moving parameter        & $\lambda\downarrow 0$ & $S\uparrow$ \\
Active sets             & $E(\lambda),L(\lambda),R(\lambda)$ & $E(S),L(S),R(S)$ \\
Triggering event        & Elbow set changes & Elbow set changes \\
Segment law             & $\lambda\mapsto\hat\upsilon(\lambda)$ affine       & $S\mapsto\hat\upsilon_S$ affine \\
Break‑points            & critical values of $\lambda$          & critical values of $S$ \\
Output                  & $\lambda\!\mapsto\!(\upsilon,\eta)$   & $S\!\mapsto\!(\upsilon,\eta)$ \\
\bottomrule
\end{tabular}
\end{center}

The resulting algorithm furnishes an explicit, efficient \emph{score‑path}
for any fixed~$\lambda$, enabling local density–adaptive conformal prediction
and other post‑hoc analyses.

{
\subsection{Cross Validation on \texorpdfstring{$\gamma$}{gamma}}

In this section, we provide additional details on the selection of $(\hat{\gamma}, \hat{\lambda})$, described in Algorithm~\ref{algo:1}.

We first select a fixed grid of $\gamma$'s, denoted as $\Gamma$. For each $\gamma \in \Gamma$, we measure k-fold validation error. This is done by running a $\lambda$-path for each $\gamma$ that yields a sequence of $\lambda$ values, $(\lambda^1, \lambda^2,\cdots)$. The error of $(\gamma, \lambda^l)$ for the $j$-th fold is then defined as the quantile regression loss,

\begin{equation*}
    \mathrm{CV}_j(\gamma,\lambda^l) =\sum_{i^{'} \in fold_j}\big((1-\alpha)[S_{i^{'}}-\hat{g}^l(X_i)]_++\alpha[S_{i^{'}}-\hat{g}^l(X_i)]_-\big)
\end{equation*}

A common practice is to get the mean of $\mathrm{CV}_j(\gamma,\lambda^l)$ over $k$ folds and choose the pair that minimizes it. However, running a $\lambda$-path for each $\gamma$ does not guarantee that the $\lambda$ values on the path will be all equivalent. To avoid this issue, we proceed with a two-step approach that first aggregates the error over all $\lambda$'s to get the optimal $\hat{\gamma}$,

\begin{equation*}
\begin{split}
    \mathrm{CV}_j(\gamma) &= \min_{\lambda^l}\mathrm{CV}_j(\gamma,\lambda^l) \\
    \hat{\gamma} &= \arg \min_{\gamma}\mathrm{CV}(\gamma)=\arg \min_{\gamma}\sum_{j=1}^k\mathrm{CV}_j(\gamma)
\end{split}
\end{equation*}

This selects $\hat{\gamma}$ that reflects the best smoothness of the function overall. Once we select $\hat{\gamma}$, we run the $\lambda$-path again on the full dataset ($m$ points) and choose $\hat{\lambda}$ that minimizes the Schwarz information criterion (SIC) \citep{schwarz1978estimating}, which is a commonly used alternative to cross validation error in kernel quantile regression for model selection \citep{li2007quantile}.}

\section{Theoretical Proof}\label{appsec: proof}

\subsection{Guarantee for Randomized Interval} \label{appsec: proof of Theorem randomized}
To incorporate the structured RKHS-based function in~\eqref{eq:closed_form} into the conformal calibration framework in \citet{gibbs2023conformal}, we need to show two propositions. Firstly, we show the monotonicity of the solution path for $S$. Namely, the mapping $S\mapsto \upsilon_{S, n+1}$ is nondecreasing in $S$. Second, we require the low-rank projection $\hat\pi(\cdot)$ to be trained symmetrically across the input data.  With these properties established, we are able to prove that our path algorithm satisfies exactly the same type of results as \citet{gibbs2023conformal}:

\begin{lemma}\label{theorem: oracle mixture} Consider the function class $\cF^{*}$ in \eqref{eq:fstar}, where RKHS component is given with the optimal $\hat \lambda$ such that $\cF_{\psi^*}=\{f_{\psi^*}(x)=\frac{1}{\hat \lambda}\sum_{i\in[n+1]}\upsilon_i \psi^*(x, X_i), \upsilon\in \mathbb{R}^{n+1}\}$. Suppose assumptions \ref{ass: data} and \ref{ass: low-rank projection} are both satisfied. Then, for all $f\in \cF^*$, SpeedCP gives
  \begin{align*}
        \bE\left[f(X_{n+1})\cdot\left(\One\{Y_{n+1}\in\hat C^{*}_{rand}(X_{n+1}) \}-(1-\alpha)\right)\right]=-2\hat\lambda\bE\left[\langle \hat g_{S^{rand}, \psi^*}, f_{\psi^*}\rangle\right],
    \end{align*}
    where $\hat g_{S^{rand}, \psi^*}(X)=\frac{1}{\hat\lambda}\sum_{i\in[n+1]}\hat\upsilon_{S^{rand}, i}\psi^*(X, X_i)$.
\end{lemma}
This result aligns with the randomization version of Theorem 3 in \citet{gibbs2023conformal}---but adapted here to our algorithm and choice of RKHS class $\cF_{\psi^*}$.

While in \citet{gibbs2023conformal} $\upsilon_i$ can be any arbitrary value, we involve the optimal $\hat\lambda$ in the definition of $f_{\psi^*}$. In this type of RKHS class, the relationship between $S$ and $\upsilon_{S, n+1}$ is explicit, while \citet{gibbs2023conformal} depends on a dual analysis, making the parameter less interpretable. Furthermore, the coverage gap $2\hat\lambda\bE[\langle \hat g_{S^{rand}, \psi^*}, f_{\psi^*}\rangle]$ arises because we have no prior information on the distribution shift and use a flexible RKHS-based function class instead. While it may lead to deviations from the nominal level $1-\alpha$ when $f_{\psi^*}\neq 0$, this deviation is measurable as shown by \citet{gibbs2023conformal}; we detail how to estimate this deviation in the latent-space setting in Appendix~\ref{appsec: coverage gap}.

\begin{proof}

By Proposition \ref{prop: nondecreasing}, $S\mapsto \upsilon_{S, n+1}$ is nondecreasing in $S$.  Furthermore, strong duality holds for the optimization problem in \eqref{eq: opt obj} (this has been shown in \citet{gibbs2023conformal}), and the KKT conditions are satisfied as shown in \ref{eq:kkt1}. Now consider a random variable $U\sim \mathrm{Unif}(-\alpha, 1-\alpha)$. Then we have the equivalence under the randomization for a given $S_{n+1}=S(X_{n+1}, y)$:
\begin{align*}
    \One\{S_{n+1}\leq \hat g_{S_{n+1}} (X_{n+1})\}\iff \One\{\hat \upsilon_{S_{n+1},n+1}\leq U\}
\end{align*}
Thus,
\begin{equation*}
\begin{aligned}
    &\bE\left[f(X_{n+1})\left(\One\left\{\hat \upsilon_{S_{n+1}, n+1}\leq U\right\}-(1-\alpha)\right)\right]\\
    =& \bE[ \bE_U\left[f(X_{n+1})\left(\One\left\{\hat \upsilon_{S_{n+1}, n+1} \leq U\right\}-(1-\alpha)\right)\mid X_{n+1},\hat \upsilon_{S_{n+1}, n+1}\right]]\\
    =&-\bE\left[f(X_{n+1})\hat \upsilon_{S_{n+1}, n+1} \right]
    \end{aligned}
\end{equation*}
Using the Lagrangian in Proposition \ref{prop: nondecreasing}, we follow the calculation in the proof of Proposition 4 of \citet{gibbs2023conformal}. By the exchangeability of the data and the symmetry of $\hat g_{S_{n+1}}$, we have
\begin{align*}
    -\bE\left[f(X_{n+1})\hat \upsilon_{S_{n+1}, n+1}\right]=-2\bE\left[\lambda\langle \hat g_{S_{n+1}, \psi^*}, f_{\psi^*}\rangle\right].
\end{align*}
Therefore, we replace $S_{n+1}$ with the randomized cutoff
$S^{rand}$ and $\lambda$ with the optimal $\hat\lambda$ to obtain the desired result.
\end{proof}

\begin{proposition}\label{prop: nondecreasing}
    For all maximizers $\{\upsilon_{ S,n+1}\}_{S\in \mathbb{R}}$ of the optimization problem in \eqref{eq: opt obj}, the mapping
    $S\mapsto\upsilon_{S,n+1}$ is non-decreasing in $S$.
\end{proposition}
\begin{proof}
    Recall the objective in \eqref{eq: opt obj}:
$$    \min_{\eta_{S}, \upsilon_{S}} \sum_{i=1}^{n+1}\ell_{\alpha}\big(S_i-\Phi^*(X_i)^\top\eta_{S}-\frac{1}{\lambda}\sum_{i^{'}=1}^{n+1} \upsilon_{S,i'}\psi^*(X_i,X_{i^{'}})\big)+\frac{1}{2\lambda}\sum_{i, i'=1}^{n+1}\upsilon_{S,i}\upsilon_{S,i^{'}}\psi^*(X_i,X_{i^{'}}).$$
Let $\mathbf{\Psi}^*=(\psi^*(x_i, x_j))_{i,j \in [n+1]}\in \bR^{(n+1)\times (n+1)}$ be the positive semidefinite kernel matrix. Following the structure in \citet{li2007quantile}, this objective is equivalent to the following quadratic program for a fixed imputed value $S$ (with $S_{n+1}=S$),
\begin{align*}
    \min_{\eta_S, \boldsymbol{\upsilon}_S} \; & (1-\alpha) \sum_{i=1}^{n+1} p_i + \alpha \sum_{i=1}^{n+1} q_i + \frac{1}{2\lambda} \upsilon^{\top}_S \mathbf{\Psi}^* \upsilon_S,
\end{align*}
subject to
\begin{align*}
   & -q_i \leq S_i - g_S(x_i) \leq p_i,\\
    &q_i, p_i \geq 0, \quad i = 1, \dots, n+1,
\end{align*}
where $$g_S(x_i) = \Phi^*(x_i)^\top\eta_S + \frac{1}{\lambda} \sum_{i'=1}^{n+1} \upsilon_{S,i'} \psi^*({x}_i, {x}_{i'}), \quad i = 1, \dots, n+1.
$$
Note that the proof of Proposition \ref{prop: nondecreasing} follows the argument structure of Theorem 4 in \citet{gibbs2023conformal}, but with a key distinction that the function $g_S(x)$ in our case incorporates an RKHS-based component that depends on $\lambda$.
The Lagrangian primal function is then defined as in \eqref{eq: lagrangian}. Setting the partial derivatives of $L_p$ with respect to $q$ and $p$ to zero, we obtain
\begin{equation}\label{eq:deriv p q}
    \begin{split}
        \frac{\partial L_p}{\partial p_i} &: \sigma_i=1-\alpha-\kappa_i\\
        \frac{\partial L_p}{\partial q_i} &: \tau_i=\alpha-\rho_i\\
    \end{split}
\end{equation}
Since minimizing with respect to $\bupsilon$ yields $\bupsilon_i=\sigma_i-\tau_i$, we can substitute this into the derivative expressions in Equation \eqref{eq:deriv p q}. We have
\begin{equation*}
    \begin{split}
(1-\alpha)\cdot \mathbf{1}-\upsilon=\kappa+\tau\\
        \alpha\cdot \mathbf{1}+\upsilon=\rho+\sigma\\
    \end{split}
\end{equation*}
Since $\kappa, \sigma, \tau, \rho$ are all non-negative, this can be simplified to
\begin{equation*}
    \begin{split}
(1-\alpha)\cdot \mathbf{1}\geq \upsilon\\
        -\alpha\cdot \mathbf{1}\leq\upsilon\\
    \end{split}
\end{equation*}
Let $Q^*(\upsilon)=-\min_{g\in \cF^*}\frac{1}{2\lambda}\upsilon^{\top} \mathbf{\Psi}^* \upsilon-\sum_{i=1}^{n+1}\upsilon_{i} g(X_i)$. Therefore, the dual formulation for \eqref{eq: lagrangian} is,
\begin{equation*}
\begin{split}
     \text{maximize}_{\upsilon} \sum_{i=1}^n\upsilon_i S_i+\upsilon_{n+1} S-Q^*(\upsilon)\\
     \text{subject to } -\alpha\leq \upsilon_i\leq 1-\alpha, 1\leq i\leq n+1
\end{split}
\end{equation*}
Note we use notation $\upsilon_S$ to denote the solution for a particular input $S$. Assume for the sake of contradiction that there exists \(\tilde{S} > S\) such that
\[
\upsilon_{ {\tilde{S}}, n+1} < \upsilon_{S,n+1}.
\]

Observe that $\sum_{i=1}^n\upsilon_i S_i-Q^*(\upsilon)$ does not depend on \(S\). The contradiction assumption implies that
\[
(\tilde{S} - S) \cdot \left(\upsilon_{ {\tilde{S}}, n+1} - \upsilon_{S,n+1}\right) < 0,
\]
or equivalently,
\[
\tilde{S} \cdot \left(\upsilon_{ {\tilde{S}}, n+1} -\upsilon_{S,n+1} \right) < S \cdot \left( \upsilon_{ {\tilde{S}}, n+1} -\upsilon_{S,n+1}\right).
\]
On the other hand, by the optimality of \(\upsilon_S\), we have that
\begin{equation*}
  \begin{split}
      \sum_{i=1}^n\upsilon_{\tilde S,i} S_i-Q^*(\upsilon_{\tilde S}) + \tilde{S} \cdot \upsilon_{\tilde S, n+1} \geq  \sum_{i=1}^n\upsilon_{ S,i} S_i-Q^*(\upsilon_{S}) + {S} \cdot \upsilon_{ S, n+1} \\
\quad \Longleftrightarrow \quad
\tilde{S} \cdot \left(\upsilon_{ {\tilde{S}}, n+1} -\upsilon_{S,n+1} \right) \geq\sum_{i=1}^n\upsilon_{ S,i} S_i-Q^*(\upsilon_{S}) -  \sum_{i=1}^n\upsilon_{\tilde S,i} S_i-Q^*(\upsilon_{\tilde S}).
  \end{split}
\end{equation*}
Applying inequality given by assumption above, we conclude that
\[
S \cdot \left( \upsilon_{ {\tilde{S}}, n+1} -\upsilon_{S,n+1} \right) >\sum_{i=1}^n\upsilon_{ S,i} S_i-Q^*(\upsilon_{S}) -  \sum_{i=1}^n\upsilon_{\tilde S,i} S_i-Q^*(\upsilon_{\tilde S}),
\]
which yields the contradiction \[
 \sum_{i=1}^n\upsilon_{\tilde S,i} S_i-Q^*(\upsilon_{\tilde S})+ \tilde S \cdot \upsilon_{\tilde S, n+1}> \sum_{i=1}^n\upsilon_{ S,i} S_i-Q^*(\upsilon_{S})  + S \cdot \upsilon_{ S, n+1}
\]

\end{proof}

\subsection{Proof of Theorem \ref{cor: local coverage}}
First, we consider the setting that $(X_1, Y_1), \dots, (X_n, Y_n)$ are independent of $(X_{n+1}, Y_{n+1}, W')$. Since $\hat\pi(\cdot)$ is a deterministic function (not a random variable), $\hat\pi(X_1), \dots, \hat\pi(X_{n})$ are also independent of $\hat\pi(X_{n+1})$. Since $\hat\pi(\cdot)$ is a pre-trained map from the covariate space to the latent space, we write $\hat\pi:\mathcal{X}\to \mathcal{W}$, where $\mathcal{W}$ denotes the latent representation space. Given this embedding, we define a kernel directly on the latent space $\psi^*_W:\mathcal{W}\times \mathcal{W}\to \mathbb{R}$. Consequently, $\psi^*(x, x')= \psi^*_W(\hat\pi(x),\hat\pi(x'))$. {   In our construction, $\psi^*$ is already normalized, so $\psi^*_W(w, \cdot)$ is a density kernel in its second argument with respect to a base measure on $\mathcal W$, i.e., \begin{align*}
    \int_\mathcal{W} \psi^*_W(w, w') dw'=1\qquad \text{ for each } w\in\mathcal{ W}.
\end{align*}}
Let $P=P_X\times P_{Y\mid X}$.
By the definition of $W'$, the joint distribution of $(X_{n+1}, Y_{n+1}, W')$ is defined by
\begin{align*}
    &X_{n+1}\sim P_X;\\
    &Y_{n+1}\mid X_{n+1}\sim P_{Y\mid X};
    \\& W'\mid (X_{n+1},Y_{n+1})\sim \psi^*(X_{n+1}, \cdot).
\end{align*}
By definition of $\psi^*_W$, we equivalently have $W'\mid (\hat\pi(X_{n+1}),Y_{n+1})\sim \psi^*_W(\hat\pi(X_{n+1}), \cdot)$.
Then, the conditional distribution $(X_{n+1}, Y_{n+1})\mid W'$ is given by
\begin{align*}
    (X_{n+1}, Y_{n+1})\mid W'=w' &\sim \frac{(P_X\circ \psi^*_W(\hat\pi(X_{n+1}), w'))\times P_{Y\mid X}}{\int_{(x, y)}(P_X\circ \psi^*_W(\hat\pi(X_{n+1}), w'))\times P_{Y\mid X}dxdy}\\
    &\sim \frac{ \psi^*_W(\hat\pi(x), w')}{\bE[ \psi^*_W(\hat\pi(X), w')]} dP_{(X, Y)}(x, y):= dP_f(x)\text{    by the symmetric of $\hat\pi(\cdot)$}
\end{align*}
Thus conditioning on $W'$, we get
\begin{align*}
    &\bE\left[\One\{Y_{n+1}\in \hat C_{rand}^*(X_{n+1})\}-(1-\alpha)\mid W'\right]\\
    =&\int \frac{ \psi^*_W(\hat\pi(x), W')}{\bE[ \psi^*_W(\hat\pi(X), W')]} \left(\One\{y\in \hat C_{rand}^*(x)\}-(1-\alpha)\right)dP_{X, Y}(x,y)\\
    =&\frac{\bE\left[ \psi^*_W(\hat\pi(X_{n+1}), W')\cdot \left(\One\{Y_{n+1}\in \hat C_{rand}^*(X_{n+1})\}-(1-\alpha)\right)\right]}{\bE[\psi^*(X, x')]}
    \\
    =&\frac{-2\hat \lambda\bE\left[\sum_{i\in[n+1]}\hat \upsilon_{S^{rand},i}/\hat \lambda\cdot   \psi^*_W(\hat\pi(X_i), W')\right]}{\bE[ \psi^*_W(\hat\pi(X), W')]}\quad\text{ by Lemma \ref{theorem: oracle mixture}}
    \\
    =&\frac{-2\bE\left[\sum_{i\in[n+1]}\hat \upsilon_{S^{rand}, i} \psi^*_W(\hat\pi(X_i), W')\right]}{\bE[ \psi^*_W(\hat\pi(X), W')]}\quad\text{ by the structure of $\hat g_{S^{rand}, \psi^*}$}
\end{align*}

\subsection{Proof of Theorem \ref{cor: mixture coverage}}\label{appsec:proof_mixture_corollary}

By the definitions in Theorem \ref{cor: mixture coverage}, for all $i\in [n]$
\begin{align*}
    &W_i\sim P_{W};\\
    &X_i\mid W_i\sim P_{X\mid W};\\
    &Y_{i}\mid X_{i}\sim P_{Y\mid X}.
\end{align*}
In this procedure, we say $Y$ is conditionally independent of $W$ given $X$. In practice, the latent variables $\{W_{i}\}_{i\in[n+1]}$ are unobserved. Firstly, for the joint distribution, we have $\{(W_i, X_i, Y_i)\}_{i\in[n]}$ independent of $(W_{n+1}, X_{n+1}, Y_{n+1})$.
In this framework, we consider the covariate shifts such that the tilt function $f(X)=\One\{\arg\max_{k'\in[K]} \pi_{k'}(X)=k\}$ for a fixed $k$. Therefore,
\begin{align*}
    &X_{n+1}, W_{n+1}\sim \frac{f(x)}{\mathbb{E}[f(X)]}P_{(X, W)};\\
    &Y_{n+1}\mid X_{n+1}\sim P_{Y\mid X}.
\end{align*}
This gives that
\begin{align*}
  X_{n+1}\sim \int_{W} \frac{f(x)}{\mathbb{E}[f(X)]}P_{X\mid W}P(W)dW&=\frac{f(x)}{\mathbb{E}[f(X)]}\int_{W} P_{X\mid W}P(W)dW\\
  &=\frac{f(x)}{\mathbb{E}_X[f(X)]} dP_X(x):= dP_f(x)
\end{align*}
Under this setting, we have for any set $C$ under the distribution $dP_f(x)$
\begin{align*}
    &\bE\left[\One\{Y_{n+1}\in C(X_{n+1})\}-(1-\alpha)\right]\\
    =&\int \left(\One\{Y_{n+1}\in C(X_{n+1})\}-(1-\alpha)\right)\frac{f(X_{n+1})}{\mathbb{E}[f(X)]} d P_X P_{Y\mid X}\\
    =&\frac{\mathbb{E}[f(X_{n+1})\left(\One\{Y_{n+1}\in C(X_{n+1})\}-(1-\alpha)\right)]}{\mathbb{E}[f(X)]}
\end{align*}
{   Note that in the oracle setting, the prediction set is constructed directly using the true embedding $\pi(\cdot)$. In this case, the RKHS class $\mathcal{F}^*$ used in our quantile regression is defined over the latent space induced by $\pi$, and we can therefore apply Lemma \ref{theorem: oracle mixture} with such a $\mathcal{ F}^*$. } By the Lemma \ref{theorem: oracle mixture}, we see the numerator equals zero since the function $f$ selected does not depend on the RKHS part.

Therefore, we have
\begin{align*}
    &\frac{\mathbb{E}[f(X_{n+1})\left(\One\{Y_{n+1}\in C(X_{n+1})\right)]}{\mathbb{E}[f(X)]}\\=&\frac{\bP(Y_{n+1}\in C(X_{n+1}
    ), T(X_{n+1})=k)}{\bP( T(X_{n+1})=k)}
    \\
    =&\bP(Y_{n+1}\in C(X_{n+1}
    )\mid \arg\max_{k'\in[K]}\ \pi_{k'}(X_{n+1})=k)=1-\alpha
\end{align*}
{
Next, we establish group-conditional coverage when groups are defined via the estimated low-rank embedding $\hat\pi(\cdot)$.

\begin{corollary}\label{cor:learn_group_coverage}
Fix $K \ge 2$ and let the latent mixture weights $\{W_i \in \Delta^{K-1}\}_{i=1}^n$ satisfy $W_i \overset{i.i.d.}{\sim} P_W$, with observations $\{X_i \mid W_i\}_{i=1}^n \overset{i.i.d.}{\sim} P_{X\mid W}$. Suppose we have an estimated embedding $\hat\pi:\mathcal{X}\to\mathbb{R}^K$ for $\pi(\cdot)$ which is defined in Theorem~\ref{cor: mixture coverage}, and define
$
\hat T(X) := \arg\max_{k\in[K]} \hat\pi_k(X)$.
Assume Assumptions~\ref{ass: data} and~\ref{ass: low-rank projection} hold, and that $\mathbb{P}(\hat T(X)=k) > 0$ for all $k\in[K]$. Let $\hat C^{*}_{\mathrm{rand}}(\cdot)$ be the randomized conformal set calibrated using the linear feature map
\begin{align*}
\Phi^*(X)
=
\big(\One\{\hat T(X)=1\},\dots,\One\{\hat T(X)=K\}\big)^\top.
\end{align*}
Then, for every $k\in[K]$,
\begin{equation}\label{eq:learned_mixture-coverage}
\mathbb{P}\!\left(
Y_{n+1}\in \hat C^{*}_{\mathrm{rand}}(X_{n+1})
\;\middle|\;
\hat T(X_{n+1})=k
\right)
= 1-\alpha.
\end{equation}
\end{corollary}
\begin{proof}
    The proof follows the same argument as Theorem~\ref{cor: mixture coverage}. The only difference is that the tilting function is now taken to be
\begin{align*}
f(x) \;=\; \One\{\hat T(x)=k\}, \qquad\text{for a fixed }k\in[K],
\end{align*}
so the reweighted distribution corresponds to conditioning on the estimated dominant group $\hat T(X)=k$. All other steps remain identical.
Since the RKHS class \(\mathcal{F}^*\) is now defined with respect to the estimated embedding \(\hat\pi(\cdot)\), applying Lemma~\ref{theorem: oracle mixture} with this tilting function yields the desired group-conditional guarantee.
To connect representation-conditional guarantee in \eqref{eq:learned_mixture-coverage} to \eqref{eq:mixture-coverage}, we apply the law of total probability such that, for any set $\hat C(X)$,
\begin{align*}
&\mathbb P\big(Y\in\hat C(X)\ \big|\ T(X)\big)
\\
=&\sum_{k=1}^K\mathbb P\left(Y\in\hat C(X), \hat T(X)=k \big|\ T(X)\right)\\
=&\sum_{k=1}^K\mathbb P\left(Y\in\hat C(X)\big| \hat T(X)=k,\ T(X)\right)\cdot \mathbb P\left(\hat T(X)=k\big| T(X)\right)
\\
=&\sum_{k=1}^K \mathbb P\big(Y\in \hat{C}(X) \big|\hat T(X)=k\big) \mathbb P\big(\hat T(X)=k\big|T(X)\big)\nonumber\\
& +\sum_{k=1}^K \big[\mathbb P\big(Y\in \hat{C}(X) \big|T(X),\hat T(X)=k\big)-\mathbb P\big(Y\in \hat{C}(X) \big|\hat T(X)=k\big)\big]\cdot\mathbb P\big(\hat T(X)=k\big| T(X)\big).
\end{align*}
The term $\mathbb P\big(Y\in \hat{C}(X) \big|T(X),\hat T(X)=k\big)\equiv\mathbb P\big(Y\in \hat{C}(X) \big|\hat T(X)=k\big)$ for all $k$ if the conformity score is sufficient for $T$ given $\hat T$, i.e., $S\perp T(X)\mid \hat T(X)$. In this case, representation-conditional coverage transfers to true-group coverage. However, the condition, $S\perp T(X)\mid \hat T(X)$, does not generally hold, since the score $S$ depends on $(X, Y)$ and $X$ depends on the unobserved variable $W$ beyond $\hat T$. It does hold under the ideal alignment $\hat T(X)=T(X)$ a.s., which is the setting in Theorem \ref{cor: mixture coverage}. The condition $\hat T(X)=T(X)$ a.s. holds when $\hat T(X)$ is essentially the Bayes-optimal argmax classifier or a marginal condition holds as shown in Lemma \ref{lem:top1-margin}.

\end{proof}
\begin{remark}\label{remark: tilted function}
While the RKHS class $\mathcal{F}^\ast$ is specified using the estimated low-rank embedding $\hat{\pi}(\cdot)$, the tilt function $f=\One\{\hat T(X)=k\}$ used for group-conditional coverage must be defined symmetrically with respect to the calibration sample and test point. Since $\hat\pi(\cdot)$ is a deterministic function of the unordered inputs $\{X_i\}_{i=1}^{n+1}$, the induced partition $\hat T(\cdot)$ is invariant to re-orderings of these inputs. Thus $f$ is data-dependent but permutation-invariant, which preserves the exchangeability conditions for our finite-sample guarantee.

This construction can be viewed as an empirical proxy for a population-level tilt determined by the latent structure. The posterior mean embedding $\pi(X)=\mathbb{E}[W\mid X]$ serves as the ideal population target, as it is a deterministic, stable summary of the latent variable $W$, thus yielding the most efficient prediction sets under the settings in Theorem \ref{cor: mixture coverage}. Our estimated embedding $\hat \pi(\cdot)$ and groups $\hat T(\cdot)$ approximate this structure from data. From the Corollary \ref{cor:learn_group_coverage} and the adaptivity of the tilt function $f$ to $\hat\pi(\cdot)$, the coverage guarantee remains robust to the errors in $\hat \pi(\cdot)$.
\end{remark}}

\subsection{Some Technical Proofs}
\begin{lemma}\label{lem:psi-expansion}
Fix $K\geq 2$ and let $\pi_i=(\pi_{ik} )_{k=1}^K$ be the true embedding representative and $\hat\pi_i=(\hat\pi_{ik})_{k=1}^K$ with $\pi_{ik} ,\hat\pi_{ik}\in(0,1)$. Define
\[
\theta _{ik}:=\log \pi_{ik} -\frac{1}{K}\sum_{\ell=1}^K\log \pi_{i\ell} ,\qquad
\hat\theta_{ik}:=\log \hat\pi_{ik}-\frac{1}{K}\sum_{\ell=1}^K\log \hat\pi_{i\ell},
\]
and write vectors $\theta_i =(\theta _{ik})_{k=1}^K$, $\hat\theta_i=(\hat\theta_{ik})_{k=1}^K$. Let $r_{ik}:=\pi_{ik} -\hat\pi_{ik}$ and $\Delta\pi_{ik}:= r_{ik}/\hat\pi_{ik}$, and define the centered vector
\[
\tilde\Delta\pi_{ik}:=\Delta\pi_{ik}-\frac{1}{K}\sum_{\ell=1}^K \Delta\pi_{i\ell}\quad (k=1,\dots,K),\qquad
\tilde \Delta\pi_i:=(\tilde \Delta\pi_{ik})_{k=1}^K.
\]
Let ${\Delta}_{1,ij}=2\langle \hat\theta_i-\hat\theta_j,\,\tilde \Delta\pi_i-\tilde \Delta\pi_j\rangle
+\|\tilde \Delta\pi_i-\tilde \Delta\pi_j\|_2^2$. Assume $\max_{i,k} |\Delta\pi_{ik}|\le \tfrac12$, then
\begin{align}
\|\theta_i -\theta_j \|_2^2
-\|\hat\theta_i-\hat\theta_j\|_2^2
=\Delta_{1,ij}
+{\Delta}_{2,ij}, \label{eq:dist-expand-noeta}
\end{align}
for an absolute constant $C$ with $|{\Delta}_{2,ij}|
\le C\!\left(\Big(\max_k |\Delta\pi_{ik}|\Big)^3+\Big(\max_k |\Delta\pi_{jk}|\Big)^3\right)$. For the Gaussian kernel $\psi _{ij}:=\exp\!\big(-\gamma\|\theta_i -\theta_j \|_2^2\big)$ and $
\hat\psi_{ij}:=\exp\!\big(-\gamma\|\hat\theta_i-\hat\theta_j\|_2^2\big)$,
we have
\begin{equation}\label{eq:psi-expand-noeta}
\psi _{ij}
=\hat\psi_{ij}\!\left(1
-\gamma\Delta_{1,ij}
+O\!\left(|\Delta_{2,ij}|
+\gamma^2 \Delta_{1,ij}^2\right)\right).
\end{equation}

\end{lemma}

\begin{proof}
Write $\pi_{ik} =\hat\pi_{ik}(1+\Delta\pi_{ik})$. Then
\[
\theta _{ik}-\hat\theta_{ik}
=\log(1+\Delta\pi_{ik})-\frac{1}{K}\sum_{\ell=1}^K\log(1+\Delta\pi_{i\ell}).
\]
For $|u|\le \tfrac12$, $\log(1+u)=u-\tfrac12 u^2+r(u)$ with $|r(u)|\le 2|u|^3$. Hence
\[
\theta _{ik}-\hat\theta_{ik}
=\tilde \Delta\pi_{ik}-\tfrac12\!\left(\Delta\pi_{ik}^2-\frac{1}{K}\sum_{\ell=1}^K \Delta\pi_{i\ell}^2\right)+\tilde r_{ik},
\qquad
|\tilde r_{ik}|\le 2\!\left(\max_k |\Delta\pi_{ik}|\right)^{\!3}.
\]
Let $q_i:=\theta_i -\hat\theta_i-\tilde \Delta\pi_i$ where $q_i$ collects the centered quadratic and remainder terms; then $\|q_i\|_2\lesssim (\max_k|\Delta\pi_{ik}|)^2$. Consequently,
\[
\theta_i -\theta_j =(\hat\theta_i-\hat\theta_j)+(\tilde \Delta\pi_i-\tilde \Delta\pi_j)+(q_i-q_j),
\]
and expanding the squared norm yields \eqref{eq:dist-expand-noeta} with
\[
\Delta_{2,ij}=2\langle \hat\theta_i-\hat\theta_j,\,q_i-q_j\rangle
+2\langle \tilde \Delta\pi_i-\tilde \Delta\pi_j,\,q_i-q_j\rangle+\|q_i-q_j\|_2^2,
\]
which is bounded as stated by Cauchy–Schwarz and the displayed bounds on $q_i,q_j$.

For the kernels, write with $\hat dist_{ij}:=\|\hat\theta_i-\hat\theta_j\|_2^2$ and $\Delta dist_{ij}:=\|\theta_i -\theta_j \|_2^2-\hat dist_{ij}$,
\[
\psi _{ij}=\hat\psi_{ij}\exp(-\gamma \Delta dist_{ij})
=\hat\psi_{ij}\Big(1-\gamma \Delta dist_{ij}+O(\gamma^2\Delta dist_{ij}^2)\Big),
\]
and substitute \eqref{eq:dist-expand-noeta} to obtain \eqref{eq:psi-expand-noeta}.
\end{proof}

\begin{lemma}[Kernel perturbation against a fixed value]\label{lem:psi-w}
Fix $K\ge2$. Let $\Delta\pi_{ik}:=(\pi_{ik} -\hat\pi_{ik})/\hat\pi_{ik}$ and the centered version $\tilde \Delta\pi_{ik}$ as in Lemma \ref{lem:psi-expansion}. Fix any $w\in\mathbb{R}^K$ and define the Gaussian kernels
\[
\psi _i(w):=\exp\!\big(-\gamma\|\theta_i -w\|_2^2\big),\qquad
\hat\psi_i(w):=\exp\!\big(-\gamma\|\hat\theta_i-w\|_2^2\big).
\]
Assume $\max_k|\Delta\pi_{ik}|\le \tfrac12$. Then, writing \(
\Delta_{1,i}(w):=2\big\langle \hat\theta_i-w,\ \tilde\Delta\pi_i\big\rangle
+\|\tilde\Delta\pi_i\|_2^2,\)
we have the distance expansion
\begin{equation}\label{eq:dist-w}
\|\theta_i -w\|_2^2-\|\hat\theta_i-w\|_2^2
=\Delta_{1,i}(w)+\Delta_{2,i},\qquad
|\Delta_{2,i}|\le C\Big(\max_k|\Delta\pi_{ik}|\Big)^3,
\end{equation}
for an absolute constant $C$. Consequently,
\begin{equation}\label{eq:abs-diff-w}
\Delta_i(w):=\big|\psi _i(w)-\hat\psi_i(w)\big|
\;\le\;
\hat\psi_i(w)\left(\gamma\,\big| \Delta_{1,i}(w)\big|
+ C\big((\max_k|\Delta\pi_{ik}|)^3+\gamma^2 \Delta_{1,i}(w)^2\big)\right).
\end{equation}
\end{lemma}

\begin{proof}
Write $\pi_{ik} =\hat\pi_{ik}(1+\Delta\pi_{ik})$.  Using the proof in Lemma \ref{lem:psi-expansion}, we have now
\[
\|\theta_i -w\|_2^2-\|\hat\theta_i-w\|_2^2
=2\langle \hat\theta_i-w,\tilde\Delta\pi_i\rangle+\|\tilde\Delta\pi_i\|_2^2+\Delta_{2,i},
\]
with
\[
\Delta_{2,i}
=2\langle \hat\theta_i-w,q_i\rangle
+2\langle \tilde\Delta\pi_i,q_i\rangle+\|q_i\|_2^2,
\]
which obeys $|\Delta_{2,i}|\le C(\max_k|\Delta\pi_{ik}|)^3$ by Cauchy–Schwarz and the bounds on $q_i$.
This proves \eqref{eq:dist-w}. For the kernels, let $\hat dist_i(w):=\|\hat\theta_i-w\|_2^2$ and $\Delta dist_i(w):=\|\theta_i -w\|_2^2-\hat dist_i(w)$.
Then
\[
\psi _i(w)=\exp\big(-\gamma(\hat dist_i(w)+\Delta dist_i(w))\big)
=\hat\psi_i(w)\Big(1-\gamma\Delta dist_i(w)+O(\gamma^2\Delta dist_i(w)^2)\Big).
\]
Substitute \eqref{eq:dist-w} for $\Delta dist_i(w)$ and take absolute values to obtain \eqref{eq:abs-diff-w}.
\end{proof}

\begin{lemma}\label{lem:top1-margin}
Let $m(X):=\pi_{(1)}(X)-\pi_{(2)}(X)$ be the (pointwise) top-1 margin, where $\pi_{(1)}\ge\pi_{(2)}\ge\cdots$ are the order statistics of $\{\pi_k(X)\}_{k=1}^K$. If
\[
\|\hat\pi(X)-\pi(X)\|_\infty \;<\; \tfrac12\,m(X)\quad\text{a.s.,}
\]
then $\hat T(X)=T(X)$ a.s.
\end{lemma}

\begin{proof}
Let $k:=T(X)$, so $\pi_k(X)-\pi_\ell(X)\ge m(X)$ for all $\ell\neq k$. Then
\[
\hat\pi_k-\hat\pi_\ell
= (\pi_k-\pi_\ell) + (\hat\pi_k-\pi_k) - (\hat\pi_\ell-\pi_\ell)
\ge m(X) - 2\|\hat\pi-\pi\|_\infty > 0,
\]
so $\hat T(X)=k$.
\end{proof}

\subsubsection{Approximate Conditional Validity Under Embedding Error}\label{appsec: embedding error local}

\begin{lemma}\label{lem:cond-valid-pi-star}
Let $W'$ be drawn according to the true neighborhood kernel, $W'\mid \pi(X_{n+1})\sim \psi^*_W(\pi(X_{n+1}),\cdot)$. Assume the conditions in Lemma \ref{lem:psi-w} are all satisfied, then
\begin{equation}\label{eq:target-bound}
 \bP(Y_{n+1}\in\hat C^{*}_{rand}(X_{n+1})\mid W'=w')=1-\alpha-\frac{2\bE[\sum_{i\in[n+1]}\hat \upsilon_{S^{rand}, i} \psi^*(w', \hat\pi(X_i))]}{\mathbb{E}[\psi^*(w', \hat\pi(X))]}+Err(w').
\end{equation}
where \begin{equation}\label{eq:ratio-perturb-general}
|\mathrm{Err}(w')|
\;\le\;
\frac{\;\bE\!\big[\Delta_{i}(w')\big]}{\bE[\psi^*_W(\pi(X),w')]}
\;+\;
\frac{\bE\!\big[\psi^*_W(\hat\pi(X),w')\big]}{\bE[\psi^*_W(\pi(X),w')]}\cdot
\frac{\;\bE\!\big[\Delta_{i}(w')\big]}{\bE[\psi^*_W(\hat\pi(X),w')]}\,,
\end{equation}
with
\(\Delta_{i}(w')\;=\;\Big|\psi^*_W(\hat\pi(X_i),w')-\psi^*_W(\pi(X_i),w')\Big|\)
.
\end{lemma}

\begin{proof}
Starting from the displayed decomposition in Theorem \ref{cor: local coverage},
\begin{align*}
&\bE\!\left[\One\{Y_{n+1}\in \hat C_{rand}^*(X_{n+1})\}-(1-\alpha)\mid W'\right]\\
=&\frac{\bE\!\left[ \psi^*_W(\pi(X_{n+1}), W')\cdot \big(\One\{Y_{n+1}\in \hat C_{rand}^*(X_{n+1})\}-(1-\alpha)\big)\right]}{\bE\!\left[\psi^*_W(\pi(X), W')\right]}.
\end{align*}
If we replace the true $\pi(X_i)$ by the estimated $\hat\pi(X_i)$, define $
N(W'):=\bE\!\left[ \psi^*_W(\pi(X), W')\cdot Z(X,Y)\right]$, $
D(W'):=\bE\!\left[ \psi^*_W(\pi(X), W')\right]$, and $
\hat N(W'):=\bE\!\left[ \psi^*_W(\hat\pi(X), W')\cdot Z(X,Y)\right]$, $
\hat D(W'):=\bE\!\left[ \psi^*_W(\hat\pi(X), W')\right]$
with $Z(X,Y):=\One\{Y\in \hat C_{rand}^*(X)\}-(1-\alpha)\in[-1,1]$.
A standard ratio perturbation yields
\[
\Big|\frac{\hat N(W')}{\hat D(W')}-\frac{N(W')}{D(W')}\Big|
\;\le\;
\frac{|\hat N(W')-N(W')|}{D(W')}+\frac{|\hat N(W')|}{\hat D(W')}\cdot\frac{|\hat D(W')-D(W')|}{D(W')},
\]
since $D(W'),\hat D(W')>0$. Next, with
$\Delta_i(W')\;:=\;\Big|\psi^*_W(\hat\pi(X_i),W')-\psi^*_W(\pi(X_i),W')\Big|$ and $
\Delta(W')\;:=\;\bE\!\big[\Delta_X(W')\big],$ we have
\[
|\hat N(W')-N(W')|
=\Big|\bE\big[(\psi^*_W(\hat\pi(X),W')-\psi^*_W(\pi(X),W'))\,Z(X,Y)\big]\Big|
\le \bE\big[\Delta_X(W')\big]=\Delta(W'),
\]
and similarly $|\hat D(W')-D(W')|\le \Delta(W')$.
Using $|\hat N(W')|\le \hat D(W')$ (because $|Z|\le 1$) gives
\begin{align*}
    &\bE\!\left[\One\{Y_{n+1}\in \hat C_{rand}^*(X_{n+1})\}-(1-\alpha)\mid W'\right]\\
=&\frac{\bE\!\left[ \psi^*_W(\hat\pi(X_{n+1}), W')\cdot \big(\One\{Y_{n+1}\in \hat C_{rand}^*(X_{n+1})\}-(1-\alpha)\big)\right]}{\bE\!\left[\psi^*_W(\hat \pi(X), W')\right]}+Err(W')\\
=&\frac{-2\bE\left[\sum_{i\in[n+1]}\hat \upsilon_{S^{rand}, i} \psi^*_W(\hat\pi(X_i), W')\right]}{\bE[ \psi^*_W(\hat\pi(X), W')]}+Err(W'),\quad\text{ Lemma \ref{theorem: oracle mixture}}
\end{align*}
where the general bound is in \eqref{eq:ratio-perturb-general}. If $\Delta_{i}(w')\to 0$, then $Err(w')\to 0$ as well. Therefore, \eqref{eq: corollary 2 neighborhood} closely approximates the conditional guarantee with respect to the true latent representation.
\end{proof}

\subsubsection{Coverage Gap Estimation}\label{appsec: coverage gap}
The idea behind estimating the coverage gap $2\hat\lambda \frac{\bE\left[\langle \hat g_{rand, \psi^*}, f_{\psi^*}\rangle_{\psi^*}\right]}{\bE[f(X)]}$ is to leverage results from $n$-sample quantile regression, applied specifically to the calibration data points. As shown in Proposition 2 of \citet{gibbs2023conformal}, the estimation error in their setting (using raw covariates) can be bounded by $O(\sqrt{d\log n/n})$. We adapt their arguments to the latent-space setting, where the feature map satisfies $\|\Phi^*(X)\|_1=1$.  The following result, Proposition \ref{proposition: estimation error in local reweighting}, provides a sharper bound on this estimation error under our setting.

To simplify the notation, let
\begin{align*}
    &\cL_n(g_{\psi^*}, \eta):=\frac{1}{n}\sum_{i\in[n]} \ell_{\alpha}(S_i-\Phi^*(X_i)^\top \eta-g_{\psi^*}(X_i))\\
    &\cL_\infty(g_{\psi^*}, \eta):=\bE\left[ \ell_{\alpha}(S_i-\Phi^*(X_i)^\top \eta-g_{\psi^*}(X_i))\right]
\end{align*}
denote the empirical and population losses with low-rank projection $\hat\pi(\cdot)$.

Recall the closed form solution in \eqref{eq:closed_form} shows the estimated coefficients are functions of $\lambda$. For a fixed $\lambda$, we denote the solution class parameterized by $\lambda$ as
\begin{align}\label{eq: g form}
\cF_{\lambda,\psi^*}=\{g_{\psi^*}: g_{\psi^*}(x)=\frac{1}{\lambda}\sum_{i\in[n+1]}\upsilon_{i}\psi^*(x, X_i)\}
\end{align}
Define the objective
\begin{align*}
    &\tilde \cL_n(g_{\psi^*}, \eta):= \cL_n(g_{\psi^*}, \eta)+\lambda\cdot\|g_{\psi^*}\|_{\psi^*}^2\\
    &\tilde \cL_\infty(g_{\psi^*}, \eta):= \cL_\infty(g_{\psi^*}, \eta)+\lambda\cdot\|g_{\psi^*}\|_{\psi^*}^2
\end{align*}
which is strictly convex in $g_{\psi^*}$ and $\eta$. Let $(\hat g_{n, \psi^*},\mathcal{B}_n), (  g^*_{\infty, \psi^*}, \mathcal{B}^*_\infty)\in \cF_{\lambda, \psi^*}\times 2^{\bR^{K}},$ denote the minimizers of $\min_{(g_{\psi^*}, \eta)\in \cF^*}\tilde \cL_n(g_{\psi^*}, \eta),  \min_{(g_{\psi^*}, \eta)\in \cF^*}\tilde \cL_\infty(g_{\psi^*}, \eta)$, respectively.

Note we write $g(x)=\Phi^*(x)^\top\eta+g_{\psi^*}(x)$ with arbitrary $( g_{\psi^*}, \eta)$. Let
\begin{align*}
    & g^*_{\infty}(x)=\Phi^*(x)^\top \eta^*_{\infty}+g^*_{\infty, \psi^*}(x)\text{ for } \eta^*_{\infty}\in \mathcal{B}^*_{\infty}\\
    &\hat g_{n}(x)=\Phi^*(x)^\top\hat \eta_{n}+\hat g_{n, \psi^*}(x) \text{ for } \hat \eta_{n}\in \mathcal{B}_n
\end{align*}

  Let $e(g, g^*_{\infty})=\mathcal{\tilde L}_{\infty}(g_{\psi^*}, \eta)-\mathcal{\tilde L}_{\infty}( g^*_{\infty, \psi^*}, \mathcal{P}_{\mathcal{B}^*_\infty}\eta)$.
\begin{assumption}[Population strong convexity]\label{ass:strong convexity} Let $d(g_{\psi^*}, \eta):=\inf_{\eta^*_{\infty}\in \mathcal{B}^*_{\infty}}\|\eta-\eta^*_{\infty}\|_2+\|g_{\psi^*}-g^*_{\infty, \psi^*}\|_{\psi^*}$ denote the distance from $(g_{\psi^*}, \eta)$ to the nearest population minimizer. Suppose $S\mid X$ has positive density on $\bR$ and is continuous.
If $d(g_{\psi^*}, \eta)\leq \epsilon_l$ for some constant $\epsilon_l>0$, then there exists some constant $C_l>0$ such that
\begin{align*}
 e(g, g^*_{\infty})\geq C_l d(g_{\psi^*}, \eta)^2
\end{align*}
\end{assumption}
This assumption is mild under some assumptions on the distribution of $S\mid X$ since $\nabla^2_{\eta}\cL_{\infty}=\bE[P_{S\mid X}(0) XX^\top]$ \citep{tan2022communication}.
\begin{assumption}\label{ass: other assumptions}
    There exist some constants $c_f, c_{\pi}, c_{f, S}>0$ such that
    \begin{align*}
    &\sup_{f\in \cF^*}\sqrt{\bE[|f(X_i)|^2]}\leq c_f \bE\left[|f(X)|\right],\quad \bE[|f(X_i)|S_i^2]\leq c_{f, S}\bE\left[|f(X)|\right]\\&\inf_{\eta:\|\eta\|_2=1, \eta\in \bR^{d}} \bE[|\Phi^*(X)^\top\eta|]\geq c_{\pi}, \bE[\|\Phi^*(X_i)\|_2^2]\leq c_0, \\&\sup_{f\in \cF^*}\bE[|f(X_i)|\|\Phi^*(X_i)\|^2_2]\leq c_1\bE[|f(X_i)|].
    \end{align*}

    Furthermore, we assume that $\bE[|S_i^2|]<\infty$ and $\sup_x\psi^*(x, x)=1$.
\end{assumption}
This assumption is stronger than Assumption 1 in \citet{gibbs2023conformal}, which requires the following moment bounds:
\begin{align*}
 \bE[\|\Phi^*(X_i)\|_2^2]\leq c_0d, \sup_{f\in \cF^*}\bE[|f(X_i)|\|\Phi^*(X_i)\|^2_2]\leq c_1\bE[|f(X_i)|]d
\end{align*}
In contrast, we assume a bounded-norm feature map in the latent space, specifically $\|\Phi^*(X)\|_2^2\leq c_0$, which does not grow with feature dimension $d$. In particular, this holds when $\Phi^*(X)$ is an indicator vector over a finite partition, in which case $\|\Phi^*(X)\|_1=1$ as well.

\begin{proposition}\label{proposition: estimation error in local reweighting} Suppose the assumptions \ref{ass:strong convexity}, \ref{ass: other assumptions} are satisfied. Under the settings in Lemma \ref{theorem: oracle mixture}. Define the $n$-sample kernel quantile regression estimate with a fixed $\lambda$
\begin{align*}
    (\hat g_{n,\psi^*}, \hat \eta_{n})=\arg\min_{g_{\psi^*}\in \cF_{\lambda, \psi^*}, \eta\in \bR^{K}}\frac{1}{n}\sum_{i\in[n]} \ell_\alpha(S_i-g_{\psi^*}(X_i)-\Phi^*(X_i)^\top\eta)+ \lambda\|g_{\psi^*}\|_{\psi^*}^2,
\end{align*}
    and for any $\epsilon>0$, let $$\cF^*_{\epsilon}:=\{f(\cdot)=f_{\psi^*}(\cdot)+\Phi^*(\cdot)^\top\eta\in \cF^*:\|f_{\psi^*}\|_{\psi^*}+\|\eta\|_2\leq 1, \bE[|f(X)|]\geq \epsilon\}.$$
    Then,
    \begin{align*}
        \sup_{f\in \cF^*_{\epsilon}}\left|2\lambda\frac{\langle \hat g_{n,\psi^*}, f_{\psi^*}\rangle_{\psi^*}}{\frac{1}{n}\sum_{i=1}^n|f(X_i)|}-2\lambda\frac{\mathbb{E}[\langle \hat g_{S_{n+1},\psi^*}, f_{\psi^*}\rangle_{\psi^*}]}{\mathbb{E}[\frac{1}{n}\sum_{i=1}^n|f(X_i)|]}\right|\leq O_{\bP}\left(\sqrt{\frac{\log n}{n}}\right)
    \end{align*}
\end{proposition}

\begin{proof}
By the results in Section 4.1.2 in \citet{boucheron2005theory} and $\|\Phi^*(X)\|_2^2\leq c_0$, we know that $\{f_{\psi^*}(\cdot)+\Phi^*(\cdot)\eta:\|f_{\psi^*}\|_{\psi^*}+\|\eta\|_2\leq 1\}$ has Rademacher complexity at most $O(\sqrt{1/n})$.
    Following the proof for Proposition 2 in \citet{gibbs2023conformal}, we can show
    \begin{enumerate}
    \item Let $\mathcal{E}_{2}=\{\|\eta-\mathcal{P}_{\mathcal{B}^*_{\infty}}\eta\|_2\leq \epsilon_1, \|g_{\psi^*}-g^*_{\infty,\psi^*}\|_{\psi^*}\leq \epsilon_2: \epsilon_1, \epsilon_2>0\}$. We have \begin{gather*}
        \bE\left\{\sup_{\eta, g_{\psi^*}\in \mathcal{E}_2}|\cL_{n} (g_{\psi^*}, \eta)-\cL_{n} (g^*_{\infty,\psi^*}, \mathcal{P}_{\mathcal{B}^*_{\infty}}\eta)-\left(\cL_{\infty} (g_{\psi^*}, \eta)-\cL_{\infty} (g^*_{\infty,\psi^*},\mathcal{P}_{\mathcal{B}^*_{\infty}}\eta)\right)|\right\}\\
        \leq O((\epsilon_1+\epsilon_2)\sqrt{1/n})
    \end{gather*}
        \item $\sup_{f\in \cF^*_\epsilon}\left|\frac{1}{n}\sum_{i\in[n]}f(X_i)-\bE[\frac{1}{n}\sum_{i\in[n]}|f(X_i)|]\right|=O_{\bP}(\sqrt{1/n})$
        \item $\sup_{f_{\psi^*}\in \cF_{\lambda, \psi^*}}\lambda\left|\bE[\langle \hat g_{S_{n+1}, \psi^*}, f_{\psi^*}\rangle_{\psi^*}]\right|=O(1)$
        \item $\sup_{f_{\psi^*}\in \cF_{\lambda, \psi^*}:\|f_{\psi^*}\|_{\psi^*}\leq 1}\lambda\left|\langle \hat g_{n, \psi^*}, f_{\psi^*}\rangle_{\psi^*}-\bE[\langle \hat g_{S_{n+1}, \psi^*}, f_{\psi^*}\rangle_{\psi^*}]\right|\leq O_{\bP}(\sqrt{\frac{\log (n)}{n}})$
    \end{enumerate}
    Using the claims above, we thus get the desired results through some calculations.
    \end{proof}
{
    \begin{remark}
        Under the setting in Theorem \ref{cor: local coverage}, the tilt function $f^{w'}(x)=\psi^*_W(\hat\pi(x), w')$ emphasizes coverage in a neighborhood around the fixed point $w'$ in the latent space. As shown in Proposition \ref{proposition: estimation error in local reweighting}, this coverage gap $\frac{2\bE[\sum_{i\in[n+1]}\hat \upsilon_{S^{rand}, i} \psi^*_W(W', \hat\pi(X_i))]}{\mathbb{E}[\psi^*_W( W', \hat\pi(X))]}$ admits a data-driven approximation $\frac{2\sum_{i\in[n]}\hat \upsilon_{n, i} \psi^*_W(w', \hat\pi(X_i))}{\frac{1}{n}\sum_{i\in[n]}\psi^*_W( w', \hat\pi(X_i))}$ where $W'=w'$ is fixed and $\{\hat \upsilon_{n,i}\}_{i\in[n]}$ are the empirical coefficients from  $\hat g_{n, \psi^*}$.
    \end{remark}}

\begin{lemma}[Regularization-controlled coverage gap]
\label{lem:gap_order_sqrt_lambda}
Consider the setting of Theorem \ref{cor: local coverage} in \emph{SpeedCP} with $\Phi^*(\cdot)\equiv 0$. Let $\mathcal W$ be the latent space.
Fix $w'\in\mathcal W$ and define the (population) kernel density
\[
p_\gamma(w') \;:=\; \mathbb E\big[\psi_W^*(w',\hat\pi(X))\big].
\]
Let $\widehat C_{\mathrm{rand}}^*$ be the randomized prediction set produced by SpeedCP
and let $\hat\upsilon_i$ denote the dual coefficients returned by the RKHS quantile regression so that $\hat g(\cdot)=\frac{1}{\lambda}\sum_{i=1}^{n+1}\hat\upsilon_i \psi^*(\cdot,X_i)$.

Assume $\sup_{w\in\mathcal W}\psi_W^*(w,w) =1$, $\mathbb E[|S_i|] =O(1)$ and $p_\gamma(w')>0$.
Then the coverage gap in Theorem \ref{cor: local coverage}, $
\mathrm{Gap}_{n+1}(w')
\;:=\;
\frac{
2\mathbb E\!\left[\sum_{i=1}^{n+1}\hat\upsilon_i\,\psi_W^*(w',\hat\pi(X_i))\right]
}{
\mathbb E\!\left[\psi_W^*(w',\hat\pi(X))\right]
},$
satisfies the bound
\begin{align*}
|\mathrm{Gap}_{n+1}(w')|
\;\le\;
\frac{2\sqrt{2}}{p_\gamma(w')}\,\sqrt{\lambda}.
\end{align*}
Consequently, on any subset $\mathcal R\subset\mathcal W$ with $\inf_{w'\in\mathcal R}p_\gamma(w')\ge c>0$,
\[
\sup_{w'\in\mathcal R}|\mathrm{Gap}_{n+1}(w')|
\;\le\;
\frac{2\sqrt{2}}{c}\,\sqrt{\lambda}
\;=\;O(\sqrt{\lambda}).
\]
\end{lemma}

\begin{proof}
Fix $w'\in\mathcal W$ and define the function $f_{w'}(\cdot)$ in the RKHS $\mathcal F_W$ over the latent space
associated with $\psi_W^*$. As $
\|f_{w'}\|_{\psi_W^*}^2 \leq \sup_{w'\in\mathcal W}\psi_W^*(w',w') =1,$
we have $
\|f_{w'}\|_{\psi_W^*} =1$.
Let the fitted RKHS component on the latent space be $
\hat g_W(w)\;:=\;\frac{1}{\lambda}\sum_{i=1}^{n+1}\hat\upsilon_i\,\psi_W^*(w,\hat\pi(X_i))$. We have
\[
\sum_{i=1}^{n+1}\hat\upsilon_i\,\psi_W^*(w',\hat\pi(X_i))
\;=\;
\lambda\,\hat g_W(w')
\;=\;
\lambda\,\langle \hat g_W,\,f_{w'}\rangle_{\psi_W^*}.
\]
Applying Cauchy--Schwarz yields
\begin{equation*}
\left|\sum_{i=1}^{n+1}\hat\upsilon_i\,\psi_W^*(w',\hat\pi(X_i))\right|
\;\le\;
\lambda\,\|\hat g_W\|_{\psi_W^*}\,\|f_{w'}\|_{\psi_W^*}
\;\le\;
\lambda\,\|\hat g_W\|_{\psi_W^*}.
\end{equation*}

Next, $\hat g_W$ minimizes the penalized pinball objective used in SpeedCP, so comparing its objective value
to $g\equiv 0$ implies that
\begin{equation*}
\frac{\lambda}{2}\|\hat g_W\|_{\psi_W^*}^2
\;\le\;
\frac{1}{n+1}\left(\sum_{i=1}^n |S_i|\right).
\end{equation*}
Taking expectations and using $\mathbb{E} [|S_i|]=O(1)$,
\begin{align*}
\mathbb E\|\hat g_W\|_{\psi_W^*}^2
\;\le\;
\frac{2}{\lambda},\quad\text{ and   }
\mathbb E\|\hat g_W\|_{\psi_W^*}
\;\le\;
\sqrt{\mathbb E\|\hat g_W\|_{\psi_W^*}^2}
\;\le\;
\sqrt{\frac{2}{\lambda}}.
\end{align*}

Combining this $\|f_{w'}\|_{\psi_W^*} =1$ and taking expectations gives
\begin{align*}
\left|\mathbb E\!\left[\sum_{i=1}^{n+1}\hat\upsilon_i\,\psi_W^*(w',\hat\pi(X_i))\right]\right|
\;\le\;
\mathbb E\left|\sum_{i=1}^{n+1}\hat\upsilon_i\,\psi_W^*(w',\hat\pi(X_i))\right|
\;\le\;
\lambda\,\mathbb E\|\hat g_W\|_{\psi_W^*}
\;\le\;
\sqrt{2}\,\sqrt{\lambda}.
\end{align*}
Finally multiply by 2 and divide by $p_\gamma(w')=\mathbb E[\psi_W^*(w',\hat\pi(X))]>0$ to conclude.
\end{proof}

\paragraph{Interpretation and Dependence on $\lambda$.}
Lemma~\ref{lem:gap_order_sqrt_lambda} shows a worst-case regularization control:
for fixed localization (fixed $\gamma$), the gap decreases at least as $O(\sqrt{\lambda})$
on any latent region where $p_\gamma(w')$ is bounded below.
This is a purely regularization-driven bound and does not require smoothness assumptions on the conditional score law.
The denominator $p_\gamma(w')=\mathbb E[\psi_W^*(w',\hat\pi(X))]$ is a kernel-smoothed density.
If $w'$ lies in a low-density region of the latent space, $p_\gamma(w')$ can be arbitrarily small, making any ratio bound
necessarily weak. Thus one typically states uniform bounds on a subset $\mathcal R$ where $\inf_{w'\in\mathcal R}p_\gamma(w')\ge c>0$,
which corresponds to restricting attention to dense latent regions.

Finally, letting $\lambda=\lambda_n\to 0$ as $n\to\infty$ makes the regularization-induced gap vanish on such dense regions,
and therefore yields asymptotic conditional coverage {under localization}.
However, taking $\lambda\to 0$ without additional structure can reintroduce the classical impossibility phenomenon:
in a fully distribution-free setting, enforcing exact conditional coverage may force prediction sets to become overly conservative
and uninformative. In contrast, the asymptotic conditional coverage guarantees in localized conformal prediction as in \citet{guan2023localized} are obtained under explicit distributional regularity conditions together with a shrinking localization radius
$h_n\to 0$ (equivalently $\gamma_n\to\infty$), and rates such as $n h_n^\beta/\log n\to\infty$ to ensure sufficient local sample mass;
these assumptions guarantee that both localization bias and stochastic error vanish asymptotically.

\section{Additional Experiments}
\subsection{Synthetic Experiments}\label{appendix:sim}

In this section, we provide additional details on the synthetic experiments and provide further experiment results.

In all of our experiments, we generate $X_i\in \mathbb{R}^{p}$ from a mixture of $K=3$ latent distributions. Specifically, we first generate $X_i$ from a multinomial distribution, $mX_i \sim \mathrm{Multinomial}(N, \sum_{k\in [K]}W_i(k)\zeta_k)$ with $W_i=w_i$ fixed and total count $N=1000$. For each sample in the training and calibration sets, we generate $W_i\sim \mathrm{Dir}([2,1,1])$ and randomly shuffle the elements to create a distribution that is more symmetric across vertices. Here, the density is higher in the central part of the simplex. For test samples, we generate from the same distribution but do not shuffle, to create a high concentration near one vertex of the 2-dimensional simplex (Figure~\ref{fig:density}).

We sample the latent component $\zeta_k \in \mathbb{R}^{p}$ from a uniform distribution and normalize it so that $\sum_{j\in[p]}\zeta_{k}(j)=1 $ for each $k \in [K]$.
We estimate $\pi(X_i)=\mathbb{E}[W_i\mid X_i]$ with pLSI (Section~\ref{sec: topic modeling}) and use $\hat{\pi}(X_i)$ as inputs of SpeedCP, CondCP, PCP, and RLCP. For SpeedCP and CondCP, we choose $\Phi^*(X_i)=(1, \One\{\arg\max_{k}\hat \pi_k(X_i)=1 \},\dots,\One\{\arg\max_{k}\hat \pi_k(X_i)=K-1 \})^\top$ using the estimated latent embeddings $\hat{\pi}(X_i)$. The response is generated as $Y_i\sim N(\sin(2\pi\cdot W_i(1))+ (W_i(2))^2+W_i^\top\eta, 0.1^2)$, where $\eta_j\sim \mathrm{Unif}(1,10)$ for $j=1,2,3$ and is normalized.

In this setting, we aim to see whether each conformal method can guarantee 0.9 coverage uniformly across the simplex, especially in boundaries (areas close to one vertex). We report our results based on 50 independent runs of data generation. At each run, we split the data into 600 training points, 300 calibration points, and 100 test points.

{
In Table~\ref{tab:predictors}, we report the computation time for two different predictors. We can see that SpeedCP is faster compared to CondCP and PCP, which are the state-of-the-art conformal prediction methods that account for the local or latent data structure. RLCP is fast but fails to attain target miscoverage level as discussed in Section~\ref{sec: experiments} of the main manuscript. In Figure~\ref{fig:nn_simplex}, we show the coverage conditional on the latent space of $\hat{\pi}(X)$ when the predictor is a neural network. The same plot for the linear regression predictor is shown in Figure~\ref{fig:ternary}. In both plots, we observe that SpeedCP achieves 0.9 across the simplex most uniformly.

\begin{figure}[ht]
    \centering
   \includegraphics[width=0.7\textwidth]{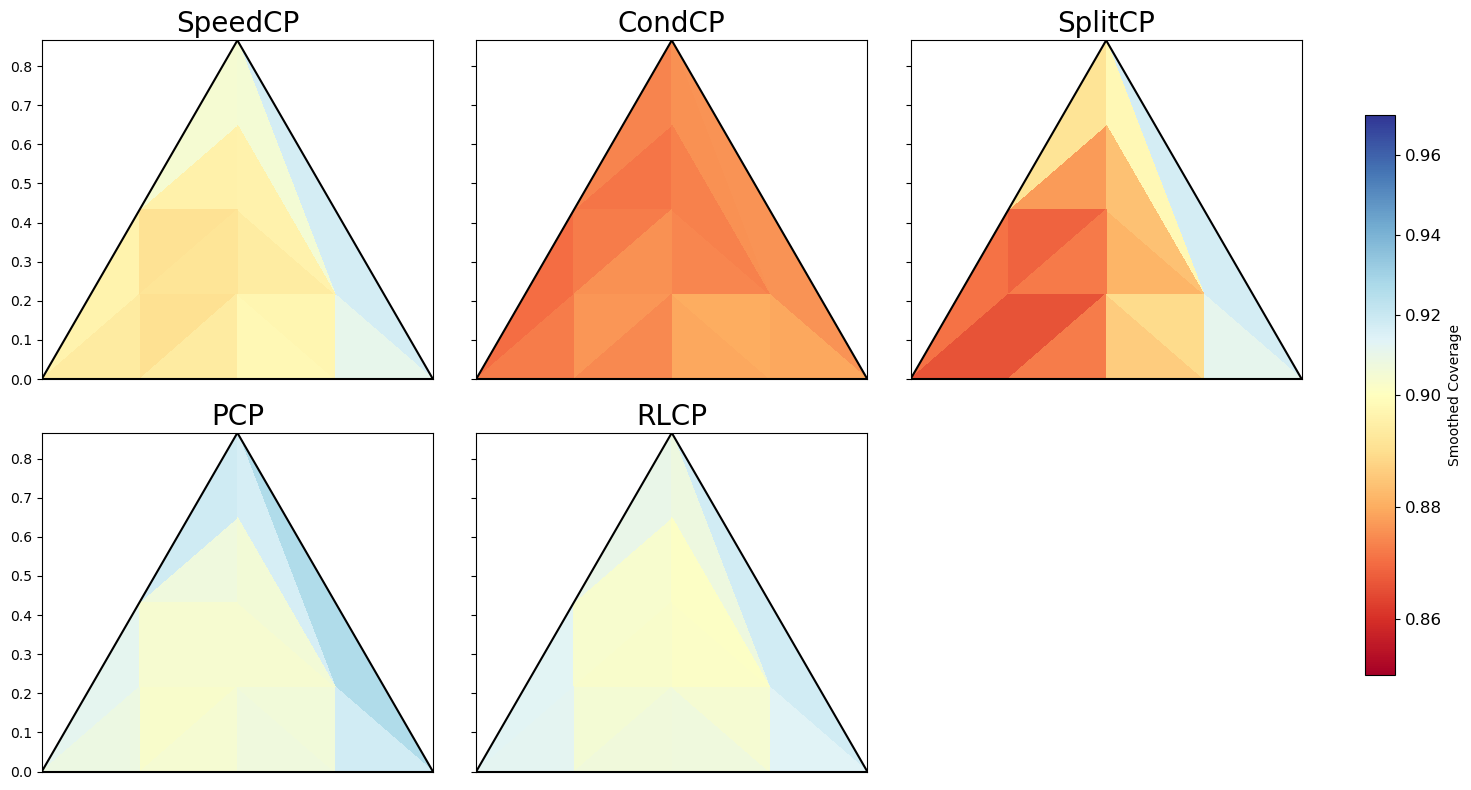}
    \caption{Mean coverage on fine-gridded partitions on the latent space (a 2D simplex) when $\hat{\mu}=\mathrm{NN}$. The results are aggregated over 50 random generations. SpeedCP shows the most uniform 0.9 (pale yellow) coverage across the simplex.}
    \label{fig:nn_simplex}
\end{figure}
}
\begin{figure}[ht]
    \centering
   \includegraphics[width=0.7\textwidth]{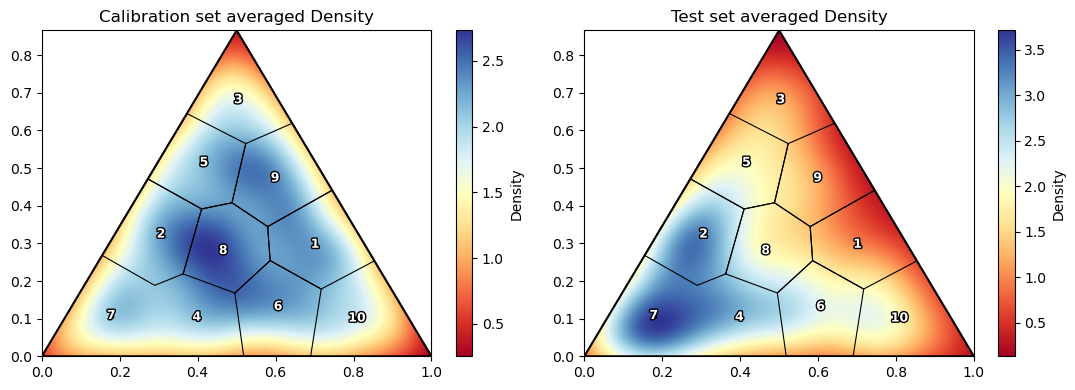}
    \caption{Averaged calibration and test density over 50 random generations of data. We use kmeans followed by Voronoi tessellation to partition the latent simplex into 10 bins.}
    \label{fig:density}
\end{figure}

{
\subsubsection{Coverage Across Different Sample Size \texorpdfstring{$n$}{n}}

We experiment with different sample sizes $n$ in Table~\ref{tab:speedcp_n}. As the calibration set grows larger, we generally expect the coverage guarantee to remain the same while the prediction set size decreases because the estimation error and the uncertainty of the kernel quantile estimator diminish. Interestingly, we observe a slight increase in prediction size at $n=2000$ followed by a decrease at $n=5000$. This fluctuation is likely due to finite-sample variability in both the estimated latent embeddings and the cross-validated hyperparameters. The computation time increases with $n$, which is consistent with the computational complexity of SpeedCP, which is approximately $O(n^3)$. Overall, for moderately large $n$, the method still remains computationally feasible.

\begin{table}[htbp]
  \centering
  \caption{Mean prediction set size and computation time for SpeedCP (linear regression predictor).}
  \label{tab:speedcp_n}
  \vspace{2mm}
  \resizebox{0.5\textwidth}{!}{
  \begin{tabular}{c c c}
    \toprule
    $n$
      & \textbf{Prediction set size}
      & \textbf{Computation time (seconds)} \\
    \midrule
    1000
      & 1.376 \,$\pm$\, 0.07
      & 12.363 \,$\pm$\, 4.10 \\
    2000
      & 2.053 \,$\pm$\, 0.11
      & 23.956 \,$\pm$\, 7.03 \\
    5000
      & 1.030 \,$\pm$\, 0.02
      & 191.255 \,$\pm$\, 29.36 \\
    \bottomrule
  \end{tabular}
  }
\end{table}

}

\subsubsection{Choices of Different \texorpdfstring{$\Phi^*(X)$}{PhiX}} We also discuss how conditional coverage changes with different choices of $\Phi^*(X)$ of our function class $\cF^*$ \eqref{eq:fstar}. When running a RKHS-based quantile regression on the scores, $\Phi^*(X)^{\top}\eta$ acts as the linear component with the design matrix $\Phi^*(X)$ and parameters $\eta$. $\Phi^*(\cdot)$ allows flexible modeling of different types of conditional coverage. For example, in this synthetic experiment, we can consider four different $\Phi^*(X)$ based on the estimated latent embedding $\hat{\pi}(X)$,

\begin{enumerate}\label{eq:phi}
    \item Taking $\Phi^*({X})=1$ yields the marginal coverage.
     \item Taking $\Phi^*({X})=\hat{\pi}(X)$ yields mixture-conditional coverage, where we guarantee coverage linearly reweighted with $\hat{\pi}(X)$.
     \item In our experiments, taking {$\Phi^*({X})=(\One, \One\{\hat T(X)=1\},\dots,\One\{\hat T(X)=K-1\})^\top $; equivalently, $\Phi^*({X})=( \One\{\hat T(X)=1\},\dots,\One\{\hat T(X)=K\})^\top $} where $\hat T(X)=\arg\max_{k\in [K]}\hat\pi_k({X})$ yields topic-conditional coverage, where the topic is defined as the latent distribution with the highest mixture proportion weight.
 \end{enumerate}

Through our experiments, we observed that in high-dimensional settings, coverage using SpeedCP is primarily affected by the RKHS component, $f_{\Psi^*}$ rather than the linear term. If more prior information is available on the conditional distribution, and the goal is to achieve more precise conditional coverage at level $1 - \alpha$, one may instead calibrate scores using a function class restricted to the linear term, as in \citet{gibbs2023conformal}. However, the inclusion of the RKHS component can lead to smaller prediction sets even without those additional prior structures. Further investigation is needed to determine whether choosing $\Phi^*(X)$ as the indicators of topics, or the latent embeddings, improves performance under varying covariate dimensionality $p$ or the signal-to-noise ratio in $X$.

\subsubsection{Sensitivity to Kernel Family}
\label{app:kernel-family-sensitivity}

The main experiments use an RBF kernel on the low-rank latent representation. We additionally evaluate whether the empirical performance of SpeedCP is sensitive to this kernel choice. This question is practically relevant because the kernel controls how the RKHS quantile regression borrows information across nearby calibration points in the latent space. The RBF kernel is very smooth and decays rapidly with distance, whereas Mat\'ern kernels allow finite smoothness controlled by a smoothness parameter, and inverse multiquadric (IMQ) kernels have heavier polynomial tails \citep{rasmussen2006gaussian,stein1999interpolation,gorham2017measuring}.

We use the same latent representation as in Appendix~\ref{sec: topic modeling}. Let
\[
    \widehat{\theta}_i = \operatorname{clr}(\widehat{\pi}(X_i)),
    \qquad
    \widehat{\theta}_{ik}
    =
    \log \widehat{\pi}_k(X_i)
    -
    \frac{1}{K}\sum_{\ell=1}^K \log \widehat{\pi}_\ell(X_i),
\]
where \(\widehat{\pi}(X_i)\in\Delta^{K-1}\) is the estimated mixture-proportion embedding and \(\operatorname{clr}(\cdot)\) is the centered log-ratio transformation for compositional data \citep{aitchison1982statistical}. Define the latent distance
\[
    r_{ij}
    :=
    \|\widehat{\theta}_i-\widehat{\theta}_j\|_2,
    \qquad
    u_{ij}
    :=
    \sqrt{\gamma}\, r_{ij},
\]
where \(\gamma>0\) is the inverse bandwidth parameter. We compare the following normalized radial kernels, each satisfying \(\psi^\ast(X_i,X_i)=1\):

\[
\begin{aligned}
\text{RBF:}\qquad
\psi^\ast_{\mathrm{RBF},\gamma}(X_i,X_j)
&=
\exp\{-u_{ij}^2\}
=
\exp\{-\gamma r_{ij}^2\},
\\[4pt]
\text{Mat\'ern:}\qquad
\psi^\ast_{\mathrm{Mat},\gamma,\nu_{\rm Mat}}(X_i,X_j)
&=
\frac{2^{1-\nu_{\rm Mat}}}{\Gamma(\nu_{\rm Mat})}
\left(\sqrt{2\nu_{\rm Mat}}\,u_{ij}\right)^{\nu_{\rm Mat}}
K_{\nu_{\rm Mat}}\!\left(\sqrt{2\nu_{\rm Mat}}\,u_{ij}\right),
\\[4pt]
\text{IMQ:}\qquad
\psi^\ast_{\mathrm{IMQ},\gamma,\beta_{\rm IMQ}}(X_i,X_j)
&=
\left(1+u_{ij}^2\right)^{-\beta_{\rm IMQ}}
=
\left(1+\gamma r_{ij}^2\right)^{-\beta_{\rm IMQ}} .
\end{aligned}
\]
Here \(K_\nu(\cdot)\) denotes the modified Bessel function of the second kind, \(\nu_{\rm Mat}>0\) is the Mat\'ern smoothness parameter, and \(\beta_{\rm IMQ}>0\) controls the polynomial tail of the IMQ kernel. The Mat\'ern expression is interpreted by continuity at \(u_{ij}=0\), where it equals one. The IMQ form above is equivalent to the common parameterization \((c^2+\|x-y\|_2^2)^\beta\) with \(\beta<0\), after normalization and reparameterization of the bandwidth.

We repeat the shifted latent-mixture simulation from Appendix~\ref{appendix:sim} using the three kernel families above. All other components are held fixed: the pLSI embedding, the centered log-ratio transformation, the topic-based linear feature map \(\Phi^\ast\), the \(\lambda\)-path and \(S\)-path algorithms, and the same hyperparameter selection procedure. For each kernel family, \(\gamma\) and \(\lambda\) are selected using the same tuning strategy as in the main experiments, while the shape parameters \(\nu_{\rm Mat}\) and \(\beta_{\rm IMQ}\) are fixed a priori and are not tuned. Since different kernel families map the same numerical value of \(\gamma\) to different effective neighborhoods, the selected \(\gamma\)'s should not be interpreted as directly comparable across families.

Table~\ref{tab:kernel-family-sensitivity} reports marginal and topic-wise coverage, score cutoffs, selected hyperparameters, and runtime. Across the three kernel families, SpeedCP maintains coverage close to the nominal target \(0.9\), both marginally and within each latent topic. The runtime is also nearly unchanged across kernels, indicating that the computational behavior is mainly governed by the path-following algorithm and the size of the elbow set, rather than by the analytic form of the radial kernel.

The main difference appears in the score cutoffs. RBF gives the smallest and most stable cutoffs. IMQ is slightly more conservative on marginal cutoff, consistent with its heavier tail and more global smoothing behavior. Mat\'ern achieves similar coverage but has larger cutoff variability, especially in Topic~2, where the cutoff is \(1.343\pm 0.996\). This suggests that, in the shifted latent-mixture experiment, validity and computation are relatively insensitive to the kernel family, whereas efficiency is affected by how locally the kernel smooths the latent space. Based on these results, we use the RBF kernel as the default choice in the main experiments.

\begin{table}[ht]
\centering
\begin{tabular}{lccc}
\toprule
 & SpeedCP-RBF & SpeedCP-Mat\'ern & SpeedCP-IMQ \\
\midrule
Marginal coverage & 0.893 & 0.884 & 0.892 \\
Topic 0 coverage  & 0.891 & 0.882 & 0.896 \\
Topic 1 coverage  & 0.901 & 0.900 & 0.903 \\
Topic 2 coverage  & 0.882 & 0.885 & 0.889 \\
\midrule
Marginal cutoff & $1.164 \pm 0.141$ & $1.225 \pm 0.406$ & $1.203 \pm 0.167$ \\
Topic 0 cutoff  & $1.099 \pm 0.150$ & $1.106 \pm 0.367$ & $1.144 \pm 0.176$ \\
Topic 1 cutoff  & $1.215 \pm 0.169$ & $1.268 \pm 0.443$ & $1.240 \pm 0.215$ \\
Topic 2 cutoff  & $1.190 \pm 0.201$ & $1.343 \pm 0.996$ & $1.238 \pm 0.331$ \\
\midrule
Selected $\gamma$ & $8.046 \pm 11.006$ & $8.612 \pm 12.433$ & $5.679 \pm 9.568$ \\
Selected $\lambda$ & $166.457 \pm 594.037$ & $149.159 \pm 584.517$ & $194.275 \pm 760.146$ \\
Runtime (s) & $61.060 \pm 24.150$ & $63.550 \pm 26.930$ & $60.960 \pm 27.600$ \\
\bottomrule
\end{tabular}
\caption{Kernel-family sensitivity of SpeedCP on the shifted latent-mixture experiment. Coverage entries are empirical averages over repeated runs. Cutoffs, selected hyperparameters, and runtime are reported as mean $\pm$ standard deviation.}
\label{tab:kernel-family-sensitivity}
\end{table}

{
\subsubsection{Using Calibration Set for Tuning \texorpdfstring{$(\gamma, \lambda)$}{gammalammbda}}\label{appendix:calib}
In our experiments, we use the calibration set for selection of $(\gamma, \lambda)$ instead of setting aside a separate validation set for efficiency. We agree that, in the current implementation, \((\gamma, \lambda)\) is selected by cross-validation using the calibration data, so the chosen pair is technically data-dependent. Therefore, this dependence can introduce some bias and small finite-sample distortion of coverage. However, the selected pair $(\lambda, \gamma)$ converges in probability to a deterministic value when the calibration set size is large. Recent work on adaptive coverage policies shows that it is empirically valid to use the calibration set itself to select regularization parameters via leave-one-out or cross-validation (Theorem 2.6 in \citep{gauthier2025adaptive}). \citet{gibbs2023conformal} also show that the selection of $\lambda$ using the calibration set does not affect coverage significantly.

To assess whether the calibration set can be reliably used for hyperparameter tuning, we compare it against a split strategy in which half of the calibration set is used for tuning and the remaining half for calibration, avoiding the potential issues discussed earlier. As shown in Figure~\ref{fig:split}, the two approaches have similar conditional coverage and prediction set sizes, with the full-calibration procedure exhibiting only a slight overestimation of coverage. In Figure~\ref{fig:grid}, we observe that the chosen $(\gamma, \lambda)$ pairs from the two approaches are also similar.

\subsubsection{Uniform Coverage on Any \texorpdfstring{$(\gamma, \lambda)$}{gammalambda}}\label{appendix: uniform_grid}
We assess whether the uniform coverage guarantee assumed in Section~\ref{sec: theory} holds across all choices of $(\gamma, \lambda)$. We obtain the approximate joint hyperparameter space by gathering results from running the $\lambda$-path on each $\gamma$ in the $\gamma$ grid. We then select 25 pairs and run the $S$-path to obtain coverage and prediction set size. We observe in Figure~\ref{fig:grid} that coverage holds uniformly at 0.9 across the pairs, affirming that coverage is not affected by the choice of $(\gamma, \lambda)$. However, we observe that the prediction set size differs by the choice of $(\gamma, \lambda)$. Our cross-validation approach chooses $(\gamma, \lambda)$ in the region where prediction set size is small.

\begin{figure}[ht]
    \centering
   \includegraphics[width=0.8\textwidth]{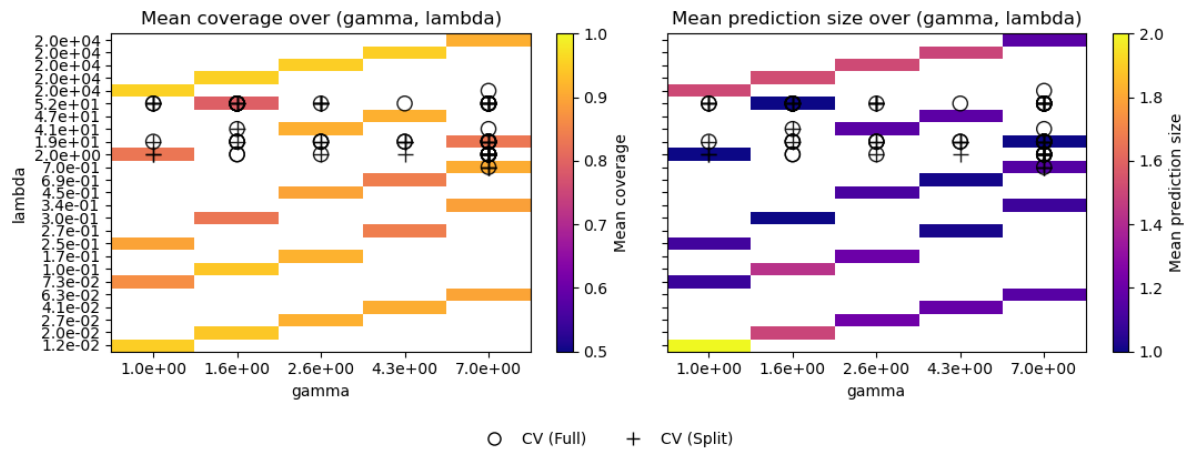}
    \caption{Marginal coverage and prediction set size of 25 pairs of $(\gamma, \lambda)$ on the joint hyperparameter space. We also show selected $(\gamma, \lambda)$'s using $k$-fold cross validation on calibration set (full) or validation set (split).}
    \label{fig:grid}
\end{figure}

\subsubsection{Effect of Randomization on \texorpdfstring{$S$}{S}}
As described in Section~\ref{sec: method}, we use the randomized cutoff $ S^{rand}(X_{n+1}) = \sup\{S\mid \hat \upsilon_{S,n+1}< U\}$, where $U\sim \mathrm{Unif}(-\alpha,1-\alpha)$, to construct prediction sets. In practice, this randomization introduces little variability. This is because the $S$-path starts with a small value, and along the path, the $S$ values are generally smaller than $\hat{g}_S$, leading to $\hat \upsilon_{S,n+1}=-\alpha$. In Figure~\ref{fig:rand}, we observe that the standard deviation of $S^{rand}_i$'s for each run (seed) is small, confirming this behavior.
}

\begin{figure}[ht]
    \centering
   \includegraphics[width=0.5\textwidth]{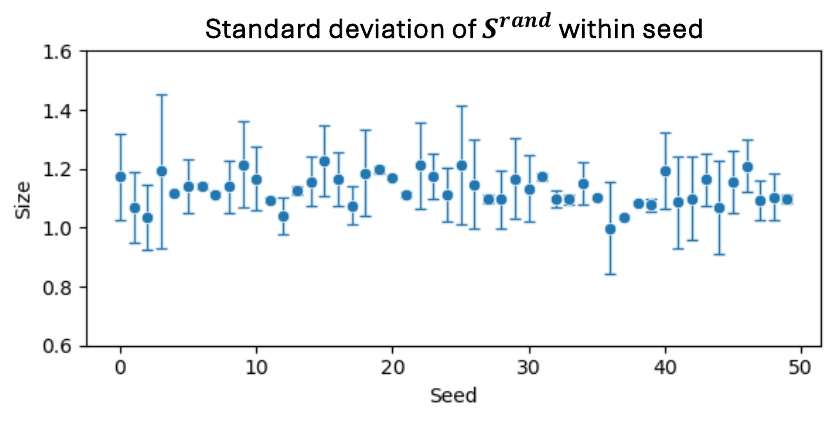}
    \caption{Mean and standard deviation of $S^{rand}_i$'s for each seed.}
    \label{fig:rand}
\end{figure}

\subsubsection{Empirical Behavior of the Elbow Set Size}
\label{app:elbow-set-size}

In particular, since the computational complexity of the $l$th step scales as $O((|E^l|+d)^3)$, the worst-case complexity can approach $O((n+d)^3)$ if $|E^l|$ becomes comparable to $n$. To better understand whether this occurs in practice, we conducted additional synthetic experiments in the latent mixture setting and empirically measured the elbow set size across different sample sizes, covariate distributions, and label noise levels.

We varied three parameters in the synthetic experiments: sample sizes $n \in \{500, 2000, 5000\}$, the train/test distribution of $\pi(X)$ through the Dirichlet parameters $\alpha \in \{[1,1,1], [2,1,1], [5,1,1]\}$, where $\alpha=[5,1,1]$ corresponds to a stronger covariate shift regime in which most test samples concentrate on the first latent cluster, and the clusterwise label noise levels, considering $\sigma=[0.05,0.05,0.05]$ for low noise, $\sigma=[0.1,0.1,0.3]$ for moderate noise, and $\sigma=[0.5,0.5,1.0]$ for high noise. All results are aggregated over $10$ independent data generations. We report both the elbow set size $|E^l|$ and the number of breakpoints along the $\lambda$- and $S$-paths.

\begin{table*}[!ht]
\vspace{-1mm}
\caption{Elbow set size and number of breakpoints across sample sizes.}
\label{tab:elbow-size-n}
\centering
\resizebox{0.55\textwidth}{!}{
\begin{tabular}{l c c c}
\toprule
                         & $n=500$         & $n=2000$          & $n=5000$          \\
\midrule
$|E^l|$ for $\lambda$-path & $4.43 \pm 1.32$ & $10.26 \pm 2.15$  & $10.82 \pm 3.08$  \\
$|E^l|$ for $S$-path       & $4.64 \pm 2.14$ & $5.62 \pm 1.46$   & $7.03 \pm 3.22$   \\
$\lambda$ breakpoints              & $23.4 \pm 9.04$ & $130.2 \pm 44.18$ & $195.0 \pm 15.0$  \\
$S$ breakpoints                    & $6.19 \pm 3.92$ & $6.62 \pm 3.66$   & $12.77 \pm 11.52$ \\
\bottomrule
\end{tabular}
}
\vspace{-1mm}
\end{table*}

\begin{table*}[!ht]
\vspace{-1mm}
\caption{Elbow set size and number of breakpoints across distributions of $\pi(X)$, with $n=2000$.}
\label{tab:elbow-size-distribution}
\centering
\resizebox{0.55\textwidth}{!}{
\begin{tabular}{l c c c}
\toprule
                         & $\alpha=[1,1,1]$         & $\alpha=[2,1,1]$          & $\alpha=[5,1,1]$          \\
\midrule
$|E^l|$ for $\lambda$-path & $7.35 \pm 0.71$    & $10.26 \pm 2.15$  & $9.05 \pm 3.55$  \\
$|E^l|$ for $S$-path       & $6.81 \pm 2.51$    & $5.62 \pm 1.46$   & $5.35 \pm 1.80$  \\
$\lambda$ breakpoints              & $112.7 \pm 31.47$  & $130.2 \pm 44.18$ & $115.6 \pm 21.44$ \\
$S$ breakpoints                    & $9.77 \pm 11.61$   & $6.62 \pm 3.66$   & $8.33 \pm 5.13$  \\
\bottomrule
\end{tabular}
}
\vspace{-1mm}
\end{table*}

\begin{table*}[!ht]
\vspace{-1mm}
\caption{Elbow set size and number of breakpoints across label noise levels, with $n=2000$.}
\label{tab:elbow-size-noise}
\centering
\resizebox{0.6\textwidth}{!}{
\begin{tabular}{l c c c}
\toprule
                         & $\sigma=[0.05,0.05,0.05]$         & $\sigma=[0.1,0.1,0.3]$          & $\sigma=[0.5,0.5,1.0]$          \\
\midrule
$|E^l|$ for $\lambda$-path & $10.48 \pm 4.04$ & $10.26 \pm 2.15$  & $6.55 \pm 0.91$  \\
$|E^l|$ for $S$-path       & $9.99 \pm 10.09$ & $5.62 \pm 1.46$   & $4.74 \pm 2.18$  \\
$\lambda$ breakpoints              & $136.3 \pm 40.81$ & $130.2 \pm 44.18$ & $63.2 \pm 21.55$ \\
$S$ breakpoints                    & $6.68 \pm 3.69$  & $6.62 \pm 3.66$   & $10.70 \pm 13.52$ \\
\bottomrule
\end{tabular}
}
\vspace{-1mm}
\end{table*}

Across all settings, the elbow set size remained small relative to $n$ for both the $\lambda$- and $S$-paths. In Table~\ref{tab:elbow-size-n}, increasing the sample size from $n=500$ to $n=5000$ increased the average elbow set size only from $4.43$ to $10.82$ along the $\lambda$-path, and from $4.64$ to $7.03$ along the $S$-path. Tables~\ref{tab:elbow-size-distribution} and~\ref{tab:elbow-size-noise} further show that the elbow set size remains stable across different levels of covariate shift and label noise. This suggests that, in these experiments, $|E^l|$ grows much more slowly than $n$ and does not appear to be highly sensitive to either covariate shift or label noise level. Therefore, the practical per-step cost $O((|E^l|+d)^3)$ is far below the worst-case $O((n+d)^3)$.

\subsection{Real Data Experiment}\label{appendix:real}

\subsubsection{ArXiv Abstracts} \label{appendix: arxiv}
We sample 5000 abstracts from ArXiv metadata \citep{clement2019arxiv} in mathematics, statistics, and computer science categories. The processed abstract-word count matrix has a vocabulary size of $p=11,516$. We project the abstracts onto $K=5$ latent mixture proportions, $\hat{\pi}(X_i)$, using pLSI, the topic modeling approach described in Section~\ref{sec: topic modeling}. We use $\hat{\pi}(X_i)$ as inputs for all methods.

For SpeedCP and CondCP, we additionally set the linear representation $\Phi^*(X_i)$ as a one-hot encoding of the topic:  { \(\Phi^*(X) = \big( \ \One\{\hat T(X)=\texttt{Geometry}\},\ \One\{\hat T(X)=\texttt{Algebra}\}, \One\{\hat T(X)=\texttt{ML}\},\ \One\{\hat T(X)=\texttt{Vision}\},\ \One\{\hat T(X)=\texttt{Quantum}\}\big)^\top\)}. Figure~\ref{fig:arxiv topics} displays the top words for each estimated topic, while Figure~\ref{fig:arxiv_prop} shows the proportion of documents in each estimated topic. At a resolution of $K=5$, the topics are readily interpretable and correspond to distinct subfields within mathematics, statistics, and computer science. pLSI estimates soft assignments $\hat{\pi}(X_i)\in\mathbb{R}^5$, representing mixture proportions over the topics, which we use as inputs to SpeedCP, CondCP, PCP, and RLCP.

The goal is to construct prediction intervals that achieve nominal level 0.9 across topics.  {CondCP is omitted because, in our experiments, it did not finish within the allotted time budget (30 hours). This occurred consistently across the larger datasets we evaluated.} We present topic-conditional coverage and prediction set size in Table~\ref{tab:coverage}.  {To illustrate performance under a poor predictor, we choose linear regression of citation counts on raw word frequencies, which fails to extract any meaningful associations between words and citation counts.} As a result, RLCP produces overly wide prediction intervals and PCP fails to uncover any latent mixture structure of the conditional score distribution $S\mid\hat{\pi}(X)$ and becomes equivalent to SplitCP. In contrast, SpeedCP leverages kernel smoothing, resulting in tighter and more accurate prediction intervals.

\begin{table}[htbp]
  \centering
  \vspace{-3mm}
  \caption{Mean coverage across topics and prediction set size of ArXiv dataset.}
  \label{tab:coverage}
  \resizebox{\textwidth}{!}{
      \begin{tabular}{l c c c c c c c c }
        \toprule
        Method
          & \multicolumn{5}{c}{\textbf{Target coverage ($1-\alpha=0.9$)}}
          & \textbf{Size} & \textbf{Time (seconds)}\\
        \cmidrule(lr){2-6}

          & Geometry & Algebra & ML & Vision & Quantum
          &   \\
        \midrule
        \textbf{SpeedCP} & \textbf{0.880\,$\pm$0.02} & \textbf{0.890\,$\pm$0.05} & 0.730\,$\pm$0.34
                  & \textbf{0.920\,$\pm$0.02} & 0.822\,$\pm$0.11
                  & 15.835\,$\pm$3.05 &  8.682\,$\pm$3.10\\
        SplitCP & 0.877\,$\pm$0.02 & 0.876\,$\pm$0.04 & 0.659\,$\pm$0.35
                  & 0.926\,$\pm$0.02 & 0.762\,$\pm$0.08
                  & 15.661\,$\pm$1.17 & $< 0.01$\\
        PCP      & 0.877\,$\pm$0.02 & 0.876\,$\pm$0.04 & 0.659\,$\pm$0.35
                  & 0.926\,$\pm$0.02 & 0.762\,$\pm$0.08
                  & 15.661\,$\pm$1.17 & 17.501\,$\pm$0.54 \\
        RLCP      & 0.935\,$\pm$0.02 & 0.958\,$\pm$0.03 & \textbf{0.956\,$\pm$0.16}
                  & 0.923\,$\pm$0.02 & \textbf{0.962\,$\pm$0.04}
                  & 42.493\,$\pm$45.308 & 1.184\,$\pm$0.01\\
        \bottomrule
      \end{tabular}
  }
\end{table}

\begin{figure}[ht]
    \centering
    \includegraphics[width=1.0\linewidth]{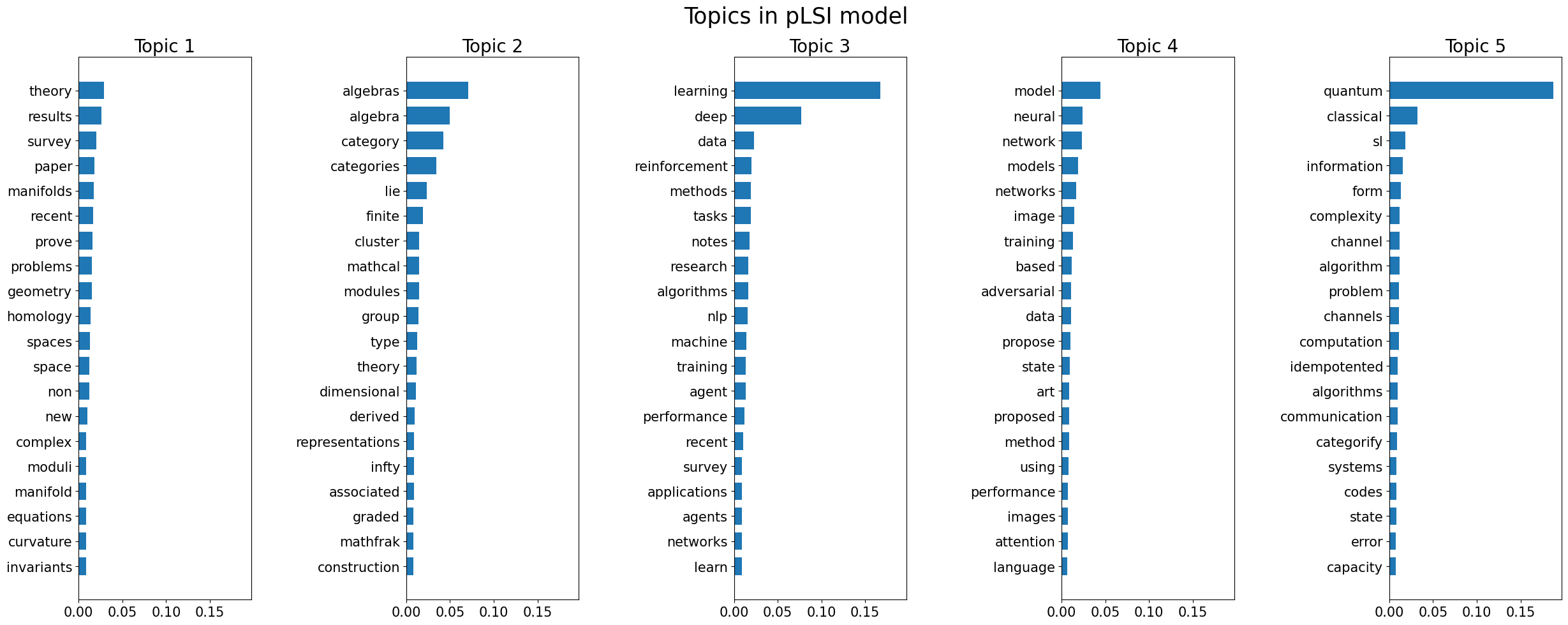}
    \caption{Latent topics of ArXiv abstracts identified by probabilistic latent semantic indexing (pLSI), a topic modeling approach. We plot the top 20 words with the largest weights for each topic. We name each topic as \textit{Geometry, Algebra, Machine Learning, Computer Vision,} and \textit{Quantum theory} based on the top words.}
    \label{fig:arxiv topics}
\end{figure}

\begin{figure}[ht]
    \centering
    \includegraphics[width=0.5\linewidth]{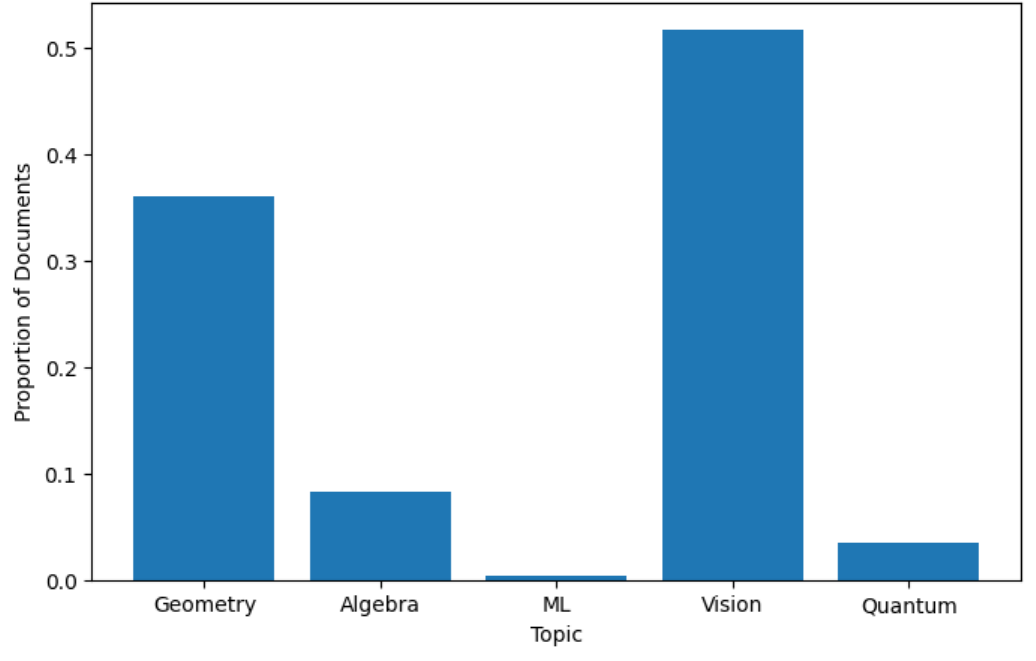}
    \caption{Distribution of the most likely topic over all abstracts with $n=5000$.}
    \label{fig:arxiv_prop}
\end{figure}

\subsubsection{Molecule Graphs}
\label{app:molecule}
We provide additional results of the molecule dataset example in Section~\ref{sec: experiments}. For each dataset, we subsample 2000 molecule graphs at each run with 50 runs in total, and split into 1000/500/500 training, calibration, and test points. Using the 1000 molecule graphs, we train a GIN predictor $\hat{\mu}(\cdot)$ to extract the 64-dimensional last layer and compute conformal scores $S_i=|\hat\mu(X_i)-y_i|$. In this experiment, we consider the intercept for the linear term, $\Phi^*(X_i)=1$ and $\pi(X_i)$ as the PC score. In Figure~\ref{fig:molecule_vor}, we plot the Voronoi partitions on which we measure the coverage (Figure~\ref{fig:molecule_coverage}) as well as the mean prediction set size (Figure~\ref{fig:molecule_cutoffs}) at level $\alpha=0.1$.

Our method, SpeedCP, and SplitCP construct the smallest prediction sets overall. However, while SplitCP applies a single global cutoff across the entire PC space, SpeedCP adapts to the local structure of the embeddings. For instance, in the QM9 dataset we find that SpeedCP produces slightly larger prediction sets in sparser regions of the PC space (e.g., partitions 2, 4, and 6), which allows it to maintain consistent $0.9$ coverage across all partitions.

\begin{figure}[ht]
    \centering
    \includegraphics[width=0.85\linewidth]{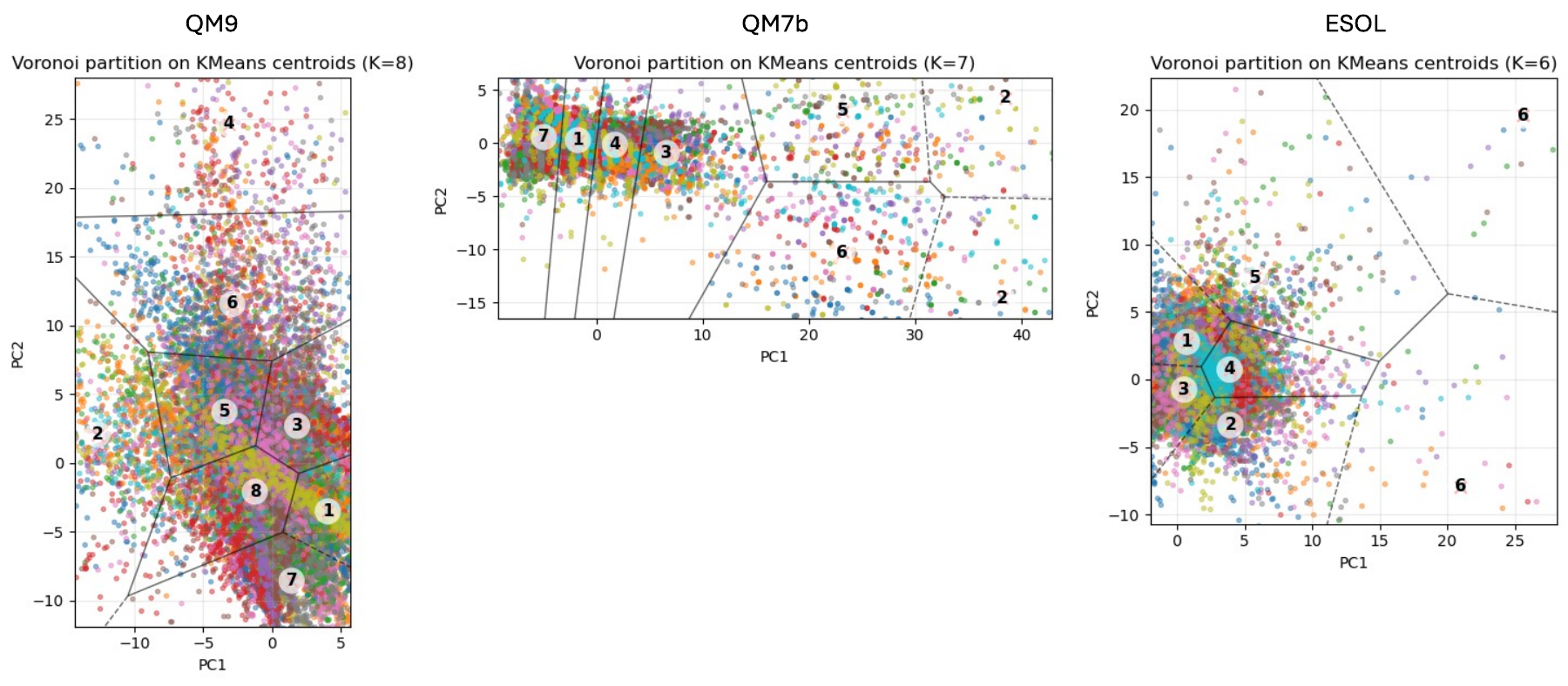}
    \caption{Voronoi tessellation of the PC space. We plot PC representations of graph embeddings, where each color denotes a random subsample of the dataset.}
    \label{fig:molecule_vor}
\end{figure}

\begin{figure}[ht]
    \centering
    \includegraphics[width=0.8\linewidth]{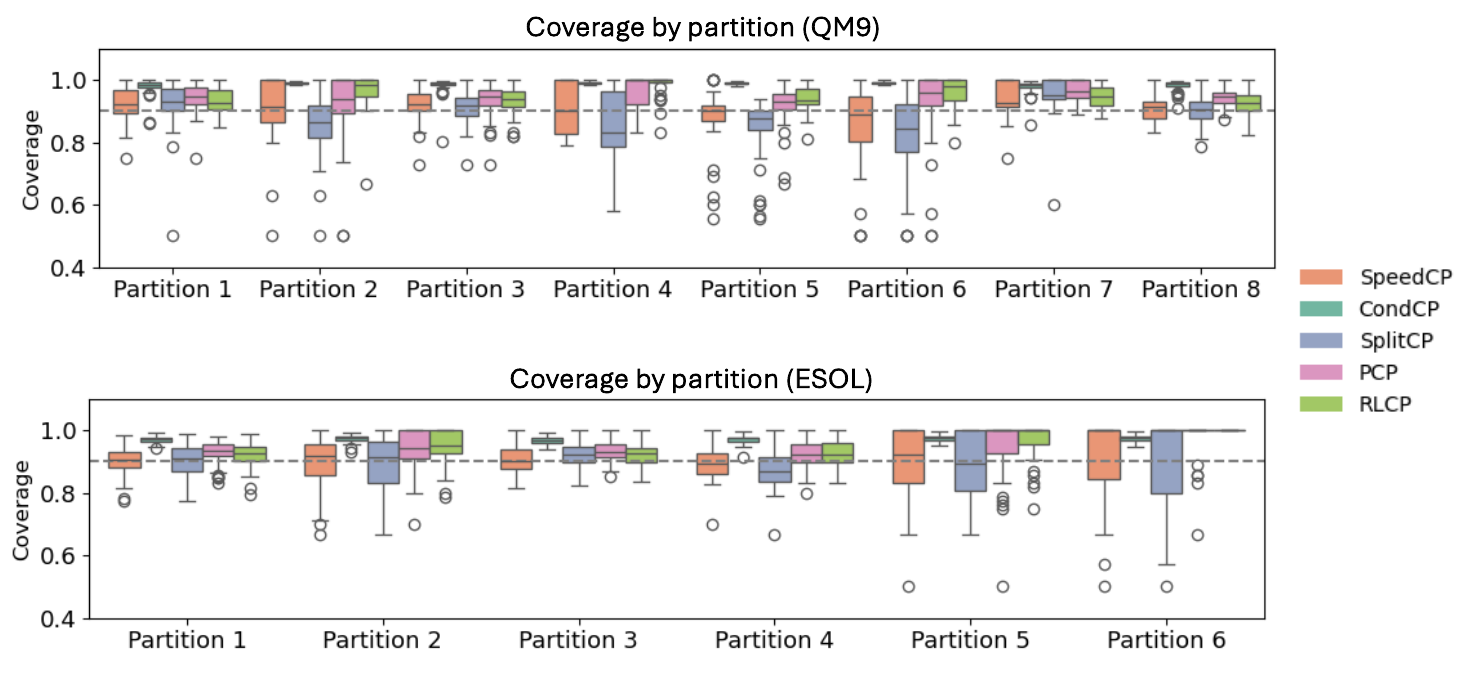}
    \caption{Coverage on fixed partitions of the PC space for QM9 and ESOL. We use PCA on the last layer embeddings of GNN with $K=3$ dimensions. The dashed line denotes the target coverage rate $1-\alpha=0.9$.}
    \label{fig:molecule_coverage}
\end{figure}

\begin{figure}[ht]
    \centering
    \includegraphics[width=0.8\linewidth]{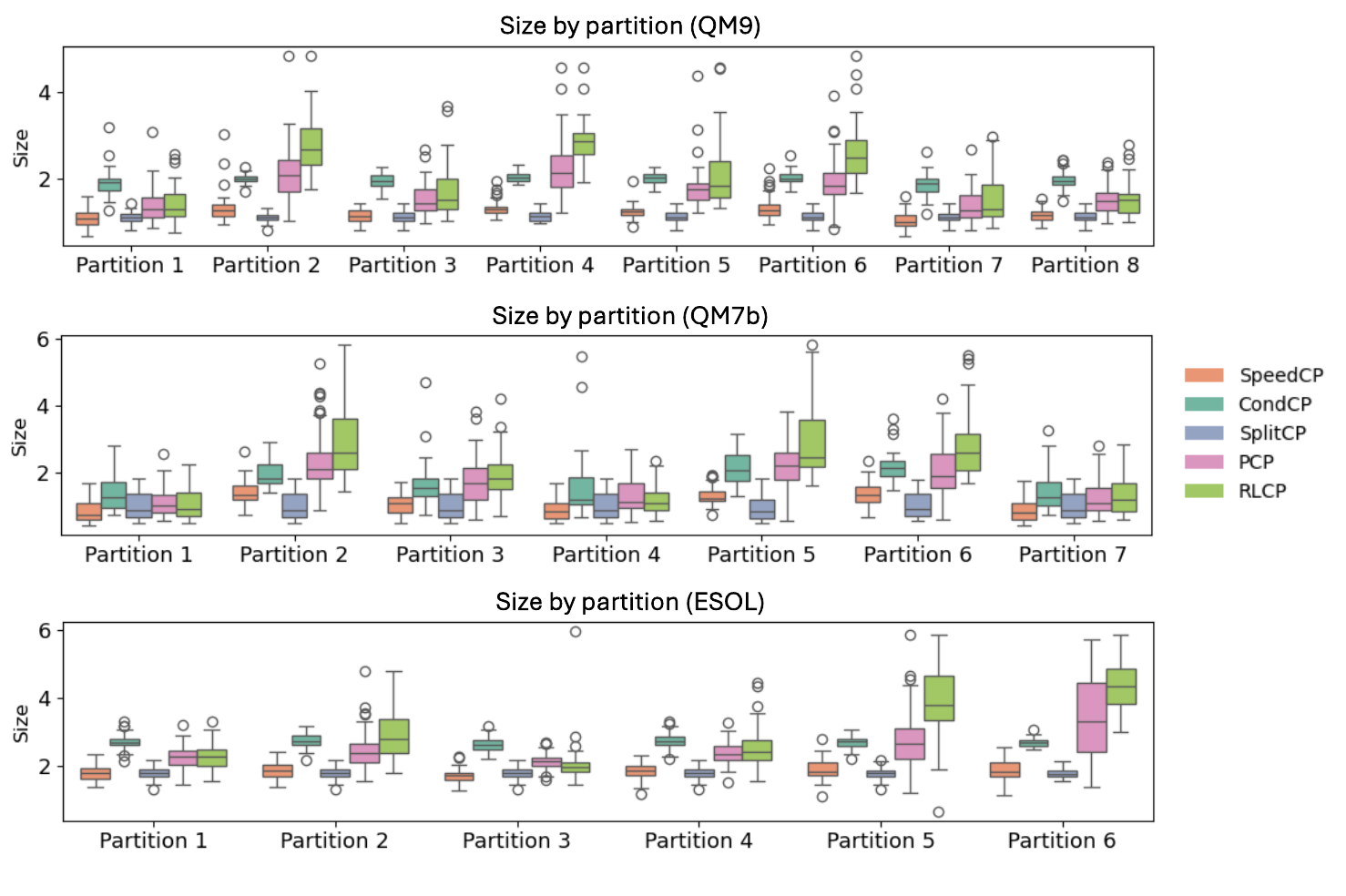}
    \caption{Prediction set size on fixed partitions of the PC space for each molecule dataset. We use PCA on the last layer embeddings of GNN with $K=3$ dimensions.}
    \label{fig:molecule_cutoffs}
\end{figure}

\subsubsection{Brain Tumor MRI}\label{appsec: brain mri}
We train a CNN classifier $\hat\mu(\cdot) $ on 2{,}000 images and extract the 256-dimensional NN features from the last layer. We report the performance of the CNN classifier \(\hat\mu(\cdot)\) in Figure~\ref{fig:mri-eval}, which shows the evaluation metrics on the training and validation sets.

\begin{figure}[htbp]
    \centering
    \includegraphics[width=0.9\linewidth]{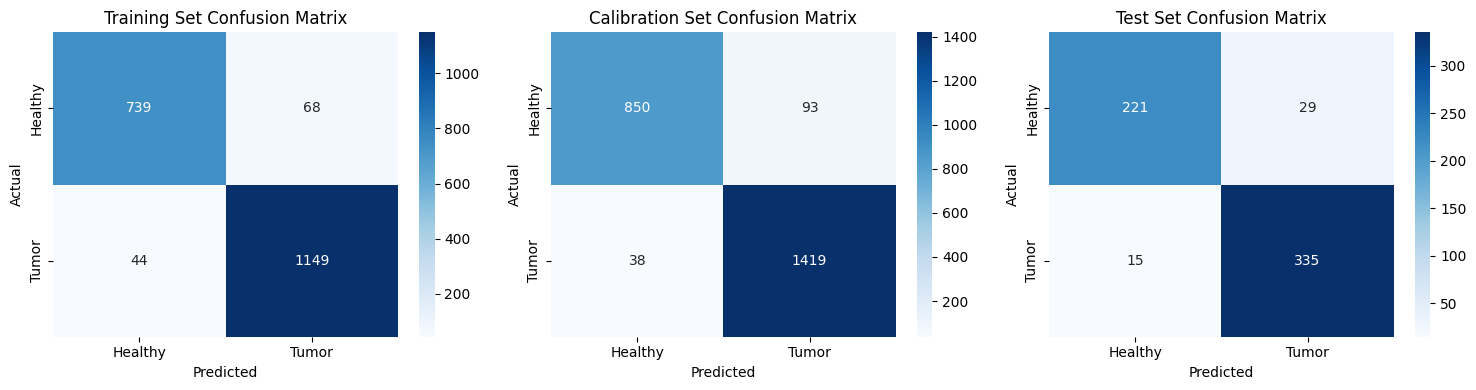}
    \caption{Evaluation of the CNN classifier on the Brain Tumor MRI dataset. }
    \label{fig:mri-eval}
\end{figure}

Using another 2{,}400 images for calibration, we compute conformal scores $S_i=|\hat\mu(X_i)-y_i|$ and apply our RKHS path-following quantile regression in the latent space to obtain thresholds at level $\alpha=0.1$.

In this experiment, we evaluate both \emph{marginal} coverage and \emph{per-label} (predicted-label) coverage
\(\mathbb{P}\!\big(Y_{n+1}\in \hat C^{*}_{\mathrm{rand}}(X_{n+1}) \mid \hat\mu(X_{n+1})=\hat y\big)\) using  600 test images over 50 simulation trials.  We exclude CondCP from the analysis because a single simulation takes over 50,000 seconds and the algorithm fails to converge. For comparison, we perform calibration using the 256-dimensional neural network features directly as the embedding \(\hat\pi(\cdot)\). To further reduce dimensionality, we apply a post hoc PCA to rank~8 on these features; the resulting principal components define \(\hat\pi:\mathcal X\to\mathbb R^{8}\).

\paragraph{Using 256-Dim Features from NN.}
We include illustrative results corresponding to Table~\ref{tab:coverage for mri} from the main paper. Empirically, the cutoffs produced by SplitCP and RLCP are effectively identical in our high-dimensional setting. Intuitively, RLCP’s locality weights become uninformative in high dimensions (the distance metric loses discriminative power), so RLCP reduces to uniform weighting over the calibration set, recovering the SplitCP cutoff.

\begin{table}[htbp]
\centering
\caption{Summary statistics of conformal cutoffs (marginal and by predicted label) using the 256-dim features from NN as input space for conformal prediction.}
\label{tab:cutoff-mri-raw}
\small
\begin{tabular}{lcccc}
\toprule
\textbf{Method} & \textbf{Mean} & \textbf{Std} & \textbf{Min} & \textbf{Max} \\
\midrule
\multicolumn{5}{l}{\textbf{Marginal}} \\
\addlinespace[2pt]
\textbf{SpeedCP$(\One)$ }& 0.2662&0.0908&0.0012&0.9985 \\
\textbf{SpeedCP$(\Phi^*)$ }& 0.2828 & 0.0820 & 0.0025 & 0.9714 \\
SplitCP     & 0.3482 & 0.0091 & 0.3271 & 0.3660 \\
RLCP      & 0.3482 & 0.0091 & 0.3271 & 0.3660 \\
PCP       & 0.2310 & 0.2899 & 0.0000 & 0.9984 \\
\midrule
\multicolumn{5}{l}{\(\boldsymbol{\hat y = \texttt{healthy}}\)} \\
\addlinespace[2pt]
\textbf{SpeedCP$(\One)$ } &0.2500&0.0954&0.0012&0.9938 \\
\textbf{SpeedCP$(\Phi^*)$ } & 0.2662 & 0.0819 & 0.0025 & 0.9533 \\
SplitCP     & 0.3482 & 0.0091 & 0.3271 & 0.3660 \\
RLCP      & 0.3482 & 0.0091 & 0.3271 & 0.3660 \\
PCP       & 0.2818 & 0.2904 & 0.0000 & 0.9984 \\
\midrule
\multicolumn{5}{l}{\(\boldsymbol{\hat y = \texttt{tumor}}\)} \\
\addlinespace[2pt]
\textbf{SpeedCP$(\One)$ } & 0.2758&0.0866&0.0963&0.9985 \\
\textbf{SpeedCP$(\Phi^*)$ } & 0.2925 & 0.0805 & 0.0952 & 0.9714 \\
SplitCP     & 0.3482 & 0.0091 & 0.3271 & 0.3660 \\
RLCP      & 0.3482 & 0.0091 & 0.3271 & 0.3660 \\
PCP       & 0.2012 & 0.2855 & 0.0000 & 0.9984 \\
\bottomrule
\end{tabular}
\end{table}

\begin{figure}[htbp]
    \centering
    \includegraphics[width=0.7\linewidth]{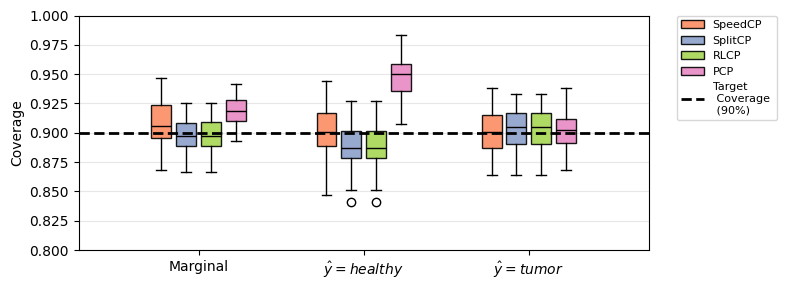}
\caption{Predicted-label conditional coverage on the Brain Tumor MRI test set under the PCA-based model. Calibration is performed using the linear feature map \(\Phi^*(X) = \big(\One\{\hat\mu(X)=\texttt{healthy}\},\ \One\{\hat\mu(X)=\texttt{tumor}\}\big)^\top\) under the 256-dim features layer from NN.}
    \label{fig:mri_raw-predlabel-coverage}
\end{figure}

\begin{figure}[htbp]
    \centering
    \includegraphics[width=0.7\linewidth]{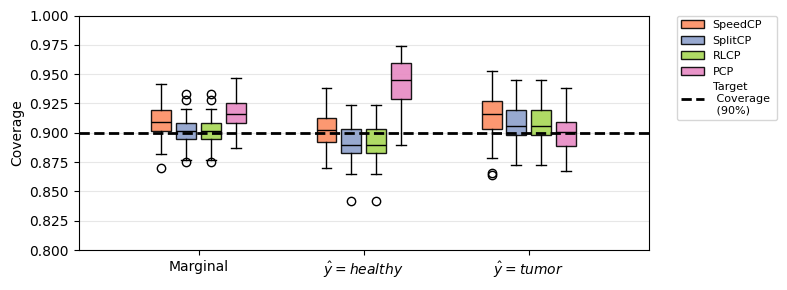}
    \caption{Predicted-label conditional coverage on the Brain Tumor MRI test set by calibrating with the intercept only $\Phi^*(X)=1$ under the 256-dim features layer from NN.}
    \label{fig:mri_raw-intercept-coverage}
\end{figure}

\paragraph{Using Principal Components.}

To further reduce dimensionality, we extract features from the neural network and project them onto a low-rank embedding via PCA with $K=8$, fitted on the first 2,000 training samples.
SplitCP attains similar coverage but requires more conservative sets in lower-dimensional space, whereas our method delivers narrower sets with near-nominal predicted-label coverage. RLCP and PCP tend to over-cover, particularly for the healthy class, and exhibit unstable cutoffs with high variance and frequent near-zero values (see Table~\ref{tab:cutoff-mri}). Consequently, even after dimensionality reduction, RLCP and PCP produce overly conservative conditional coverage.

Compared to results using higher-dimensional features, the low-rank projection further reduces the cutoff without compromising conditional guarantees (comparing Table \ref{tab:coverage for mri} with \ref{tab:coverage for mri pca}), thereby yielding narrower prediction sets.

\begin{table}[htbp]
  \centering
  \vspace{-3mm}
  \caption{Mean coverage and prediction set size across predicted labels in the MRI dataset under the PCA-based model.}
  \label{tab:coverage for mri pca}
  \resizebox{\textwidth}{!}{
      \begin{tabular}{l c c c c c c c c }
        \toprule
        Method
          & \multicolumn{3}{c}{\textbf{Target coverage ($1-\alpha=0.9$)}}
          & \multicolumn{3}{c}{\textbf{Prediction set size}} & \textbf{Time (seconds)}\\
        \cmidrule(lr){2-4} \cmidrule(lr){5-7}

          & Marginal & Healthy & Tumor & Marginal & Healthy & Tumor \\
        \midrule
        \textbf{SpeedCP($\One$)}&{0.910\,$\pm$0.01} & \textbf{0.901\,$\pm$0.02} & 0.915\,$\pm$0.01
                  & {0.239\,$\pm$0.07} & \textbf{0.230\,$\pm$0.07}
                  & 0.244\,$\pm$0.08 &  286.1\,$\pm$14.2\\
       \textbf{SpeedCP}($\Phi^*$)&{0.905\,$\pm$0.02} &\textbf{0.898\,$\pm$0.03} & \textbf{0.900\,$\pm$0.02 }
                  & {0.247\,$\pm$0.08}  & {0.241\,$\pm$0.08}
                  &{0.251\,$\pm$0.08}&  294.5\,$\pm$20.9\\
        SplitCP
        & \textbf{ 0.901\,$\pm$0.01} & 0.893\,$\pm$0.02 & 0.906\,$\pm$0.01
                  & 0.350\,$\pm$0.00 & 0.350\,$\pm$0.00
                  & 0.350\,$\pm$0.00  & $< 0.01$\\
        PCP
        & 0.906\,$\pm$0.02 & 0.925\,$\pm$0.03 & 0.895\,$\pm$0.02
                  &\textbf{ 0.230\,$\pm$0.27} & 0.279\,$\pm$0.26
                  & \textbf{0.200\,$\pm$0.26} & 130.1\,$\pm$ 28.9 \\
        RLCP
        & 0.916\,$\pm$0.01 & 0.926\,$\pm$0.02 & {0.911\,$\pm$0.02}
                  & 0.359\,$\pm$0.38 & {0.388\,$\pm$0.37}
                  & 0.342\,$\pm$0.38 & 2.095\,$\pm$0.13\\
        \bottomrule
      \end{tabular}
  }
\end{table}

\begin{figure}[htbp]
    \centering
    \includegraphics[width=0.7\linewidth]{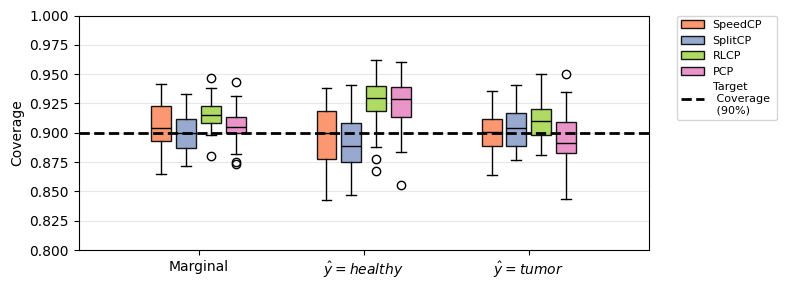}
\caption{Predicted-label conditional coverage on the Brain Tumor MRI test set under the PCA-based model. Calibration is performed using the linear feature map \(\Phi^*(X) = \big(1,\ \One\{\hat\mu(X)=\texttt{healthy}\},\ \One\{\hat\mu(X)=\texttt{tumor}\}\big)^\top\).}
    \label{fig:mri-predlabel-coverage}
\end{figure}

\begin{figure}[htbp]
    \centering
    \includegraphics[width=0.7\linewidth]{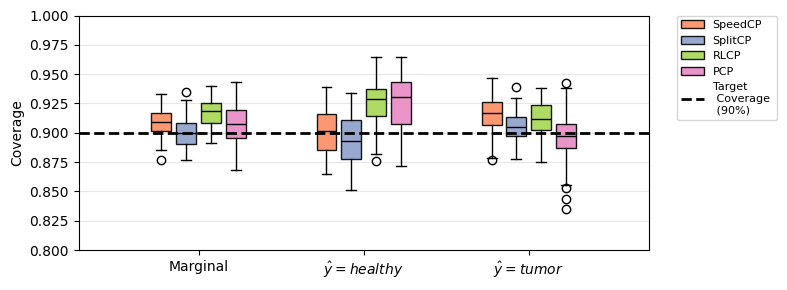}
    \caption{Predicted-label conditional coverage on the Brain Tumor MRI test set by calibrating with the intercept only $\Phi^*(X)=1$.}
    \label{fig:mri-intercept-coverage}
\end{figure}

\begin{table}[htbp]
\centering
\caption{Summary statistics of conformal cutoffs (marginal and by predicted label) using PCA-based model.  \textbf{SpeedCP$(\Phi^*)$} calibrates scores with a linear term that includes predicted labels, whereas \textbf{SpeedCP($\One$)} uses an intercept-only term.}
\label{tab:cutoff-mri}
\small
\begin{tabular}{lcccc}
\toprule
\textbf{Method} & \textbf{Mean} & \textbf{Std} & \textbf{Min} & \textbf{Max}  \\
\midrule
\multicolumn{5}{l}{\textbf{Marginal}} \\
\addlinespace[2pt]
\textbf{SpeedCP($\One$)}  & 0.2391 & 0.0738 & 0.0654 & 0.8641 \\
\textbf{SpeedCP$(\Phi^*)$ } & 0.2470 & 0.0805 &0.0442& 1.2279 \\
SplitCP      & 0.3505 & 0.0087 & 0.3315 & 0.3729  \\
RLCP       & 0.3594 & 0.3797 & 0.0000 & 0.9984  \\
PCP        & 0.2301 & 0.2672 & 0.0000 & 0.9984\\
\midrule
\multicolumn{5}{l}{\(\boldsymbol{\hat y = \texttt{healthy}}\)} \\
\addlinespace[2pt]
\textbf{SpeedCP($\One$)}  & 0.2300 & 0.0697 & 0.0654 & 0.7414 \\
\textbf{SpeedCP$(\Phi^*)$ } & 0.2409 & 0.0785&0.0442&1.2279\\
SplitCP      & 0.3506 & 0.0088 & 0.3315 & 0.3729  \\
RLCP       & 0.3883 & 0.3711 & 0.0000 & 0.9984 \\
PCP        & 0.2788 & 0.2654 & 0.0000 & 0.9984  \\
\midrule
\multicolumn{5}{l}{\(\boldsymbol{\hat y = \texttt{tumor}}\)} \\
\addlinespace[2pt]
\textbf{SpeedCP($\One$)}    & 0.2445 & 0.0756 & 0.1486 & 0.8641  \\
\textbf{SpeedCP$(\Phi^*)$ } & 0.2506&0.0815&0.0615&1.2225\\
SplitCP      & 0.3505 & 0.0087 & 0.3315 & 0.3729 \\
RLCP       & 0.3420 & 0.3838 & 0.0000 & 0.9984  \\
PCP        & 0.2009 & 0.2641 & 0.0000 & 0.9984  \\
\bottomrule
\end{tabular}
\end{table}

\subsection{Details on Computation Resources}\label{appendix:computation}
All experiments were conducted on a cloud-based computing cluster. Each job was allocated 4 CPU cores and 4 GB of memory. No GPUs were used. For CondCP, we used the MOSEK solver in CVXPY to solve the underlying convex optimization problems. All code was implemented in Python3 and run in a consistent computing environment to ensure reproducibility.

\end{document}